\documentclass[runningheads,a4paper,final]{llncs}

\usepackage{amstext,amsmath,amssymb,url}
\allowdisplaybreaks
\usepackage{newunicodechar, pifont, float} % stix
\usepackage{listings}
\makeatletter
\newcommand{\verbatimfont}[1]{\renewcommand{\verbatim@font}{\ttfamily#1}}
\makeatother

\usepackage{parskip}

\usepackage[page,titletoc]{appendix}
\makeatletter

\usepackage[show]{ed}
\usepackage{basics}
\usepackage{multirow}
\usepackage{etoolbox}

\newbool{lipics}

\usepackage{theorems}
\markEnv{example}

\usepackage[bookmarks,bookmarksopen,colorlinks,linkcolor=blue,citecolor=blue,bookmarksnumbered,urlcolor=gray]{hyperref}

\usepackage{placeins}

\usepackage{xfrac}
\usepackage{tikz}
\usepackage{array}
\usepackage{xspace}

\newunicodechar{❘}{\\&\quad}%{\ensuremath{\;\text{\ding{120}}\;}}
\newunicodechar{❙}{\\&}%{\ensuremath{\;\text{\ding{121}}\;}}
\newunicodechar{❚}{}%{\ensuremath{\;\text{\ding{122}}\;}}
\newunicodechar{Π}{\ensuremath{\ \Pi\ }}
\DeclareUnicodeCharacter{80}{\ensuremath{\ @\ }}
\DeclareUnicodeCharacter{036D}{}

\newunicodechar{⟶}{\ensuremath{\rightarrow}}
\newunicodechar{⇒}{\ensuremath{\Rightarrow}}
\newunicodechar{⇔}{\ensuremath{\Leftrightarrow}}
\newunicodechar{∧}{\ensuremath{~\land~}}
\newunicodechar{∨}{\ensuremath{~\lor~}}
\newunicodechar{≐}{\ensuremath{~\doteq~}}
\newunicodechar{⊦}{\ensuremath{\vdash}}
\newunicodechar{⦂}{\ensuremath{\typecolon}}
\newunicodechar{π}{\ensuremath{\pi}}
\newunicodechar{λ}{\ensuremath{\lambda}}

\newcommand{\hol}{HOL\xspace}
\newcommand{\dhol}{DHOL\xspace}
\newcommand{\dhole}{DHOL\xspace}

\newcommand{\lambdaFun}[2]{\lambda \varname{#1}\reserved{\ofT}#2.~}
\newcommand{\piType}[2]{\reserved{\Pi} \varname{#1}\reserved{\ofT}\typeC{#2}.~}
\newcommand{\ofT }{\ensuremath{\mathord{:}}}
\newcommand{\subtype}[2]{#1|_{#2}\hspace{0.4mm}}

\usepackage{scalerel}

\newcommand{\quot}[2]{\sfrac{\scaleobj{1.2}{#1}}{\scaleobj{1.2}{#2}}}
\newcommand{\univQuant}[2]{\ensuremath{\reserved{\forall} \,\varname{#1}\reserved{\ofT}\typeC{#2}.~}}
\newcommand{\implC}{\impl}

\newcommand{\judg}{\Gamma\dedT}

\newcommand{\existQuant}[2]{\ensuremath{\exists\,\varname{#1}\reserved{\ofT}\typeC{#2}.~}}

\newcommand{\Thy}[1]{\ensuremath{#1\ \mathsf{Thy}}}
\newcommand{\Ctx}[1]{\ensuremath{#1\ \mathsf{Ctx}}}
\newcommand{\namedax}[2]{\assert#2}%\axname{#1}\reserved{\ofT}
\newcommand{\namedass}[2]{\assertL#2}%\axname{#1}\reserved{\ofT}
\newcommand{\typingAxName}[1]{\ensuremath{#1^*}}
\newcommand{\typingAx}[2]{\namedax{\typingAxName{#2}}{\PredPhi{#1}{#2}}}
\newcommand{\typingAssName}[1]{\ensuremath{#1^*}}
\newcommand{\typingAss}[2]{\namedass{\typingAssName{#2}}{\PredPhi{#1}{#2}}}
\newcommand{\concatCtx}[2]{#1,\;#2}
\newcommand{\concatThy}[2]{#1,\;#2}
\newcommand{\emptyThy}{\circ}
\newcommand{\emptyCtx}{.}
\newcommand{\ctxIn}[2]{\ensuremath{#1\text{ in }#2}}
\newcommand{\thyIn}[2]{\ensuremath{#1\text{ in }\theorycolor{#2}}}
\newcommand{\Type}[1]{\ensuremath{#1\ \type}}
\newcommand{\defaultTerm}[1]{\ensuremath{\termC{w}_{#1}}}

\usepackage{scalerel,stackengine}
\newcommand\reallywidetilde[1]{%
	\savestack{\tmpbox}{\stretchto{%
			\scaleto{%
				\scalerel*[\widthof{\ensuremath{#1}}]{\kern-.6pt\sim\kern-.6pt}%
				{\rule[-\textheight/2]{1ex}{\textheight}}%WIDTH-LIMITED BIG WEDGE
			}{\textheight}% 
		}{0.5ex}}%
	\stackon[1pt]{#1}{\tmpbox}%
}
\newcommand\widevec[1]{%
\savestack{\tmpbox}{\stretchto{%
		\scaleto{%
			\scalerel*[\widthof{\ensuremath{#1}}]{\kern-.6pt\kern-.6pt}%
			{\rule[-\textheight/2]{1ex}{\textheight}}%WIDTH-LIMITED BIG WEDGE
		}{\textheight}% 
	}{0.5ex}}%
\stackon[1pt]{#1}{\tmpbox}%
}

\newcommand{\PredPhi}[2]{\PredPhiName{#1}\ #2\ #2}
\newcommand{\PredPhiName}[1]{{\typeC{#1}}^{\textcolor{\translationColor}{*}}}
\newcommand{\PhiAppl}[1]{\colorlet{temp}{.}\color{\translationColor}\overline{\color{temp}#1}\color{temp}}
\newcommand{\quasiImage}[1]{\reallywidetilde{#1}}
\newcommand{\termEqT}[3]{\PredPhiName{#1}\ #2\ #3}
\newcommand{\pbool}{\PredPhiName{\bool}}
\newcommand{\boolPred}[1]{\pbool\ #1\ #1}

\newcommand{\betaEtaRed}[1]{\ensuremath{#1^{\beta\eta}}}
\newcommand{\norm}[1]{\ensuremath{\mathsf{repl}\left[#1\right]}}
\newcommand{\sRed}[1]{\ensuremath{\mathsf{sRed}\left(#1\right)}}

\newcommand{\isEqRel}[1]{\mathrm{EqRel}(#1)}

\newcommand{\reserved}[1]{\ensuremath{\textcolor{\reservedColor}{\mathsf{#1}}}}%\textcolor{blue}{\textrm{#1}}}}
\usepackage{xspace,xcolor,realboxes}
\definecolor{backcolor}{gray}{.92}

\newcommand{\symbolFont}[1]{\ensuremath{\mathtt{#1}}}

\newcommand{\type}{\reserved{tp}}

\newcommand{\bool}{\reserved{bool}\xspace}
\newcommand{\T}{\reserved{true}\xspace}
\newcommand{\F}{\reserved{false}\xspace}

\newcommand{\dedH}{\mathbin{\vdash^H}}
\newcommand{\dedT}[1][]{\mathbin{\vdash^{#1}_{\theorycolor{T}}}}

\newcommand{\dedPT}[1][]{\mathbin{\vdash^{#1}_{\PsiAppl T}}}

\newcommand{\ded}{\reserved{\ensuremath{\vdash\,}}}
\newcommand{\termEquals}[1]{\ensuremath{\,\reserved{=}_{#1}\,}} % I prefer = if it doesn't ruin readability

\newcommand{\assert}{\triangleright\,}
\newcommand{\assertL}{\triangleright\,}

\newcommand{\thy}{\theorycolor{T}}
\newcommand{\thyT}{\theorycolor{\PhiAppl{T}}}
\newcommand{\ctx}{\contextcolor{\Gamma}}
\newcommand{\ctxT}{\contextcolor{\PhiAppl{\Gamma}}}
\newcommand{\ctxD}{\contextcolor{\Delta}}

\newcommand{\termF}{\termC{F}}
\newcommand{\termG}{\termC{G}}
\newcommand{\termFp}{\termC{F'}}

\newcommand{\termf}{\termC{f}}

\newcommand{\termg}{\termC{g}}
\newcommand{\termfp}{\termC{f'}}

\newcommand{\x}{\ensuremath{\varname{x}}\xspace}
\newcommand{\y}{\ensuremath{\varname{y}}\xspace}
\newcommand{\z}{\ensuremath{\varname{z}}\xspace}
\newcommand{\tm}{\ensuremath{\termC{t}}\xspace}
\newcommand{\s}{\ensuremath{\termC{s}}\xspace}
\newcommand{\A}{\ensuremath{\typeC{A}}\xspace}
\newcommand{\typB}{\ensuremath{\typeC{B}}\xspace}
\newcommand{\typC}{\ensuremath{\typeC{C}}\xspace}
\newcommand{\p}{\ensuremath{\termC{p}}\xspace}
\newcommand{\q}{\ensuremath{\termC{q}}\xspace}
\renewcommand{\r}{\ensuremath{\termC{r}}\xspace}
\renewcommand{\a}{\ensuremath{\tpdeclname{a}}\xspace}

\renewcommand{\c}{\ensuremath{\constname{c}}\xspace}

\newcommand{\varn}{\varname{n}\xspace}
\newcommand{\varm}{\varname{m}\xspace}
\newcommand{\varl}{\varname{l}\xspace}

\newcommand{\xp}{\ensuremath{\varname{x'}}}

\newcommand{\zp}{\ensuremath{\varname{z'}}}
\newcommand{\termtp}{\ensuremath{\termC{t'}}\xspace}
\renewcommand{\sp}{\ensuremath{\termC{s'}}}
\newcommand{\Ap}{\ensuremath{\typeC{A'}}}
\newcommand{\Bp}{\ensuremath{\typeC{B'}}}

\newcommand{\pp}{\ensuremath{\termC{p'}}}

\newcommand{\rp}{\ensuremath{\termC{r'}}}
\newcommand{\PhiA}{\ensuremath{\PhiAppl{\A}}\xspace}
\newcommand{\PhiB}{\ensuremath{\PhiAppl{\typB}}\xspace}

\newcommand{\termEqB}{\termEquals{\bool}}
\newcommand{\typeEquals}{\ensuremath{\mathbin{\equiv}}}
\newcommand{\subtyping}{\ensuremath{\mathbin{\prec:}}}
\newcommand{\supertyping}{\ensuremath{\mathbin{:\succ}}}
\newcommand{\substOp}[2]{\ensuremath{[\sfrac{\varname{#1}}{#2}]}}
\newcommand{\subst}[3]{\ensuremath{#1\substOp{#2}{#3}}}
\newcommand{\NDLine}[3]{#1\ded& #2 &&\text{#3}}
\newcommand{\NDLineH}[3]{#1\dedH& #2 &&\text{#3}}
\newcommand{\NDLinePT}[3]{#1\dedPT& #2 &&\text{#3}}
\newcommand{\NDLinePTG}[2]{\PhiAppl{\ctx}\dedPT& #1 &&\text{#2}}
\newcommand{\NDLinePTD}[2]{\NDLinePT{\ctxD}{#1}{#2}}

\newcommand{\NDLineT}[3]{#1\dedT& #2 &&\text{#3}}

\newcommand{\NDLineTG}[2]{\ctx\dedT& #1 &&\text{#2}}
\newcommand{\IH}{Induction hypothesis}
\newcommand{\byAss}{By Assumption}

\newcommand{\termColor}{blue}
\newcommand{\typeColor}{violet}
\newcommand{\translationColor}{purple}
\newcommand{\reservedColor}{black}

\newcommand{\varname}[1]{\ensuremath{\textcolor{\termColor}{#1}}}
\newcommand{\constname}[1]{\ensuremath{\textcolor{\termColor}{\symbolFont{#1}}}}
\newcommand{\tpdeclname}[1]{\ensuremath{\textcolor{\typeColor}{\symbolFont{#1}}}}
\newcommand{\axname}[1]{\ensuremath{\textcolor{orange}{\symbolFont{#1}}}}
\newcommand{\contextcolor}[1]{\ensuremath{\textcolor{red}{#1}}}
\newcommand{\theorycolor}[1]{\contextcolor{#1}}
\newcommand{\bnfcolor}[1]{\ensuremath{\textcolor{black}{#1}}}
\newcommand{\metavar}[1]{\ensuremath{\textcolor{magenta}{#1}}}

\newcommand{\termC}[1]{\color{\termColor}{#1}}
\newcommand{\typeC}[1]{\color{\typeColor}{#1}}

\ifcsmacro{grammar}{
	\renewenvironment{grammar}{\[\begin{array}{l@{\ \bbc\quad}l@{\quad}l}}{\end{array}\]}}{
	\newenvironment{grammar}{\[\begin{array}{lll}}{\end{array}\]}
	\newcommand{\bnf}[1]{\bnfcolor{#1}}
	\newcommand{\bbc}{\bnf{::=}}

	\newcommand{\opt}[1]{\bnf{[}#1\bnf{]}}
	\newcommand{\bnfbracket}[1]{\bnf{(}#1\bnf{)}}
	\newcommand{\rep}[1]{#1^{\bnf{\ast}}}

	}

\makeatletter
\newcommand{\rulelabelAppendix}[2]{%
	\protected@write \@auxout {}{\string \newlabel {#1}{{#2}{\thepage}{#2}{#1}{}} }%
	\hypertarget{#1}{}%
}
\makeatother

\newcommand{\rulelabel}[2]{}

\newcommand{\namedRule}[3]{\rul[\text{}]{#3}{#2}}
\newcommand{\rnamedRule}[4]{\rul[\text{}]{#3}{#2}}
\newcommand{\snamedRule}[3]{\rul[\text{}]{#3}{#2}}

\newcommand{\refnamedRule}[3]{\rul[\text{#1}]{#3}{#2}}
\newcommand{\refrnamedRule}[4]{\rul[\text{#1}]{#3\rulelabelAppendix{#4}{#1}}{#2}}
\newcommand{\refsnamedRule}[3]{\rul[\text{#1}]{#3\rulelabelAppendix{#1}{#1}}{#2}}

\newcommand{\ruleRef}[1]{(\ref{#1})}

\newcommand{\PrefLabelsSuff}[4]{&\label{#1}\addtocounter{#3}{1}\tag{#4\arabic{#3}#2}}

\newcommand{\sredlabel}[1]{\PrefLabelsSuff{#1}{}{SRlabel}{SR}}

\newbool{inAppendix}
\setbool{inAppendix}{false}
\newcommand{\plabel}[1]{\ifbool{inAppendix}{\PrefLabelsSuff{#1}{}{transDefnLabelsP}{PT}}{}}

\newcommand{\Impl}{\tb\text{implies}\tb}
\newcommand{\Mand}{\;\text{and}\;}

\usepackage{comment}
\newcommand{\negSp}{\!\!\!\!}
\newcommand{\QNegSp}{\negSp\negSp}
\newcommand{\Quad}{\quad\ \ \!\!\,}
\newcommand{\QQNegSp}{\QNegSp\QNegSp}
\newcommand{\QQuad}{\Quad\Quad}%\qquad\quad\negSp\;}
\newcommand{\QQQNegSp}{\QQNegSp\QQNegSp}
\newcommand{\QQQuad}{\QQuad\QQuad}
\newcommand{\QQQQuad}{\QQQuad\QQQuad}
\newcommand{\QQQQNegSp}{\QQQNegSp\QQQNegSp}

\renewcommand{\obj}{\tpdeclname{\symbolFont{obj}}\xspace}

\newtheorem{thm}{Theorem}

\newtheorem{rem}[thm]{Remark}

\usepackage[bookmarks,colorlinks]{hyperref}

\begin{document}
	
\title{Subtyping in DHOL -- Extended preprint}
	
\author{Colin Rothgang\inst{1,2}\orcidID{0000-0001-9751-8989} \and
Florian Rabe\inst{3}\orcidID{0000-0003-3040-3655}}
\authorrunning{Rothgang, Rabe}
\institute{Imdea Software Institute, Madrid, Spain\and Universidad Politécnica de Madrid, Madrid, Spain \and Computer Science, FAU Erlangen-N\"urnberg, Germany}
\maketitle

\begin{abstract}
The recently introduced dependent typed higher-order logic (DHOL) offers an interesting compromise between expressiveness and automation support.
It sacrifices the decidability of its type system in order to significantly extend its expressiveness over standard HOL.
Yet it retains strong automated theorem proving support via a sound and complete translation to HOL.

We leverage this design to extend DHOL with refinement and quotient types.
Both of these are commonly requested by practitioners but rarely provided by automated theorem provers.
This is because they inherently require undecidable typing and thus are very difficult to retrofit to decidable type systems.
But with DHOL already doing the heavy lifting, adding them is not only possible but elegant and simple.

Concretely, we add refinement and quotient types as special cases of subtyping.
This turns the associated canonical inclusion resp. projection maps into identity maps and thus avoids costly changes in representation.
We present the syntax, semantics, and translation to HOL for the extended language, including the proofs of soundness and completeness.
\end{abstract}

\section{Introduction and Related Work}\label{sec:intro}
\paragraph*{Motivation}
Recently dependently typed higher-order logic (DHOL) was introduced \cite{RRB:dhol:23}.
It is a variant of HOL \cite{churchtypes,andrews_truthproof} that uses dependent function types $\piType{x}{\A}\typB$ instead of simple function types $\A\to \typB$.
It is designed to remain as simple and as close to HOL and ATPs as possible while meeting the frequent user demand of dependent types.
Notably, contrary to typical formulations of dependent type theory, DHOL features a straightforward equality and classical Booleans at the cost of making typing undecidable.

Concretely, DHOL uses a type $\bool$ of propositions in the style of HOL, and equality $\s\termEquals{\A}\tm\ofT\,\bool$ of typed terms is a proposition, whose truth may depend on axioms in the theory or assumptions in the context.
Equality $\A\typeEquals \typB$ of types (which is not a proposition but a meta-level judgment) uses a straightforward congruence rule: if a dependent type constructor is applied to equal arguments, it produces equal types.
Thus, equality of types and typing are undecidable.
To yield practical tool support, DHOL reduces typing judgments to a series of proof obligations, and \cite{RRB:dhol:23} gives a sound and complete translation to HOL that allows using existing automated theorem provers for HOL to discharge these.

While undecidable typing is not used by most current ATPs or ITPs, it is justified by the pragmatic consideration that the ultimate task of theorem proving is undecidable anyway, and the difficulty of typing-related proof obligations is often small in comparison.
Follow-up work on DHOL includes a native ATP for DHOL \cite{dhol_satallax} and extensions with a choice operator \cite{dhol_choice} and polymorphism  \cite{dhol_poly}.

%Among \emph{interactive} prover developers, the subtle interaction between dependent types and decidability of typing is well-known.
%Most systems \cite{agda,coq,lean} based on Martin-L\"of type theory \cite{martinlof} use two equalities: an undecidable internal one subject to arbitrary axioms, and (possibly) a decidable meta-level one for type-checking.
%Reflecting equality means the former implies the latter; it is done less commonly, e.g., in \cite{andromeda}.

\paragraph*{Contribution}
The present paper leverages this key design choice and extends DHOL's expressivity at low cost by adding refinement and quotient types.
Both work elegantly in languages with undecidable typing so that DHOL, for which the necessary meta-theory and infrastructure already exist, is a good base to support them.
Indeed, the necessary changes to DHOL's syntax and semantics turned out to be few and simple --- only the extension of the proof was difficult.
We see DHOL as an intermediate between the automation support of HOL and the more expressive type theories of interactive provers such as those based on richer dependent type theories.
In this sense, the present paper pushes the boundary of ATP-near languages a little further.

\textbf{Refinement types} $\subtype{\A}{\p}$ is the subtype of $\A$ consisting of the terms satisfying the predicate $\p\ofT\A\to\bool$.
They correspond to comprehension in set theory.
They were already proposed in \cite{RRB:dhol:23} (with ad-hoc subtyping rules), and here we give a systematic treatment.
Critically, our refinement types avoid a change in representation: the injection $\subtype{\A}{\p}\to\A$ is always a no-op, i.e., an identity map.
This is in contrast to encodings of refinement types in decidable type theories, such as using the type $\Sigma \x\ofT\A.\p\,x$, or in set-theory, where, e.g., the injection $(\A\to\typB) \to (\subtype\A\p\to\typB)$ is not a no-op.

\textbf{Quotient types} $\quot{\A}{\r}$, intuitively, consist of all equivalence classes of the equivalence relation $\r\ofT\A\to \A\to\bool$.
Again we avoid representation changes:
We use all terms of type $\A$ as terms of type $\quot{\A}{\r}$ and adjust the equality $\termEquals{\quot{\A}{\r}}$ to obtain the quotient semantics.
Thus, the projection $\A\to\quot\A\r$ is a no-op and we have the subtyping relation $\A\subtyping\quot{\A}{\r}$.
In contrast, the usual definitions in set theory (via equivalence classes) or in decidable type theories (via setoids) require explicit\ednote{CR@FR: Reviewer 2 doesn't think this is really expensive, hence I changed expensive to explicit} changes in representation.

The statement $\A\subtyping\quot{\A}{\r}$ may look odd.
It is sound because we use a different equality relation at the two types: $\x\termEquals{\A}\y$ implies $\x\termEquals{\quot{\A}{\r}}\y$ but not the other way round.
This approach captures the mathematical practice of using elements of $\A$ as elements of the quotient, often to the point that readers do not even notice anymore they are technically working with equivalence classes.

\newcommand{\stdots}{\!\subtyping\!\!\ldots\!\!\subtyping\!}
Together, this yields a subtype hierarchy of refinements and quotients of $A$:
$\subtype{\A}{\lambdaFun{x}{\A}\F}
  \stdots
  \subtype{\A}{\p}
  \stdots
  \subtype{\A}{\lambdaFun{x}{\A}\T}\!\typeEquals\!\A\!\typeEquals\!\quot{\A}{\termEquals{\A}}
  \stdots
  \quot{\A}{\r}
  \stdots
  \quot{\A}{\lambdaFun{x,y}{\A}\T}
$
from initial (empty) type to terminal (singleton) type.

\paragraph{Related Work}
Refinement types have been considered on top of HOl in \cite{DBLP:conf/esop/Kuncar017}, but that logic has no subtyping, unlike our extension of DHOL.
Definitional quotient types have been considered in HOL as well in \cite{designHOQ} or \cite{designQIsabelleHol}, but in those logics the quotients are completely new types leading to explicit changes of representations. Even their combination has been considered before in HOL \cite{typesToPERs}, using PERs to represent types, similarly to our own approach, but again with an explicit change of representation rather than subtyping as in our approach for DHOL.

In formal systems, the approach of quotients as supertypes has been adopted only occasionally, e.g., in NuPRL's quotients \cite{nuprl} or in Quotient Haskell \cite{quotHaskell}.
NuPRL's type theory in particular features refinement and quotient types similar to ours and uses essentially the same PER semantics \cite{nuprl_per}.
The main difference to our approach is that DHOL tries to be as close to HOL (and thus HOL ATPs) as possible whereas NuPRL uses a very rich type theory.

PVS \cite{pvs} subsumes DHOL and refinement types with polymorphism.
%It lacks quotients but adds (among others) polymorphism.
%It focuses on interactive proving supported by decision procedures.
It does not support quotients, but its record subtyping resembles our quotient subtyping.
Notably, it treats refinement subtyping and record subtyping as two separate judgments with slightly different rules.

In soft type systems, all types are refinements of a fixed universe of objects.
For example, Mizar's \cite{mizar} type system is inherently undecidable and supports dependent types and refinement types.
It supports quotient types but only with a change in representation and not as supertypes.

Most ITPs allow users to construct refinement and quotient types at the cost of representation changes.
Systems based on decidable dependent type theory like Coq \cite{coq} or Lean \cite{lean} typically use $\Sigma$-types for refinement and setoids for quotients.
Lean provides some kernel support for quotients.
Systems based on HOL \cite{hol} use the subtype definition principle to introduce types isomorphic to $\subtype{\A}{\p}$.

\cite{dhol_poly} adds polymorphism to DHOL, and in future work we want to combine both features.
The key challenge will be to support \emph{subtype-bounded} polymorphism.
%We expect this to be in reach for future work, but we do not understand the interplay of undecidable typing, type variables, and type bounds well enough yet.
%Generally, \emph{automated} theorem proving for logics with dependent types is very difficult, and strong ATPs generally do not feature dependent types.
%But with ATPs becoming ever stronger, the approach of accepting undecidable typing to obtain simpler languages is becoming more appealing:
%  DHOL via its proximity in style and formal translation to HOL might be the first step towards developing strong ATP support for dependent types.

\paragraph*{Overview}
We give a self-contained definition of DHOL in Sect.~\ref{sec:prelim}.
Then we add subtyping in Sect.~\ref{sec:sub}, refinements in Sect.~\ref{sec:ref}, and quotients in Sect.~\ref{sec:quot}.
We develop the meta-theory in Sect.~\ref{sec:interact} (type normalization) and Sect.~\ref{sec:meta} (soundness/completeness).
We present an application to typed set theory in Sect.~\ref{sec:soft}.

\section{Preliminaries: Dependently Typed Higher-Order Logic}\label{sec:prelim}
  The DHOL \cite{RRB:dhol:23} \textbf{grammar} uses $\termC{\text{terms}}$ and $\typeC{\text{types}}$.
A theory $\thy$ declares typed constants $\c\ofT\A$, axioms $\assert \termF$, and dependent type symbols $\a\ofT\piType{x_1}{\metavar{A_1}}\ldots\piType{x_n}{\metavar{A_n}}\type$, which are applied to terms to obtain base types $\a\ \tm_1\ \ldots\ \tm_n$.
Contexts declare typed variables $\x\ofT\A$ and local assumptions $\assertL\termF$ (but no new types).
Dependent functions $\lambdaFun{x}{\A}\tm$ of type $\piType{x}{\A}\typB$ (written $\A\to \typB$ if \x does not occur free in \typB) map terms $\x\ofT\A$ to terms of type $\typB(\x)$.
We recover HOL as the fragment in which all base types $\a$ have arity $0$, in which case all function types are simple.
\begin{grammar}
	\thy & \emptyThy \alt \thy,\,\a\ofT (\piType{x}{\A}\!\!)^*\,\type \alt \thy,\,\c\ofT \A \alt \thy,\,\namedax{\axname{ax}}{\termF} &\text{theories}\\
	\ctx & \emptyCtx\alt \ctx,\,\x\ofT \A \alt \ctx,\,\assertL\termF	&\text{contexts} \\
	\A, \typB	& \a\ \metavar{\tm}^* \alt \piType{x}{\A} \typB \alt \bool &\text{types}\\
	\s,\tm,\termF,\termG & \c\alt \x \alt \lambdaFun{x}{\A} \tm \alt \s\ \tm\alt \s\termEquals{\A}\tm \alt \termF\impl \termG &\text{terms (incl. propositions)}
\end{grammar}

Following typical HOL-style \cite{andrews_truthproof}, DHOL defines all connectives and quantifiers from the equality connective $\s\termEquals{\A}\tm$.
For example, forall is defined as
$
\univQuant{x}{\A}\termF := \lambdaFun{x}{\A}\termF\termEquals{\piType{x}{\A}\bool}\lambdaFun{x}{\A}\T.
$
In particular, we can define the usual connectives and quantifiers and derive their usual proof rules.
The only subtlety is that, DHOL needs \emph{dependent} binary connectives: in an implication $\termF\implC \termG$, the well-formedness of $\termG$ may depend on the truth of $\termF$, and accordingly for conjunction $\termF\land \termG$ which is defined as $\neg(\termF\implC \neg \termG)$.
Because we see no way to define these from equality, we make dependent implication an additional primitive.
%Dependent conjunction and disjunction are then definable and behave accordingly.

DHOL uses axiomatic equality $\s\termEquals{\A}\tm$ in the style of FOL and HOL with a straightforward congruence rule for base types: type equality $\a\ \s_1\ \ldots\ \s_n\typeEquals\a\ \tm_1\ \ldots\ \tm_n$ holds if each $\s_i$ is equal to $\tm_i$.
This makes type equality and thus typing undecidable.
%On the plus side, the system behaves very intuitively, e.g., strong elimination of equality, where equal terms are substituted into a dependent type, poses no problem as the resulting two types are themselves equal.
In line with HOL's simplicity and unlike dependent type theories based on Martin-L\"of type theory \cite{martinlof}, there is no support for type-valued computation like large elimination.

%Observe that the non-emptiness assumption for DHOL types becomes the assumption that all HOL types are non-empty in this fragment.
 %(and thus all function types are simple function types).

\newcommand{\nat}{\ensuremath{\typeC{\mathrm{nat}}}\xspace}
\newcommand{\zero}{\ensuremath{\termC{\mathrm{zero}}}\xspace}
\newcommand{\suc}{\ensuremath{\termC{\mathrm{succ}}}\xspace}
\newcommand{\plus}{\ensuremath{\termC{\mathrm{plus}}}\xspace}
\newcommand{\List}{\ensuremath{\typeC{\mathrm{list}}}\xspace}
\newcommand{\nil}{\ensuremath{\termC{\mathrm{nil}}}\xspace}
\newcommand{\cons}{\ensuremath{\termC{\mathrm{cons}}}\xspace}
\newcommand{\conc}{\ensuremath{\termC{\mathrm{conc}}}\xspace}
\newcommand{\noConf}{\ensuremath{\typeC{\mathrm{noConf}}}\xspace}
\newcommand{\indList}{\ensuremath{\termC{\mathrm{indList}}}\xspace}
\newcommand{\indPrfList}{\ensuremath{\typeC{\mathrm{indPrfList}}}\xspace}
\newcommand{\ofLength}{\ensuremath{\termC{\mathrm{length}}}\xspace}
\newcommand{\numOcc}{\ensuremath{\termC{\mathrm{numOcc}}}\xspace}
\newcommand{\ite}[3]{\ensuremath{\termC{\mathrm{if}\ #1\ \mathrm{then}\ #2\ \mathrm{else}\ #3}}\xspace}
\newcommand{\nList}{\ensuremath{\typeC{\mathrm{llist}}}\xspace}
\newcommand{\nnil}{\ensuremath{\termC{\mathrm{lnil}}}\xspace}
\newcommand{\ncons}{\ensuremath{\termC{\mathrm{lcons}}}\xspace}
\newcommand{\nconc}{\ensuremath{\termC{\mathrm{lconc}}}\xspace}
\newcommand{\contains}{\ensuremath{\termC{\mathrm{contains}}}\xspace}
\newcommand{\multiset}{\ensuremath{\typeC{\mathrm{multiset}}}\xspace}
\newcommand{\set}{\ensuremath{\typeC{\mathrm{set}}}\xspace}
\newcommand{\Ptp}{\ensuremath{\typeC{\mathrm{P}}}\xspace}
\newcommand{\Pprop}{\ensuremath{\termC{\mathrm{P}}}\xspace}

\begin{example}[Lists]\label{ex:lists}
As an accessible running example, we show a formalization of lists over some type $\obj$, both plain lists $\List$ and lists $\nList\ \varn$ with fixed length.
Notably, the well-typedness of the statement of associativity of $\nconc$ now requires the associativity of $\plus$.
\[\mathll{
  \nat \ofT\;\type,\tb
  \zero\ofT\;\nat,\tb
  \suc \ofT\;\nat\to\nat,\tb
  \plus \ofT\;\nat\to\nat\to\nat,\\
 
	\obj\ofT\;\type,\tb
	\List\ofT\; \type,\tb
	\nil\ofT\; \List,\tb
	\cons\ofT\; \obj\to\List\to\List,\tb
	\conc\ofT\; \List\to\List\to\List,\\

%	\noConf&\ofT&\univQuant{\x}{\obj}\univQuant{\y}{\List}\neg (\cons\ \x\ \y \termEquals{\List}\nil)\\
%	\indList&\ofT& 
%		\indPrfList&\ofT \piType{\Pprop}{\piType{x}{\List}\bool} \left(\Pprop\ \nil\right)\implC  \left(\left(\univQuant{x}{\obj}\univQuant{y}{\List}\left(\Pprop\ \y\right)\implC\Pprop\ (\cons\ \x\ \y)\right)\right)\implC \univQuant{x}{\List}\Pprop\ \x
	\nList\ofT\; \nat\to\type,\tb
	\nnil\ofT\; \nList\ \zero,\\
	\ncons\ofT\; \piType{n}{\nat}\obj\to\nList\ \varn\to\nList\ (\suc\ \varn),\\
	\nconc\ofT\; \piType{m,n}{\nat}\nList\ \varm\to\nList\ \varn\to\nList\ (\plus\ \varm \ \varn)
}\]
\end{example}

\begin{figure}[hbt]
\begin{center}
	\begin{tabular}{|c|c|c|}
		\hline 
		Name & Judgment & Intuition \\ 
		\hline
		theories & $\ded \Thy{\thy}$ & $\thy$ is a well-formed theory\\
		contexts & $\dedT \Ctx{\ctx}$ & $\ctx$ is a well-formed context\\
		types & $\ctx\dedT \A\, \type$ & $\A$ is a well-formed type \\ 
		typing & $\ctx\dedT \tm\ofT\A$ & $\tm$ is a well-formed term of well-formed type $\A$ \\
		validity & $\ctx\dedT \termF$ & well-formed Boolean $\termF$ is derivable \\ 
		% no - special case of validity
		% 	equality of terms & $\ctx\dedT t =_A t'$ & special case of validity \\
		equality of types & $\ctx\dedT \A\typeEquals \typB$ & well-formed types $\A$ and $\typB$ are equal \\ 
		\hline 
	\end{tabular}
\caption{DHOL Judgments}\label{fig:judge}
\end{center}
\end{figure}

\begin{figure}[hbtp]
{\scriptsize
	Theories and contexts:
	\[
	\snamedRule{thyEmpty}{\ded\Thy{\emptyThy}}{}\tb
	\snamedRule{thyType'}{\ded\Thy{\concatThy{\thy}{\a\ofT \piType{x_1}{\typeC{A_1}}\ldots\piType{x_n}{\typeC{A_n}}\type}}}
	{\dedT \Ctx{\varname{x_1}\ofT \typeC{A_1},\,\ldots,\varname{x_n}\ofT \typeC{A_n}}} \tb
	\snamedRule{thyConst}{\ded\Thy{\concatThy{\thy}{\c\ofT\A}}}{\dedT \Type{\A}}\tb
	\snamedRule{thyAxiom}{\ded\Thy{\concatThy{\thy}{\namedax{c}{\termF}}}}{\dedT \termF\ofT \bool}
	\]
	\[
	\snamedRule{ctxEmpty}{\dedT\Ctx{\emptyCtx}}{\ded\Thy{\thy}}\tb
	\snamedRule{ctxVar}{\dedT\Ctx{\concatCtx{\ctx}{\x\ofT \A}}}{\ctx\dedT \Type{\A}}\tb
	\snamedRule{ctxAssume}{\dedT\Ctx{\concatCtx{\contextcolor{\ctx}}{\namedass{x}{\termF}}}}{\contextcolor{\ctx}\dedT \termF\ofT\bool}
	\]\\
	Well-formedness and equality of types:
	\[
	\snamedRule{type'}{\Gamma \vdash_T\ \a\ \termC{t_1}\ \ldots\ \termC{t_n}\;\type}
	{\begin{array}{c}\a\ofT \piType{x_1}{\typeC{A_1}}\ldots\piType{x_n}{\typeC{A_n}}\type \text{ in }T\\
			\ctx \dedT \termC{t_1}\ofT \typeC{A_1} \ \;\ldots\ \; \ctx \dedT \termC{t_n}\ofT \subst{\typeC{A_n}}{x_1}{t_1}\ldots\substOp{x_{n-1}}{t_{n-1}}\end{array}}\tb
	\rnamedRule{bool}{\ctx\dedT \Type{\bool}}{\dedT\Ctx\ctx}{bool'}\tb
	\snamedRule{pi}{\Gamma\dedT \piType{x}{\A}\typB\ \type}{\ctx\dedT \Type{\A} \quad \concatCtx{\ctx}{\x\ofT \A\dedT \Type{\typB}}}\tb
	\]
	\[
	\snamedRule{congBase'}{\ctx\dedT\a\ \termC{s_1}\ \ldots\ \termC{s_n}\typeEquals \a\ \termC{t_1}\ \ldots \termC{t_n}}
	{\begin{array}{c}\a\ofT \piType{x_1}{A_1}\ldots\piType{x_n}{A_n}\type\text{ in }\thy\\
			\ctx \dedT \termC{s_1}\termEquals{A_1}\termC{t_1}\, \ldots\, \ctx \dedT \termC{s_n}\termEquals{\subst{\typeC{A_n}}{x_1}{t_1}\ldots\substOp{x_{n-1}}{t_{n-1}}}\termC{t_n}\end{array}}\quad
  \rnamedRule{cong$\bool$}{\ctx\dedT \bool\typeEquals\bool}{\dedT\Ctx\ctx}{bool'}\quad
	\rnamedRule{cong$\Pi$}{\ctx\dedT\piType{x}{\A} \typB\typeEquals \piType{x}{\Ap}\Bp}{\ctx\dedT \A\typeEquals \Ap \ \;  \ctx,\x\ofT \A\dedT \typB\typeEquals \Bp}{congPi}
	\]
	\\Typing:
		\[
	\rnamedRule{const'}{\ctx\dedT \constname{c}\ofT \A}{c\ofT \Ap\thyIn{}{T}\tb \ctx\dedT\Ap\typeEquals\A}{const''}\tb
	\snamedRule{lambda'}{\ctx\dedT (\lambdaFun{x}{\A} \tm)\ofT  \piType{x}{\Ap}\typB}{\concatCtx{\ctx}{\x\ofT \A}\dedT \tm\ofT \typB\tb \Ap\typeEquals\A}\tb
	\rnamedRule{$\impl$type'}{\ctx\dedT \termF\Rightarrow \termC{G}\ofT\bool}{\ctx\dedT \termF\ofT\bool\tb \concatCtx{\ctx}{\namedass{x}{\termF}}\dedT \termC{G}\ofT\bool}{implType'}\tb
	%\tb\rnamedRule{congtp}{\ctx\dedT \termC{t'}\ofT \Ap}{\ctx\dedT \termC{t}\ofT\A\tb \ctx\dedT \termC{t}\termEquals{\A}\termC{t'}\tb\ctx\dedT \A\typeEquals\Ap}{congtp}
	% derivable from t =A t' implies t:A which is admissible
	\]
	\[
	\rnamedRule{var'}{\ctx\dedT \x\ofT \A}{\x\ofT \Ap\ctxIn{}{\ctx}\tb \ctx\dedT\Ap\typeEquals\A}{var''}\tb
	\snamedRule{appl'}{\ctx\dedT \termf\,\tm\ofT\subst{\typB}{x}{\typeC{t}}}{\ctx\dedT \termf\ofT \piType{x}{\A} \typB \tb \ctx\dedT \tm\ofT \A}\tb
	\rnamedRule{$=$type}{\ctx\dedT \s\termEquals{\A}\typeC{t}\ofT\bool}{\ctx\dedT \s\ofT\A\tb \ctx\dedT \tm\ofT \A}{eqType'}
	\]
%	\rnamedRule{$\subtype{}{p}$I}{\ctx\dedT \tm\ofT \subtype{\A}{\p}}{\ctx\dedT \tm\ofT \A \tb  \ctx\dedT \p\ \tm}{psubI}
	Equality: congruence, reflexivity, symmetry, $\beta$, $\eta$ (transitivity and extensionality are derivable):
	\[
	\rnamedRule{cong$\lambda$'}{\ctx\dedT \lambdaFun{x}{\A} t\termEquals{\piType{x}{\A}\typB} \lambdaFun{x}{\Ap} t'}{\ctx\dedT \A\typeEquals \Ap \ \;\concatCtx{\Gamma}{\x\ofT \A} \dedT \tm\termEquals{\typB} \termtp}{congLam'}
	\quad	
	\snamedRule{congAppl'}{\ctx\dedT \termf\ \tm\termEquals{\typB} \termfp\ \termtp}{\ctx\dedT \tm\termEquals{\A} \termtp\quad \ctx\dedT \termf\termEquals{\piType{x}{\A} \typB} \termfp}
	\]
	\[
	\snamedRule{refl}{\ctx\dedT \tm\termEquals{\A} \tm}{\ctx\dedT \tm\ofT \A}\tb
	\snamedRule{sym}{\ctx\dedT \s\termEquals{\A} \tm}{\ctx\dedT \tm \termEquals{\A}\s}\tb
	\snamedRule{beta}{\ctx\dedT (\lambdaFun{x}{\A} \s)\ \tm \termEquals{B} \subst{\s}{x}{\tm} }{\ctx\dedT (\lambdaFun{x}{\A} \s)\ \tm\ofT \typB}\tb
	\snamedRule{etaPi}{\ctx\dedT \tm\termEquals{\piType{x}{\A} \typB} \lambdaFun{x}{\A}\tm\ \x}{\ctx\dedT \tm\ofT \piType{x}{\A} \typB}
	\]\\
	Rules for validity: lookup, implication, Boolean equality and Boolean extensionality
	\[
	\snamedRule{axiom}{\ctx\dedT \termF}{\thyIn{\namedax{c}{\typeC{F}}}{\thy} \tb  \dedT\Ctx{\ctx}}\tb
	\rnamedRule{$\impl$I}{\ctx\dedT \termF\Rightarrow \termC{G}}{\concatCtx{\ctx}{\namedass{x}{\termF}}\dedT \termC{G}}{implI}\tb
	\rnamedRule{cong$\ded$}{\ctx\dedT \termF}{\ctx\dedT \termF\termEquals{\bool} \termC{F'}\tb \ctx\dedT \termC{F'}}{congDed}
	\]
	\[
	\snamedRule{assume}{\ctx\dedT \termF}{\ctxIn{\namedass{x}{\termF}}{\ctx} \tb  \dedT\Ctx{\ctx}}\tb
	\rnamedRule{$\impl$E}{\ctx\dedT \termC{G}}{\ctx\dedT \termF\Rightarrow \termC{G}\tb \ctx\dedT \termF}{implE} \tb
	\snamedRule{boolExt}{\ctx,\x\ofT \bool\dedT \p\ \x}{\ctx\dedT \p\ \T \tb \ctx\dedT \p\ \F}
	\]
}
\caption{DHOL Rules}\label{fig:rules}
\end{figure}

DHOL uses the \textbf{judgments} given in Fig.~\ref{fig:judge} and the \textbf{rules} listed in Fig.~\ref{fig:rules}.
Note that equality of terms is a special case of validity, whereas equality of types is not a Boolean but a separate judgment.
Thus, users cannot state axioms equating types, and type equality is defined only by congruence.
%Also note how the typing rule for implication allows using the truth of $\termF$ when checking $\termG$.
%Finally, while there is no type-inference judgement, \cite{RRB:dhol:23} presents a bi-directional type-checking algorithm.
The rules are straightforward.
In particular, type equality is checked structurally and reduced to a set of term equalities, which must then be discharged by an ATP.

Furthermore many well-known admissible HOL rules are also admissible in DHOL, see Appendix~\ref{sec:meta-thm:1}.

%Finally, we add one axiom about non-emptyness of DHOL types: we allow the existence of empty dependent types, but we require that for each HOL type in the image of the translation there exists one non-empty DHOL type translated to it (rather than requiring all dependent types translated to it to be non-empty). 

%%%%%%%%%%%%%%%%%%%%%%%%%%%%%%%%%%%%%%%%%%%%%
\begin{figure}[hbtp]\scriptsize
	\begin{gather*}
	\text{Theories and contexts:}\qquad \PhiAppl{\emptyThy} \;:=\;\emptyThy \tb
	\PhiAppl{\concatThy{\thy}{\metavar{D}}}\;:=\;\concatThy{\PhiAppl{\thy}}{\PhiAppl{\metavar{D}}}\tb
	\PhiAppl{\emptyCtx}\;:=\;\emptyCtx \tb
	\PhiAppl{\concatCtx{\ctx}{\metavar{D}}}\;:=\;\concatThy{\ctxT}{\PhiAppl{\metavar{D}}}
	\\
	\mathll{\PhiAppl{\a\ofT  \piType{x_1}{A_1}\ldots\piType{x_n}{A_n}\type} \;:=\;
  \a\ofT  \type, \;\;
  \PredPhiName{a}\ofT  \PhiAppl{\typeC{A_1}}\to\ldots\to \PhiAppl{\typeC{A_n}} \to \a\to \a\to \bool,\;\;\\
  \tb\namedax{a_{PER}}{} \univQuant{x_1}{\PhiAppl{\typeC{A_1}}}\ldots\univQuant{x_n}{\PhiAppl{\typeC{A_n}}}\univQuant{u,v}{\a} \PredPhiName{a} \ \varname{x_1}\ \ldots\ \varname{x_n}\ \varname{u}\ \varname{v}\impl \varname{u}\termEquals{\typeC{a}}\varname{v}}
	\\
%	Thus, $\a$ is translated to a base type of the same name without arguments and a trivial PER for every argument tuple.
%	Intuitively, $\PredPhiName{a} \ \PhiAppl{\termC{t_1}}\ \ldots\ \PhiAppl{\termC{t_n}}\ \varname{u}\ \varname{u}$ defines the subtype of the HOL-type $\a$ corresponding to the DHOL-type $\a\ \termC{t_1}\ \ldots\ \termC{t_n}$.
%		
%	Constant and variable declarations are translated by adding the assumptions that they are in the PER of their type, and axioms and assumptions are translated straightforwardly:
	\PhiAppl{\constname{c}\ofT \A} \;:=\; \constname{c}\ofT \PhiAppl{\A},\;\namedax{\typingAxName{c}}{\PredPhi{A}{\constname{c}}} \tb\tb
	\PhiAppl{\x\ofT \A} \;:=\; \x\ofT \PhiAppl{\A},\;\namedass{\typingAssName{x}}{\PredPhi{A}{\x}}
	\\
	\PhiAppl{\namedax{c}{\termF}} \;:=\; \namedass{c}{\PhiAppl{\termF}} \tb\tb
	\PhiAppl{\namedass{x}{\termF}}\;:=\; \namedass{x}{\PhiAppl{\termF}}
	\\
	\text{Types:}\qquad\PhiAppl{\a\ \termC{t_1}\ \ldots \ \termC{t_n}} \;:=\; \a \tb\tb 
	\termEqT{(\a\ \termC{t_1}\ \ldots \ \termC{t_n})}{\s}{\tm} \;:=\; \PredPhiName{a}\ \PhiAppl{\termC{t_1}}\ \ldots\ \PhiAppl{\termC{t_n}}\ \s\ \tm
	\\
	\PhiAppl{\piType{x}{\A}\typB} \;:=\; \PhiAppl{\A} \to \PhiAppl{\typB}\tb\tb
	\termEqT{(\piType{x}{\A}\typB)}{\termf}{\termC{g}} \;:=\; \univQuant{x,y}{\PhiAppl{\A}}%\univQuant{y}{\PhiAppl{\A}}
	\termEqT{\A}{\x}{\y}\impl \termEqT{\typB}{\left(\termf\ \x\right)}{\left(\termC{g}\ \y\right)}
	\\
	\PhiAppl{\bool} \;:=\; \bool \tb\tb
	\termEqT{\bool}{\s}{\tm}\;:=\; \s\termEqB \tm
	\\
	%\[\PhiAppl{\subtype{\A}{\p}} \;:=\; \PhiAppl{\A} \tb\tb
	%\termEqT{\left(\subtype{\A}{\p}\right)}{\s}{\tm} \;:=\; \termEqT{A}{s}{t}\land \PhiAppl{\p}\ \s\land \PhiAppl{\p}\ \tm\]
	%\[
	%\PhiAppl{\quot{\A}{\r}}:=\A\tb\tb 
	%\termEqT{\left(\quot{\A}{\r}\right)}{\s}{\tm}:=\PhiAppl{\r}\ \s\ \tm\land \PredPhi{\A}{\s}\land \PredPhi{\A}{\tm}
	%\]
	\text{Terms:}\qquad\PhiAppl{\constname{c}} \;:=\; \constname{c} \tb 
	\PhiAppl{\x} \;:=\; \x \tb
	\PhiAppl{\lambdaFun{x}{\A} \tm} \;:=\; \lambdaFun{x}{\PhiAppl{\A}} \PhiAppl{\tm} \tb
	\PhiAppl{\termf\ \tm} \;:=\; \PhiAppl{\termf}\ \PhiAppl{\tm}
	\\
	\PhiAppl{\s\termEquals{\A}\tm} \;:=\; \termEqT{A}{\PhiAppl{\s}}{\PhiAppl{\tm}}\tb
	\PhiAppl{\termF\Rightarrow \termC{G}} \;:=\; \PhiAppl{\termF} \impl \PhiAppl{\termC{G}}
	\end{gather*}
\begin{center}
\normalsize
	\begin{tabular}{|l|l|}
		\hline
		if in DHOL & then in HOL \\
		\hline
		%theory & theory with PER for each base type\\
		%context $\Gamma$ & context $\PhiAppl{\Gamma}$ with PER-assumption for each variable\\
		type $\A$ & type $\PhiAppl{\A}$ and PER $\PredPhiName{\A}\ofT \PhiAppl{\A}\to\PhiAppl{\A}\to\bool$ \\
		term $\tm\ofT \A$ & term $\PhiAppl{\tm}\ofT \PhiAppl{\A}$ satisfying $\PredPhi{\A}{\PhiAppl{\tm}}$\\
		\hline
	\end{tabular}
\end{center}
	%\[
	%\PhiAppl{\choiceOp{x}{\A}{\p}}:=\choiceOp{x}{\PhiAppl{\A}}{\lambdaFun{x}{\PhiAppl{\A}}
	%	\PredPhi{\A}{\x}\land \PhiAppl{\p}\ \x\land \univQuant{y}{\PhiAppl{\A}} \PredPhi{\A}{\y}\land \PhiAppl{\p}\ \y \implC \termEqT{A}{\x}{\y}}\]
\caption{Definition of the Translation DHOL$\to$HOL}\label{fig:trans2}
\end{figure}

The \textbf{semantics} for DHOL and a \textbf{practical ATP workflow} are given by a sound and complete translation to HOL.
The translation is dependency erasure, e.g., translating dependent types $\a\ \termC{t_1}\ \ldots\ \termC{t_n}$ to simple types $\a$, effectively ```merging''' all instances of a dependent type into a large simple type.
Fig.~\ref{fig:trans2} shows the details.

Typing and equality at $\A$ are recovered by generating a partial equivalence relation (PER) $\PredPhiName{\A}$ for every HOL-type $\PhiAppl{\A}$.
A PER is a symmetric-transitive relation and the same as an equivalence relation on a subtype of $\PhiAppl{\A}$.
Thus, $\A$ corresponds in HOL to the quotient of the appropriate subtype of $\PhiAppl{\A}$ by $\PredPhiName{\A}$.

DHOL terms are translated to their HOL analogues except that equality is translated to the respective PER: $\PhiAppl{\s\termEquals{\A}\tm}:=\PredPhiName{\A}\ \PhiAppl \s\ \PhiAppl \tm$.
In particular, the predicate $\PredPhi{A}{\PhiAppl{\tm}}$ captures whether $\tm$ is a term of type $\A$.
For $n$-ary type symbols $\a$, the translation generates an $n+2$-ary predicate $\PredPhiName{a}$ such that $\PredPhiName{a}\ \PhiAppl{\termC{t_1}}\ \ldots\ \PhiAppl{\termC{t_n}}$ is the PER for $\a\ \termC{\termC{t_1}}\ \ldots\ \termC{\termC{t_n}}$.
For function types, the PER is the usual condition for logical relations: functions are related if they map related inputs to related outputs.

\section{Subtyping}\label{sec:sub}
The treatment of quotients as supertypes and the use of different equality relations at different types are subtly difficult.
%\ednote{CR:Reviewer 2 at CSL would like to see us elaborate on the novelty of this definition. I'm sure this was done before, do you know any references? Ormaybe just point out this is the obvious and known way of doing this.}
Thus, we first introduce subtyping by its defining extensional property, from which we will derive all subtyping rules:

\newcommand{\STequiv}{STantisym}
%\newcommand{\STType}{type}
%We immediately have the derived rules:
%\[
%\snamedRule{\STintro}{\ctx\dedT \A \subtyping  \typB}{\Gamma,\x\ofT\A\dedT \x\ofT\typB}\tb
%\snamedRule{\STType}{\ctx\dedT \tm\ofT\typB}{\ctx\dedT\A\subtyping\typB\tb \ctx\dedT \tm\ofT\A}
%\]

\begin{definition}[Subtyping]\label{def:st}
$\ctx\dedT \A\subtyping \typB$ abbreviates $\ctx,\x\ofT\A\dedT \x\ofT\typB$.
\end{definition}
\begin{lemma}\label{lem:stder}
In any extension of DHOL, $\ctx\dedT \A\subtyping\typB$ iff $\rul{\ctx\dedT \tm\ofT \A}{\ctx\dedT \tm\ofT \typB}$ is derivable.
\end{lemma}
\begin{proof}
Left-to-right: We construct the function $(\lambdaFun{x}{\A}\x) \ofT \A\to\typB$ and derive the desired rule using the typing rule for function application.\\
Right-to-left: We start with $\ctx,\x\ofT\A\dedT \x\ofT\A$ and apply the derivable rule.
\end{proof}

This subtyping relation prevents \emph{incidental} subtype instances, for which the rule from Lem.~\ref{lem:stder} is admissible but not derivable.
For example, $\subtype{\A}{\lambdaFun{x}{\A}\F}$ is a subtype of all refinements of $\A$, but not of all types.
More generally, this definition precludes using induction on the terms of $\A$ to conclude $\A\subtyping \typB$.
This restriction ensures that subtyping is preserved under, e.g., theory extensions, substitution, or language extensions.
Importantly, subtyping preserves equality:

\begin{lemma}\label{lem:steq}
Consider some extension of DHOL with productions and rules. Assume ($\ast$) that $\ctx\dedT \x\termEquals{\typB} \x$ implies $\ctx\dedT\x:\typB$.
Let $F:=\univQuant{x}{\A}\univQuant{y}{\A}\x\termEquals{\A}\y\implC \x\termEquals{\typB}\y$.
Then $\ctx\dedT \A\subtyping\typB$ iff $\ctx\dedT F:\bool$; and if these hold, then also $\ctx\dedT F$.
\end{lemma}
($\ast$) is a very mild assumption and satisfied by all extensions given in this paper.
\begin{proof}
Left-to-right: The assumption yields $\ctx,\x\ofT\A,\,\y\ofT\A,\, \assertL\x\termEquals{\A}\y\dedT (\lambdaFun{x}{\A}\x)\ofT\A\to \typB$.
We get $\ctx,\x\ofT\A,\,\y\ofT\A,\,\assertL \x\termEquals{\A}\y\dedT$ $ (\lambdaFun{x}{\A}\x)\ \x\termEquals{\typB}(\lambdaFun{x}{\A}\x)\ \y$ from congruence of function application and reflexivity and $\x\termEquals{\typB} \y$ by $\beta$-reduction.\\
Right-to-left: Assume $\x:\A$. Instantiating $\termF$ twice with $\x$ and applying modus ponens with $\x\termEquals\A\x$ yields $\x\termEquals\typB\x$, from which we get $\x:\typB$.
\end{proof}

\begin{lemma}[Preorder of Types]
In any extension of DHOL, subtyping is reflexive (in the sense that $\ctx\dedT \A\typeEquals \typB$ implies $\ctx\dedT \A\subtyping \typB$) and transitive.	
%	Moreover, $\ctx\dedT \A\typeEquals \typB$ implies $\ctx\dedT \A\subtyping \typB$ and $\ctx\dedT \typB\subtyping \A$.
\end{lemma}
\begin{proof}
Reflexivity: The assumption yields $\ctx\dedT (\lambdaFun{x}{\A}\x)\ofT \A\to \typB$.
Applying both to a term $\tm$ of type $\A$ and $\beta$-reducing yields the rule from Lem.~\ref{lem:stder}.
Transitivity follows immediately from Lem.~\ref{lem:stder}.
\end{proof}

We also want to make subtyping an order.
Anti-symmetry with respect to $\typeEquals$ is not derivable directly, i.e., we might have $\ctx\dedT \A\subtyping \typB$ and $\ctx\dedT \typB\subtyping \A$, in which case $\A$ and $\typB$ would have the same terms, without being equal.
Therefore, we \emph{add} the anti-symmetry rule
\[
\refsnamedRule{\STequiv}{\ctx\dedT \A\typeEquals \typB}{\ctx\dedT \A\subtyping \typB\tb \ctx\dedT \typB\subtyping \A}
\]
Notably, this is the only \emph{change} made to DHOL so far --- everything before has just been abbreviations.
This change is conservative in the following sense:

\begin{theorem}[Conservativity]
  For	DHOL as defined so far (without the extension we introduce below), we have
	$\A \subtyping \typB$ iff $\A\typeEquals \typB$.
\end{theorem}
\begin{proof}
We show by induction on derivations that each term has a unique type up to type equality and that all term equality axioms preserve typing.
%This is proven for \hol in an auxiliary lemma in \cite{dhol} and without introducing refinement types, the proof works the same way in DHOL. The claim then follows immediately from that observation.
\end{proof}

\begin{theorem}[Variance and Congruence for Function Types]\label{lem:funeq}
The usual rules for function types are derivable:
\[
\rnamedRule{$\subtyping\!\!$Pi}{\ctx\dedT \piType{x}{\A}\typB\subtyping \piType{x}{\Ap}\Bp}{\ctx\dedT \Ap\subtyping \A\tb \concatCtx{\ctx}{\x\ofT \Ap}\dedT \typB\subtyping \Bp}{subtPi}\tb
\rnamedRule{$\subtyping\!\!$Pi}{\ctx\dedT \piType{x}{\A}\typB\typeEquals \piType{x}{\Ap}\Bp}{\ctx\dedT \Ap\typeEquals \A\tb \concatCtx{\ctx}{\x\ofT \Ap}\dedT \typB\typeEquals \Bp}{congPi}
\]
The second rule is primitive in DHOL but derivable in DHOL with subtyping.
\end{theorem}
%The second rule in Lem.~\ref{lem:funeq} is already part of DHOL (see Fig.~\ref{fig:rules}). So derivability here means it is now derivable from the remaining rules and thus redundant.
\begin{proof}
The first rule follows from the definition of subtyping and $\eta$-expansion.
The second rule is derived by \ruleRef{\STequiv}, establishing the hypotheses using the variance rule and reflexivity of subtyping.
\end{proof}

%The following two modified versions of the lookup rules, follow directly from the definition of the subtyping judgement:
%\[
%\rnamedRule{const'}{\ctx\dedT \c\ofT \A}{\c\ofT \Ap\thyIn{}{\theorycolor{T}}\tb \ctx\dedT \Ap\subtyping \A}{const''}\tb
%\rnamedRule{var'}{\ctx\dedT \x\ofT \A}{\x\ofT \Ap\ctxIn{}{\ctx}\tb \ctx\dedT \Ap\subtyping \A}{var''}
%\]
%Then we add the \textbf{rules} given in Figure~\ref{fig:refinementTypes}.
%These induce an algorithm for deciding subtyping relative to an oracle for the undecidable validity judgment.
%The latter enters the algorithm when two refinement types are compared.
%
%\begin{figure}[!ht]\small
%	Rules that relate type equality and subtyping:
%	\[
%	\rnamedRule{$\subtyping\!\!$I}{\ctx\dedT \A\subtyping \Ap}{\begin{matrix}
%			A,A'\text{ are base types}\\
%			\ctx\dedT \A\typeEquals \Ap
%	\end{matrix}}{subtI}\tb
%	%\rnamedRule{$\subtyping\!\! $Trans}{\ctx\dedT \A\subtyping \typC}{\ctx\dedT \A\subtyping \typB\tb \ctx\dedT \typB\subtyping \typC}{subtTrans}\tb 
%	\rnamedRule{$\subtyping$E}{\ctx\dedT \A\typeEquals\typB}{\ctx\ded \A\subtyping\typB\tb \ctx\ded \typB\subtyping\A}{subtE}
%	%\rnamedRule{$\subtyping\!\!$E}{\ctx\dedT c\ofT \Ap}{\ctx\dedT c\ofT \A\tb \ctx\dedT \A\subtyping \Ap\tb \thyIn{c\ofT \A}{T}}{subtE}
%	\]
%	\caption{Additional Rules for Refinement types and quotient types}\label{fig:refinementTypes}
%\end{figure}

\section{Refinement types}\label{sec:ref}
To add refinement types, we add only one production for types.
We do not add productions for terms --- refinement types only provide new typing properties for the existing terms.
Then we add rules for, respectively, formation, introduction, elimination (two rules), and equality:
\begin{grammar}
	\A & \subtype{\A}{\p} &\text{type $\A$ refined by predicate $\p$ on $\A$}
\end{grammar}
	\[\mathll{
	\rnamedRule{$\subtype{}{\p}\type$}{\ctx\dedT \subtype{\A}{\p}\; \type}{\ctx\dedT \p\ofT \A\to\bool}{psubType}\tb
	\rnamedRule{$\subtype{}{\p}$I}{\ctx\dedT \tm\ofT \subtype{\A}{\p}}{\ctx\dedT \tm\ofT \A \tb  \ctx\dedT \p\ \tm}{psubI}\\[4mm]
	\rnamedRule{$\subtype{}{\p}$E1}{\ctx\dedT \tm\ofT\A}{\ctx\dedT \tm\ofT \subtype{\A}{\p}}{psubE1}\tb
	\rnamedRule{$\subtype{}{\p}$E2}{\ctx\dedT \p\ \tm}{\ctx\dedT \tm\ofT \subtype{\A}{\p}}{psubE2}\tb
	\rnamedRule{$\subtype{}{\p}$Eq}{\ctx\dedT \s\termEquals{\subtype{\A}{\p}} \tm}{\tb\ctx\dedT \s\termEquals{\A}\tm\tb \ctx\dedT \p\ \s}{psubEq}
	}\]

\begin{example}[Refining Lists by Length]
We extend Ex.~\ref{ex:lists} by defining fixed-length lists as a refinement of lists.
%We need to add an induction principle $\indList$ for defining predicates of type $\piType{\x}{\List}\piType{\y}{\nat}\bool$ by induction on $\List$.
First, we axiomatize a predicate \ofLength on lists:
\[\mathll{
	%\indList&\ofT \piType{\Ptp}{\piType{x}{\List}\piType{\y}{\nat}\bool} \piType{Pn}{\Ptp\ \nil}\\
	%&\qquad \piType{Pcons}{\piType{x}{\obj}\piType{y}{\List}\piType{IHy}{\Ptp\ \y}\Ptp\ (\cons\ \x\ \y)} \piType{x}{\List}\Ptp\ \x\\
%	\indList&\ofT \piType{Pn}{\left(\piType{\y}{\nat}\bool\right)} \piType{Pcons}{\left(\piType{x}{\obj}\piType{y}{\List}\piType{IHy}{\left(\piType{\y}{\nat}\bool\right)}\piType{\z}{\nat}\bool\right)}\\&\quad \piType{x}{\List}\piType{\y}{\nat}\bool\\
	\ofLength \ofT\; \List\to\nat \tb\tb \assert\;\ofLength\ \nil\termEquals\nat \zero\\
	\assert\;\univQuant x\obj\univQuant l\List \ofLength\ (\cons\ x\ \varl)\termEquals\nat \suc\ (\ofLength\ \varl)
}\]
Then we define $\nList\ n\; :=\; \subtype{\List}{\lambdaFun l\List \ofLength\ l\termEquals\nat \varname{n}}$, and we can derive
\begin{align*}
	\ded\nil \ofT\nList\ \zero \tb\tb  \varn\ofT\nat\ded\cons:\piType x\obj\piType l{\nList \ \varn}\nList\ (\suc\ \varn)
\end{align*}
\end{example}

%Like for function types, we can \emph{derive} the congruence and variance rules:

\begin{theorem}[Congruence and Variance]\label{lem:varRefinements}
The following rules are derivable if the involved types are well-formed:
	\[\mathll{
	\rnamedRule{$\subtyping\!\! \subtype{}{\p}$}{\ctx\dedT \subtype{\A}{\p}\subtyping \subtype{\Ap}{\pp}}{\ctx\dedT \A\subtyping \Ap\tb \concatCtx{\ctx}{\x\ofT \A,\;\assertL\p\ \x}\dedT \pp\ \x}{subtPsubCong}\tb
	\rnamedRule{$\subtyping\!\! $Top}{\ctx\dedT \A\typeEquals \subtype{\A}{\lambdaFun{x}{\A}\T}}{\ctx\dedT\A\,\type}{subtPsubTriv}
	\\[1em]
	\rnamedRule{$\subtyping\!\! \subtype{}{\p}$}{\ctx\dedT \subtype{\A}{\p}\typeEquals \subtype{\Ap}{\pp}}{\ctx\dedT \A\typeEquals \Ap\tb \ctx\dedT \p\termEquals{\A\to\bool}\pp}{eqPsubCong}\tb
  \rnamedRule{$\subtyping\!\! $Top}{\ctx\dedT \subtype{\A}{\p}\subtyping \A}{\ctx\dedT \subtype\A\p\,\type}{subtPsubRef}
  }\]
\end{theorem}
\begin{proof}
To derive the first rule, we assume the hypotheses and $\x\ofT\subtype\A\p$.
The elimination rules yield $\x\ofT\A$ and $\p\ \x$, then the hypotheses yield $\x\ofT\Ap$ resp. $\pp\ \x$, finally the introduction rule yields $\x\ofT\subtype\Ap\pp$.\\
To derive the second rule, we apply {\ruleRef\STequiv} and use the introduction/elimination rules to show the two subtype relationships.\\
These then imply the other rules.
\end{proof}		

\section{Quotient types}\label{sec:quot}
To add quotient types we also extend the grammar with only one production for the type and rules for formation, introduction, elimination, and equality, where $\isEqRel{r}$ abbreviates that $\r$ is an equivalence relation:

\begin{grammar}
	\A & \quot{\A}{\r} &\text{quotient of $\A$ by equivalence relation $\r$}
\end{grammar}
%\begin{figure}[!ht]\small
	%Variance rules for quotient types types:
	%\[
	%\rnamedRule{Q$\subtyping\!\!$}{\ctx\dedT \quot{\A}{\r}\subtyping \quot{\Ap}{\rp}}{\ctx\dedT \A\subtyping\Ap\tb \concatCtx{\ctx}{\varname{x,y}\ofT\A}\dedT \r\ \x\ \y\implC \rp\ \x\ \y}{subtQ}
	%\]
	% This rule can be derived from the below rules.
	\[
	\rnamedRule{Q}{\ctx\dedT \Type{\quot{\A}{\r}}}{\ctx\dedT \Type{\A}\tb \ctx\dedT \r\ofT \A\to\A\to\bool\tb \ctx\dedT\isEqRel{\r}}{Qtype}\tb
	\rnamedRule{QI}{\ctx\dedT \tm\ofT\quot{\A}{\r}}{\ctx\dedT \tm\ofT\A\tb \ctx\dedT \quot\A\r\,\type}{Qintro}\]
	%\snamedRule{QE}{\ctx\dedT \quot{\A}{\r}\subtyping \A}{\ctx\dedT \Type{\quot{\A}{\r}}\tb \ctx,\x,\y\ofT\A\dedT\r\ \x\ \y\implC \x\termEquals{\A}\y}\\[.5em]
	%\snamedRule{quotE}{\ctx,\, u\ofT\quot\A\r\dedT \tm \ofT\typB}{\ctx,\, u\ofT\A\dedT \tm\ofT \typB\tb \ctx,\,\x\ofT\A,\,\y\ofT\A,\,xy\ofT \r\ \x\ \y\dedT \tm\substOp u\x\termEquals{\typB\substOp u\x} \tm\substOp u\y}\\[.5em]
	\[\rnamedRule{QE}{\ctx\dedT \tm\substOp\x\s \ofT\typB\substOp\x\s}{
	  \ctx\dedT \s\ofT\quot\A\r\quad
	  \ctx,\,\x\ofT\A,\,\assertL\x\termEquals{\quot\A\r}\s\dedT \tm\ofT \typB\quad
	  \ctx,\,\x\ofT\A,\,\xp\ofT\A,\,\assert\x\termEquals{\quot\A\r}\s,\,\assertL\xp\termEquals{\quot\A\r}\s\dedT \tm\termEquals{\typB} \tm\substOp\x\xp
	}{quotE} \]
\[\rnamedRule{Q$\termEquals{}$}{\ctx\dedT (\s\termEquals{\quot{\A}{\r}}\tm)\termEquals\bool(\r\ \s\ \tm)}
	{\ctx\dedT \s\ofT\A\tb \ctx\dedT \tm\ofT\A \tb \ctx\dedT\r\ofT\A\to\A\to\bool\tb\ctx\dedT\isEqRel{\r}}{QEq}
	\]

\begin{example}[Sets]\label{ex:sets}
We extend Ex.~\ref{ex:lists} by obtaining sets as a quotient of lists.
First, we axiomatize a predicate for containing an element:
\[\mathll{
%			\indList&\ofT \piType{Pn}{\left(\piType{\y}{\obj}\nat\right)} \piType{Pcons}{\left(\piType{x}{\obj}\piType{y}{\List}\piType{IHy}{\left(\piType{\y}{\obj}\nat\right)}\piType{\z}{\obj}\nat\right)}\\
%				&\quad \piType{x}{\List}\piType{\y}{\obj}\nat\\
%			\numOcc&\ofT \piType{\x}{\List}\piType{\y}{\obj}\bool := \lambdaFun{x}{\List} \lambdaFun{\y}{\obj}\indList\ 	(\lambdaFun{x}{\obj}\zero) \\
%				&\quad \big(\lambdaFun{\x}{\obj}\lambdaFun{\y}{\List} \lambdaFun{IHy}{\left(\piType{y}{\obj}\nat\right)}\lambdaFun{z}{\obj} \\
%				&\qquad\ite{\z\termEquals{\obj}\x}{\suc\ (\varname{IHy}\ \z)}{\varname{IHy}\ \z}\big)\ \x\ \y\\
%
%			\multiset&\ofT \type := \quot{\List}{\lambdaFun{x,y}{\List} \numOcc\ \x\termEquals{\piType{z}{\obj}\nat}\numOcc\ \y}\\
%			\contains&\ofT \piType{x}{\List}\piType{y}{\obj}\bool := \lambdaFun{x}{\List}\lambdaFun{y}{\obj}\neg (\numOcc\ \x\ \y\termEquals{\nat}\zero)\\
%			\set&\ofT \type := \quot{\multiset}{\lambdaFun{x,y}{\List} \contains\ \x\termEquals{\piType{z}{\obj}\nat}\contains\ \y}
\contains\ofT\; \List\to\obj\to\bool\tb\tb
\assert\; \univQuant x\obj \neg(\contains\ \nil\ \x)\\
\assert\; \univQuant x\obj \univQuant y\obj\univQuant l\List (\contains\ (\cons\ \y\ \varl)\ \x)\termEquals\bool (\x\termEquals\obj \y \vee  \contains\ \varname{l}\ \x)
}\]

Now we can define $\set:=\quot{\List}{\lambdaFun l\List\lambdaFun m\List\univQuant x\obj \contains\ \varname{l}\ \x\termEquals\bool\contains\ \varname{m}\ \x}$ as the type of lists containing the same elements.
The equality at \set immediately yields extensionality $\ded \univQuant{\x,\y}{\set} \x\termEquals{\set}\y\Leftrightarrow \left(\univQuant{z}{\obj}\contains\ \x\ \z\termEqB \contains\ \y\ \z\right)$.

Any $l\ofT\List$ can be used as a representative of the respective equivalence class in $\set$, and operations on sets can be defined via operations on lists, e.g., we can establish $\vdash \conc:\set\to\set\to\set$.
To derive this, we assume $\termC{u}\ofT\set$ and apply the elimination rule twice.
First we apply it with $\typB=\List\to\set$ and $\tm=\conc\ \termC{u}$;
we have to show $\conc\ \x\termEquals{\List\to\set}\conc\ \xp$ under the assumption that $\x$ and $\xp$ are equal as sets.
That yields a term $\conc\ \termC{u}:\List\to\set$.
We assume $\termC{v}\ofT\set$ and apply the elimination rule again with $\typB=\set$ to obtain $\conc\ \termC{u}\ \termC{v}\ofT\set$, and then conclude via $\lambda$-abstraction and $\eta$-reduction.
\end{example}

The elimination rule above looks overly complex.
It can be understood best by comparing it to the following, simpler and more intuitive rule
\begin{align*}
\snamedRule{quotES}{\ctx,\, \x\ofT\quot\A\r\dedT \tm \ofT\typB}{\ctx,\, \x\ofT\A\dedT \tm\ofT \typB\tb \ctx,\,\x\ofT\A,\,\xp\ofT\A,\,\assertL\r\ \x\ \xp\dedT \tm\termEquals{\typB} \tm\substOp \x\xp}(\ast)
\end{align*}
This rule captures the well-known condition that a function $\tm$ on $\A$ may be used as a function on $\quot\A\r$ if $\tm$ maps equivalent representatives $\x,\xp$ equally.
It follows from our elimination rule by putting $\s=\x$, but is subtly weaker: 

\begin{example}
Continuing Ex.~\ref{ex:sets}, assume a total order on $\obj$ and a function $\termg\ofT\subtype\List{\mathrm{nonEmpty}}\to\obj$ picking the maximum from a non-empty list.
We should be able to apply $\termg$ to some $\s\ofT\set$ that we know to be non-empty.
But if we try to apply $(\ast)$ to obtain $\termg\ \s\ofT\,\obj$, we get stuck trying to prove $\termg\ \x\termEquals\obj \termg\ \xp$ for any $\x,\xp$ that are representatives of an \emph{arbitrary} equivalence class of lists.
We cannot use the condition that $\s$ is non-empty and thus only non-empty lists need to be considered.
Thus, we cannot derive the well-formedness of $\termg\ \x$.

Our elimination rule remedies that: here we need to show $\termg\ \x\termEquals\obj \termg\ \xp$ for any $\x,\xp$ that are representatives of \emph{the class of} $\s$.
Thus, we can use that $\x$ and $\xp$ are non-empty and that thus $\termg\ \x$ is well-formed.
\end{example}

In dependent type theory, the two elimination rules are equivalent because we have a type $\s=\x$ and can use $\piType{s}{\quot{A}{r}}\s=\x \to \typB$ as the return type.
This is not possible in DHOL where $\s=\x$ is not a type but a Boolean. 

\begin{theorem}[Congruence and Variance]\label{lem:varQuotients}
The following rules are derivable if the involved types are well-formed and $\r,\rp$ are equivalence relations:
	\[\mathll{
	\rnamedRule{}{\ctx\dedT \tb\quot{\A}{\r}\subtyping \quot{\Ap}{\rp}}{\ctx\dedT \A\subtyping \Ap\tb \concatCtx{\ctx}{\x\ofT \A,\;\y\ofT\A,\;\assertL\r\ \x\ \y}\dedT \rp\ \x\ \y}{subtQuotCong}\tb
	\rnamedRule{}{\ctx\dedT \A\typeEquals \quot{\A}{\lambdaFun{x}{\A}\lambdaFun{y}{\A}\x\termEquals\A\y}}{\ctx\dedT \A\,\type}{subtQuotTriv}
	\\[1em]
	\rnamedRule{}{\ctx\dedT \A\subtyping \quot\A\r}{\ctx\dedT \quot\A\r\,\type}{subtQuotRef}\tb
	\rnamedRule{}{\ctx\dedT \quot{\A}{\r}\typeEquals \quot{\Ap}{\rp}}{\ctx\dedT \A\typeEquals\Ap\tb \ctx\dedT \r\termEquals{\A\to\A\to\bool}\rp}{subtQuotCong}}
  \]
\end{theorem}
\begin{proof}
For the first rule: Assume the hypotheses and $\s\ofT\quot\A\r$.
Apply the elimination rule with $\typB=\quot\Ap\rp$ and $\tm=\x$.
 $\x\ofT\A,\;\y\ofT\A,\;\assertL\x\termEquals{\quot\A\r}\s,\;\assertL\xp\termEquals{\quot\A\r}\s\dedT\x\termEquals{\quot\Ap\rp}\xp$
by the equality rule (using $\A\subtyping\Ap$ and the second assumption). %which requires $\x,\xp\ofT\Ap$ (shown using $\A\subtyping\Ap$) and $\rp\ \x\ \xp$, our second assumption.\\
For the second rule, apply {\ruleRef\STequiv} and use the introduction/elimination rules to show the two subtype relationships.
The other rules follow from those two.
\end{proof}

\section{Normalizing Types}\label{sec:interact}
To build a type-checker, we derive normalization rules that reduce subtyping conditions to validity conditions that can then be discharged via an ATP.
We prove rules for merging consecutive refinements and quotients and for the $4$ possible combinations of a function type with a refinement or quotient:

\begin{theorem}[Repeated Refinement/Quotient]\label{lem:normSub}
The following are derivable whenever the LHS is well-formed
\begin{align}
&\dedT\subtype{(\subtype\A\p)}\pp \typeEquals \subtype\A{\lambdaFun x\A \p\ \x\wedge \pp\ \x} \tag{RR}\label{RR}\\
&\dedT\quot{(\quot\A\r)}\rp \typeEquals \quot\A{\lambdaFun x\A\lambdaFun y\A \rp\ \x\ \y}\tag{QQ}\label{QQ}
\\
&\dedT\subtype{(\quot\A\r)}\p \typeEquals \quot{(\subtype\A\p)}\r\tag{RQ}\label{RQ}
\end{align}
\end{theorem}
\begin{proof}
\emph{\eqref{RR}}: 
well-formedness of the LHS yields $\p\ofT\A\to\bool$ and $\pp\ofT\subtype\A\p\to\bool$, so $\pp\ \x$ is well-formed as $\wedge$ is a \emph{dependent} conjunction and $\p\ \x$ can be assumed while checking $\pp\ \x$.
The well-formedness of the right-hand side (RHS) follows. 
Verifying the equality is straightforward by showing subtyping in both directions.\\
\emph{\eqref{QQ}}: 
well-formedness of the LHS yields $\r\ofT\A\to\A\to\bool$ and $\rp\ofT\quot\A\r\to\quot\A\r\to\bool$, so $\A\subtyping \quot\A\r$ implies $\rp\ \x\ \y$ (and thus the RHS) is well-formed.
The relation on the RHS is an equivalence relation since $\rp$ is.
To verify the type equality, we use Lem.~\ref{lem:steq} and show that both types induce the same equality.
In particular, the type of $\rp$ already guarantees that it subsumes $\r$.\\
\emph{\eqref{RQ}}: well-formedness of the LHS yields $\r\ofT\A\to\A\to\bool$ and $\p\ofT\quot\A\r\to\bool$, implying  $\r\ofT\subtype\A\p\to\subtype\A\p\to\bool$ and $\p\ofT\A\to\bool$. The well-formedness of the RHS follows. (Note the other direction does not hold in general.)
To show the equality, we show both subtyping directions.
For LHS$\subtyping$RHS, we assume $\x\ofT\quot\A\r$ and $\p\ \x$ and apply the elimination rule for quotients using $\tm=\x$ and $\typB=\quot{(\subtype\A\p)}\r$.
(Critically, this step would not go through if we had only used the weaker rule $\ast$ in Sect.~\ref{sec:quot}.)
For RHS$\subtyping$LHS, we assume $\x\ofT\quot{(\subtype\A\p)}\r$ and apply the elimination rule for quotients using $\tm=\x$.
\end{proof}

\begin{theorem}[Refinement/Quotient in a Function Type]\label{lem:piRQ}
The following are derivable if either side is well-formed:
\begin{align}
&\dedT
\piType{x}{\A}(\subtype{\typB}{\p})\typeEquals
\subtype{(\piType{x}{\A}\typB)}{\lambdaFun{f}{(\piType{x}{\A}\typB)}\univQuant{x}{\A}\p\ \left(\termf\ \x\right)}&&\tag{RCod}\label{RCod}
\\
&\dedT
\piType{x}{\quot{\A}{\r}}\typB\typeEquals
\subtype{(\piType{x}{\A}\typB)}{\lambdaFun{f}{\piType{x}{\A}\typB}\univQuant{x,y}{\A}\r\ \x\ \y\implC \left(\termf\ \x\right)\termEquals{\typB}\left(\termf\ \y\right)}&&\tag{QDom}\label{QDom}
\\
&\dedT
\piType{x}\A\quot\typB\r \supertyping
\quot{(\piType x\A\typB)}
  {\lambdaFun{f,g}{\piType{x}{\A}\typB}\univQuant{x}{\A}\r\ \left(\termf\ \x\right)\left(\termg\ \x\right)}&&\tag{QCod}\label{QCod}\\
\intertext{The following is derivable if the RHS is well-formed:}
&\dedT
\piType{x}{\subtype\A\p}\!\!\typB \supertyping
\quot{(\piType{x}{\A}\!\!\typB)}
  {\lambdaFun{\termf,\termg}{\piType{x}{\A}\!\!\typB}\!\univQuant{\x}{\A}\!\p\ \x\implC \left(\termf\ \x\right)\termEquals{\typB}\!\left(\termg\ \x\right)}\!\!\!\!&&\tag{RDom}\label{RDom}
\end{align}
\end{theorem}
\begin{proof}
\emph{\eqref{RCod}}:
Both subtyping directions are straightforward, as terms on either side are given by $\lambdaFun x\A\tm$ where $\tm$ has type $\typB$ and satisfies $\p$.\\
\emph{\eqref{QDom}}:
Both subtyping directions are straightforward, as both sides are subtypes of $\piType x\A\typB$ so their elements must preserve $\r$.\\
\emph{\eqref{QCod}}:
Assume a term $\termf$ of RHS-type and show $\x\ofT\A\ded \termf\ \x\ofT\quot\typB\r$ using the rules for quotients.\\
\emph{\eqref{RDom}}:
Assume a term $\termf$ of RHS-type and show $\x\ofT\subtype\A\p\ded \termf\ \x\ofT\typB$ using the quotient elimination rule.
The well-formedness of the LHS does not imply the well-formedness of the RHS since the well-formedness of $\typB$ can rely on $\p\ \x$.
\end{proof}

Maybe surprisingly, two of the subtyping laws in Thm.~\ref{lem:piRQ} are not equalities.
The law for the refined domain \emph{must not} be an equality:

\begin{example}[Refined Domain \eqref{RDom}]\label{ex:refdom}
The assumption $\p\ \x$ makes more terms well-typed, thus there may be functions $\piType{x}{\subtype{\A}{\p}}\typB$ that are not a restriction of a function $\piType{x}{\A}\typB$.
Consider the theory $\a\ofT\bool\to\type,\;\c\ofT \a\,\T$.
Then $\a\,\F$ is empty and so are $\piType{x}{\bool}\a\,\x$ and its quotients.
But with $\p=\lambdaFun{x}{\bool}\x$, we have
$\ded \lambdaFun{x}{\subtype{\bool}{\p}}\c\;\ofT\;\piType{x}{\subtype{\bool}{\p}}\a\ \x$.
\end{example}

The law for the quotiented codomain \emph{may or may not} be an equality.
This is related to the axiom of choice.
Consider the two statements
\begin{align*}
\dedT \existQuant{repr}{\quot{\typB}{\r}\to\typB} \termC{repr}\termEquals{\quot\typB\r\to\quot\typB\r}\lambdaFun x{\quot\typB\r}\x\\[0.5em]
\termf\ofT\piType x\A\quot\typB\r\dedT \existQuant{\termg}{\piType{x}\A\typB} \termf\termEquals{\piType x\A\quot\typB\r}\termg
\end{align*}
(Note that the first one is well-typed because $\typB\subtyping\quot\typB\r$.) % and thus $\termC{repr}:\quot\typB\r\to\quot\typB\r$.)
Both have a claim to be called the axiom of choice: The first one expresses that every equivalence relation has a system of representatives.
The second generalizes this to a family of equivalence relations.
The latter implies the former (put $\A:=\quot\typB\r$ and $\termf:=\lambdaFun x{\quot\typB\r}\x$).
In the simply-typed case the former also implies the latter (pick $\termC{repr}\circ\termf$ for $\termg$); but in the dependently-typed case, $\typB$ and $\r$ may depend on $\x$ and the implication might not hold. % depending on what other language features are present (e.g., $\Sigma$-types or choice).

Both statements construct a new term from an existing one ($\termC{repr}$ behaves like the identity, and $\termg$ like $\termf$) that has a different type but behaves the same up to quotienting.
If the direction $\subtyping$ were to hold in the law for the refined codomain, it would not only imply the existence of $\termg$ from $\termf$ but also allow using $\termf$ as a representative of the equivalence class of possible values for $\termg$.
That is in keeping with our goal of avoiding changes of representation. Therefore:

\begin{definition}[Quotiented Codomain]\label{ax:quotcod}
We adopt the rule below (which is an equality with Thm.~\ref{lem:piRQ}) as an axiom whenever either side is well-formed:
\begin{align}
	\dedT
	\piType{x}\A\quot\typB\r \subtyping
	\quot{(\piType x\A\typB)}
	{\lambdaFun{f,g}{\piType{x}{\A}\typB}\univQuant{x}{\A}\r\ \left(\termf\ \x\right)\left(\termg\ \x\right)}\tag{**}\label{quotAx}
\end{align}
%\[\ded
%\piType{x}\A\quot\typB\r \typeEquals
%\quot{(\piType x\A\typB)}
%  {\lambdaFun{f,g}{\piType{x}{\A}\typB}\univQuant{x}{\A}\r\ \left(\termf\ \x\right)\left(\termg\ \x\right)}
%\]
\end{definition}

Aggregating the above laws, we obtain a normalization algorithm for types:
\begin{theorem}[Normalizing Types]\label{thm:normalizing}
Every type is equal to a type of the form $\quot{(\subtype\A\p)}\r$ where $\A,\typB::=\bool\alt\a\ \metavar{t}^* \alt \piType{x}{\subtype\A\p}\typB$.
In particular, if a type does not use refined domains, it is equal to a quotient of a refinement of a DHOL type.
\end{theorem}
\begin{proof}
Using Thm.~\ref{lem:piRQ} with the axiom from Def.~\ref{ax:quotcod}, all refinements and quotients can be pushed out of all function types except for a single refinement of the domain; if there is no such refinement, we can use $\p:=\lambdaFun x\A\T$.
And using Thm.~\ref{lem:normSub}, those can be collected into a single quotient+refinement.
\end{proof}

Together with the equality and variance rules from Thm.~\ref{lem:funeq}, \ref{lem:varRefinements},~\ref{lem:varQuotients}, this induces a \textbf{subtype-checking algorithm} that reduces subtyping to validity without ever expanding Def.~\ref{def:st}.
The latter is important because expanding Def.~\ref{def:st} would recurse into a computationally expensive problem.
This algorithm is obviously sound as it only chains derived rules.
We are confident, but have not proved yet, that it is also complete in the sense that subtyping can always be derived by using only our derived rules (i.e., without Def.~\ref{def:st}) and a sound and complete theorem prover (which we obtain in Sect.~\ref{sec:meta}).

It may be surprising and certainly complicates subtype-checking that we need to allow for refined domains in the normal forms.
This limitation echoes an observation first made in \cite{pvs_predicatesets} about combining dependent types and refinements.
Effectively, the culprits are partial dependent functions that cannot be extended to total functions because the return type is well-defined and non-empty only for the refined domain.
For future work, it would be interesting to use a higher-order logic with partial functions as a translation target, like the one of \cite{imps_kernel}.
But we have not considered that option due to the lack of ATP support for such logics.

%\section{A Subtyping Algorithm and Implementation}\label{sec:algo}
%\input{algo}

%\section{Choice operators}
%\input{choice}
%  
%\section{Sigma types}
%\input{sigmas}

\section{Soundness and Completeness}\label{sec:meta}
We extend the translation from Fig.~\ref{fig:trans2} with cases for our two new productions:
\begin{align*}
\mathll{
\PhiAppl{\subtype{\A}{\p}} := \PhiAppl{\A}\plabel{PTPStype}\tb\tb
	\termEqT{\left(\subtype{\A}{\p}\right)}{\s}{\tm} := \termEqT{\A}{\s}{\tm}\land \PhiAppl{\p}\ \s\land \PhiAppl{\p}\ \tm \\[0.2em] %\plabel{PTPSpred} \\
\PhiAppl{\quot{\A}{\r}} := \PhiAppl{\A}\tb\tb
\termEqT{\left(\quot{\A}{\r}\right)}{\s}{\tm} := \PhiAppl{\r}\ \s\ \tm\land\PredPhi{A}{\s}\land\PredPhi{A}{\tm}
}\end{align*}
These definitions are not surprising as PERs in HOL are known to be closed under refinements and quotients.
%We will use the terms \emph{sound} and \emph{complete} from the perspective of using a HOL-ATP for theorem proving in \dhol, e.g., \emph{sound} means if $\PhiAppl{\termF}$ is a HOL-theorem, then $\termF$ is a \dhol-theorem, and \emph{complete} is the dual.

\begin{example}[PERs for a Quotiented Codomain]\label{exam:persQuotCod}
We calculate the PERs for both sides of the law for quotiented codomains in Thm.~\ref{lem:piRQ}:
\[
\mathll{\termEqT{\left(\piType{x}\A\quot\typB\r\right)}{\termf}{\termg}\\
%=\univQuant{\x,\y}{\PhiAppl{\A}} \termEqT{A}{\x}{\y}\implC\termEqT{\left(\quot\typB\r\right)}{(\termf\ \x)}{(\termg\ \y)}\\
=\univQuant{\x,\y}{\PhiAppl{\A}} \termEqT{A}{\x}{\y}\implC
(\PhiAppl{\r}\ (\termf\ \x)\ (\termg\ \y)\land \PredPhi{B}{(\termf\ \x)}\land \PredPhi{B}{(\termg\ \y)})
}\]
Both this and
$\termEqT{\left(\quot{(\piType x\A\typB)}
{\lambdaFun{f,g}{\piType{x}{\A}\typB}\univQuant{x}{\A}\r\ \left(\termf\ \x\right)\left(\termg\ \x\right)}\right)}{\termf}{\termg}$
 %=\\
%\termEqT{\left(\piType{x}{\A}\typB\right)}{\termf}{\termg}\lor \left(\univQuant{x}{\PhiAppl{\A}}\PredPhi{A}{\x}\implC \PhiAppl{\r}\ (\termf\ \x)\ (\termg\ \x)\right)\\
%\univQuant{x}{\PhiAppl{\A}}\PredPhi{A}{\x}\implC \PhiAppl{\r}\ (\termf\ \x)\ (\termg\ \x)\land\PredPhi{B}{(\termf\ \x)}\land\PredPhi{B}{(\termg\ \y)}
simplify to
$\univQuant{x}{\PhiAppl{\A}}\PredPhi{A}{\x}\implC \PhiAppl{\r}\ (\termf\ \x)\ (\termg\ \x)\land\PredPhi{B}{(\termf\ \x)}\land\PredPhi{B}{(\termg\ \x)}$.
This justifies adopting axiom \eqref{quotAx}.
The simplification uses the substitution lemma from the appendix, the well-definedness of $\r$, and the transitivity of $\PhiAppl{\r}$.
% In the presence of A* x y and B* (g x) (g x), the substitution properties (as proved in the appendix) yield B* (g x) (g y).
% Due to the well-definedness of r, we have that B* (g x) (g y) implies \bar{r} (g x) (g y).
% We can then apply transitivity of \bar{r}.

%\left(\univQuant{\x,\y}{\PhiAppl{\A}}\termEqT{A}{\x}{\y}\implC\termEqT{B}{(\termf\ \x)}{(\termg\ \y)}\right)
%\lor \left(\univQuant{x}{\PhiAppl{\A}}\PredPhi{A}{\x}\implC \PhiAppl{\r}\ (\termf\ \x)\ (\termg\ \x)\land\PredPhi{B}{(\termf\ \x)}\land\PredPhi{B}{(\termg\ \y)}\right)
%}\]
%and Def.~\ref{ax:quotcod} is justified because both are equivalent to
%\[\univQuant{\x}{\PhiAppl{\A}} \PredPhi{A}{\x}\implC \termEqT{B}{(\termf\ \x)}{(\termg\ \x)}\]
\end{example}

\begin{example}[PERs for a Refined Domain]
We calculate the PERs for both sides of the law for refined domains in Thm.~\ref{lem:piRQ}:
\[\termEqT{\left(\piType{x}{\subtype\A\p}\typB\right)}{\termf}{\termg}=
	\univQuant{\x,\y}{\PhiAppl{\A}}\termEqT{A}{\x}{\y}\land\PhiAppl{\p}\ \x\land\PhiAppl{\p}\ \y\implC \termEqT{B}{(\termf\ \x)}{(\termg\ \y)}
%=\univQuant{\x}{\PhiAppl{\A}}\PredPhi{A}{\x}\land\PhiAppl{\p}\ \x\implC \termEqT{B}{(\termf\ \x)}{(\termg\ \x)}
\]\vspace*{-1.em}
\[\mathll{\termEqT{\left(\quot{(\piType x\A\typB)}
		{\lambdaFun{f,g}{\piType{x}{\A}\typB}\univQuant{x}{\A}\p\ \x\implC \left(\termf\ \x\right)\termEquals{\typB}\left(\termg\ \x\right)}\right)}{\termf}{\termg} =\\
	\univQuant{x}{\PhiAppl{\A}}
	\left(\PredPhi{A}{\x}\implC \PhiAppl{\p}\ \x\implC \termEqT{B}{(\termf\ \x)}{(\termg\ \x)}\right)\land \nonumber\\
	\left(\univQuant{\x,\y}{\PhiAppl{\A}}\termEqT{A}{\x}{\y}\implC
	\termEqT{B}{(\termf\ \x)}{(\termf\ \y)} \right)\land \nonumber\\
	\left(\univQuant{\x,\y}{\PhiAppl{\A}}\termEqT{A}{\x}{\y}\implC
	\termEqT{B}{(\termg\ \x)}{(\termg\ \y)} \right)
}\]
These are indeed not equivalent in line with our observation from Ex.~\ref{ex:refdom}.
\end{example}

Like in \cite{RRB:dhol:23}, this translation yields a sound (defined as in the previous paper) and complete theorem prover:
\renewcommand{\dedPT}[1][]{\ensuremath{\vdash^{#1}_{\PhiAppl{\theorycolor{T}}}}}
\begin{theorem}[Completeness]\label{thm:completePaper}
We have the invariants from Fig.~\ref{fig:complete}.
%The translations of all derivable DHOL judgments hold in HOL:
%	\setlength\extrarowheight{5pt}

\begin{figure}
\begin{tabular}{|l@{\;\;}|@{\;\;}l@{\tb and \tb}l|}
			\hline 
			if in \dhol & \multicolumn{2}{c|}{then in HOL} \\ 
			\hline 
			$\phantom{\ctx}\ded \Thy{\theorycolor{T}}$ & \multicolumn{2}{l|}{$\phantom{\PhiAppl{\ctx}}\ded\Thy{\PhiAppl{\theorycolor{T}}}$}  \\ 
			$\phantom{\ctx}\dedT \Ctx{\ctx}$ &  \multicolumn{2}{l|}{$\phantom{\PhiAppl{\ctx}}\dedPT\Ctx{\PhiAppl{\ctx}}$}  \\ 
			$\ctx\dedT \Type{\A}$ & $\PhiAppl{\ctx}\dedPT\Type{\PhiAppl{\A}}$ & $\PhiAppl{\ctx}\dedPT\PredPhiName{\A}: \PhiAppl{\A}\to\PhiAppl{\A}\to \bool$ and $\PredPhiName{\A}$ is a PER \\ 
			$\ctx\dedT \A \typeEquals \typB$ & $\PhiAppl{\ctx}\dedPT\PhiAppl{\A} \typeEquals \PhiAppl{\typB}$ & $\concatCtx{\PhiAppl{\ctx}}{\varname{x},\y\ofT\PhiAppl{\A}}\dedPT\PredPhiName{\A}\ \varname{x}\ \y \termEqB \PredPhiName{\typB}\ \varname{x}\ \y$ \\ 
			$\ctx\dedT \A \subtyping \typB$ & $\ctxT\dedPT \PhiAppl{\A} \typeEquals \PhiAppl{\typB}$ & $\concatCtx{\ctxT}{\x ,\y\ofT\PhiAppl{\typB}}\dedPT \termEqT{A}{\x}{\y}\impl \termEqT{B}{\x}{\y}$\\
			%$\ctx\dedT \A \subtyping \typB$ & $\PhiAppl{\ctx}\dedPT \PhiAppl{\A} \typeEquals \PhiAppl{\typB}$ & $\concatCtx{\PhiAppl{\ctx}}{\varname{x},\y\ofT\PhiAppl{\A}}\dedPT \termEqT{\A}{\varname{x}}{\y}\impl \termEqT{\typB}{\varname{x}}{\y}$ \\ 
			$\ctx\dedT \tm\ofT\A$ & $\PhiAppl{\ctx}\dedPT\PhiAppl{\tm}\ofT\PhiAppl{\A}$ & $\PhiAppl{\ctx}\dedPT\PredPhi{\A}{\PhiAppl{\tm}}$ \\ 
			$\ctx\dedT \termF$ & \multicolumn{2}{l|}{$\PhiAppl{\ctx}\dedPT\PhiAppl{\termF}$} \\
			\hline 
\end{tabular}
\caption{Invariants of the Translation}\label{fig:complete}
\end{figure}
%	Furthermore, the typing relations $\PredPhiName{A}$ are symmetric and transitive on all well-formed types $A$ and the substitution lemma holds, i.e.,
%	\begin{align*}
%		\concatCtx{\ctx}{\varname{x}\ofT\A}\dedT \tm\ofT\typB\;\Mand\;\ctx\ded \termC{u}\ofT \A &\Impl \PhiAppl{\ctx}\dedPT\PhiAppl{\subst{\tm}{x}{\termC{u}}}\termEquals{\PhiAppl{\typB}} \subst{\PhiAppl{\tm}}{x}{\PhiAppl{\termC{u}}}.
%	\end{align*}
%	Given some technical assumption we also have:
%	\begin{align*}
%		\concatCtx{\ctx}{\x \ofT\A}\dedT \tm\ofT\typB &\Impl \concatCtx{\concatCtx{\ctxT}{\x ,\xp \ofT\PhiAppl{\A}}}{\namedass{xRx'}{\termEqT{A}{\x}{\xp }}} \dedPT \termEqT{B}{\PhiAppl{\tm}}{\subst{\PhiAppl{\tm}}{x}{\xp }}
%	\end{align*}
\end{theorem}
\begin{proof}
The subtyping claim is a slightly strengthened version of the claim obtained by expanding the definition of $\subtyping$.
The proof, given in Appendix~\ref{appendix:Complete}, adapts the proof from \cite{RRB:dhol:23} with additional cases for new productions and rules.
\end{proof}

%The case for subtyping in Thm.~\ref{thm:completePaper} provides a criterion for which subtyping instances should hold. 
%Computing the PERs for the types in Def.~\ref{ax:quotcod} does indeed yield the same results. 
%Computing the PERs for both sides of the subtyping law for a refined domain, yields different results, in line with our observation from Ex.~\ref{ex:refdom}. 

As in \cite{RRB:dhol:23}, the converse of Thm.~\ref{thm:completePaper} is much harder to state and prove.
First we need a technical assumption: We call a type symbol $\tpdeclname{a}$ \emph{inhabited} if at least one of its instances is provably non-empty.

\begin{theorem}[Soundness]\label{thm:soundPaper}
In a well-formed DHOL-theory $\ded\Thy{\thy}$ in which every type symbol is inhabited:%
%	Further assume that for all base types $\a:\piType{x_1}{\typeC{A_1}}\ldots \piType{x_n}{\typeC{A_n}}\type$ declared in $\thy$ the type $\a\ \termC{t_1}\ \ldots\ \termC{t_n}$ is nonempty for at least \emph{some} choice of arguments $\termC{t_i}$.
	\begin{align*}\text{If}\;\;\ctx\dedT[DHOL] \termF:\bool \;\;\color{black}\text{and}\;\; \ctxT\dedPT[HOL]\PhiAppl{\termF}, \;\;\color{black}\text{then}\;\; \ctx\dedT[DHOL]\termF\end{align*}
%	In particular, if $\ctx\dedT \s\ofT\A$ and $\ctx\dedT \tm:\A$ and $\ctxT\dedPT \termEqT{A}{\PhiAppl{\s}}{\PhiAppl{\tm}}$, then $\ctx\ded \s \termEquals{\A}\tm$.
\end{theorem}
\vspace{-1.5em}
\begin{proof}
The key idea is to transform a HOL-proof of $\PhiAppl{\termF}$ into one that is in the image of the translation, at which point we can read off a \dhol-proof of $\termF$.
The full proof is given in Appendix~\ref{appendix:Sound}.
We expect the inhabitation requirement to be redundant, but have not been able to complete the proof without it yet.
In any case, it is harmless because it is satisfied by all practical examples.
\end{proof}

It is now straightforward to extend the DHOL implementation we gave in \cite{RRB:dhol:23}:
First run a bidirectional type-checker for DHOL, using the subtyping-checker sketched in Sect.~\ref{sec:interact}, to establish well-typedness of theory and conjecture.
Then translate conjecture and generated proof obligations and apply a HOL ATP.

\section{Application to Typed Set Theory}\label{sec:soft}
\newcommand{\isList}{\ensuremath{\termC{\mathrm{isList}}}\xspace}
\newcommand{\elem}{\ensuremath{\termC{elem}}\xspace}
\newcommand{\isSubset}{\ensuremath{\termC{\subseteq}}\xspace}
\newcommand{\tuple}{\ensuremath{\termC{tuple}}\xspace}
\newcommand{\card}{\ensuremath{\termC{card}}\xspace}
\newcommand{\pair}{\ensuremath{\termC{pair}}\xspace}
\newcommand{\tpair}{\ensuremath{\termC{tpair}}\xspace}
\newcommand{\fun}{\ensuremath{\termC{fun}}\xspace}
\newcommand{\lam}{\ensuremath{\termC{lam}}\xspace}
\newcommand{\tlam}{\ensuremath{\termC{tlam}}\xspace}
\newcommand{\isTuple}{\ensuremath{\termC{isTuple}}\xspace}
\newcommand{\cartProd}{\ensuremath{\typeC{Prod}}\xspace}
\newcommand{\Relats}{\ensuremath{\typeC{Rels}}\xspace}
\newcommand{\Functs}{\ensuremath{\typeC{Functions}}\xspace}
\newcommand{\Injs}{\ensuremath{\typeC{Injs}}\xspace}

DHOL with subtyping enables a novel formalization of typed set-theory: 
\[\set\ofT\,\type,\tb\in\ofT\,\set\to\set\to\bool,\tb\elem\ \s:=\subtype\set{\lambdaFun x\set\x\in\s}\]
%Here we use advanced features of the MMT framework \cite{rabe:recon:17}, in which DHOL (with refinements but not yet with quotients) is implemented, namely definitions, infix notations, and type inference.
The key idea is that $\elem\,\s$ is the DHOL \emph{type} of set-theoretical elements of the \emph{set} $\s$.
Leveraging that refinements and quotients do not require change of representation, we obtain a powerful combination of elegant high-level typed formalization and efficient low-level reasoning.
All the routine constructions of untyped set theory can be lifted to their typed counterparts.
For example, for products, we use $\times\ofT\set\to\set\to\set$ and $\pair\ofT\set\to\set\to\set$ and the property $\assert\;\univQuant{x,y,s,t}\set(\x\in\s)\wedge(\y\in\tm)\implC\pair\ \x\ \y\in \s\times\tm$, from which we can show that $\pair\;\ofT\; \elem\ \s\to\elem\ \tm\to\elem\ (s\times t)$.

%For functions, the usual development of untyped set theory yields
%\[\mathll{
%  \fun\ofT\set\to\set\to\set,\tb\lam\ofT\set\to(\set\to\set)\to\set,\\
%    \assert\;\univQuant{s,t}\set\univQuant\termf{\set\to\set}
%      (\univQuant x\set\x\in\s  \implC  (\termf\ \x)\in\tm)\implC(\lam\ \s\ \termf)\in\fun\ \s\ \tm
%}\]
%where $\lam$ uses higher-order abstract syntax to declare the $\lambda$-binder of set theory, i.e.,
%the function $\s\ni\x\mapsto \tm$ is represented as $\lam\ \s\ (\lambdaFun x\set\tm)$.
%But this formalization is unsatisfactory because we need to define $\tm$ for all sets even though we only need it for elements of $\s$.
%This led us to adopt much more involved constructions in \cite{IR:foundations:10}.
%We can now simplify those by constructing instead
%\[\tlam\ofT\piType{s,t}\set(\elem\ \s\to\elem\ \tm)\to\elem\ (\fun\ \s\ \tm) \tb\]
%which can be defined in set theory but not in terms of $\lam$.

We can also use DHOL-functions $\set\to\set$ as set-theoretical functions between sets $\s$ and $\tm$ without a change in representation as the type
$\Functs\,\s\,\tm:= \quot{\subtype{(\set\to\set)}{\p}}\r$
where
$\p\ \termf=\univQuant x\set \x\in\s\implC(\termf\ \x)\in\tm$ and
$\r\ \termf\ \termg=\univQuant x\set \x\in\s\implC(\termf\ \x)\termEquals\set(\termg\ \x)$.
This allows us to represent set-theoretical function application and composition $\circ$ directly as DHOL application/composition.
%\[
%\termf.\tm:= \termf \tm;\tb
%\termf \circ \termg:= \lambdaFun{x}{\set} (\termf (\termg \x)).
%\]

%We can then define function application in the usual way.
% $\circ$ in the usual way, leading to the conjecture (to be proven by the ATP) that the composition of functions from $\s$ to $\tm$ and from $\tm$ to $\termC{u}$ yields a function from $\s$ to $\termC{u}$:
%\[
%\assert \univQuant{s,t,u}{\set}\univQuant{f}{\Functs\ \varname{s}\ \varname{t}}\univQuant{g}{\Functs\ \varname{t}\ \varname{u}}\univQuant{x}{\set}\x\in\s\implC((\varname{g}\ \circ\ \varname{f})\ \x)\in \varname{u}
%\]
%%%
%\[
%\set\ofT\type,\tb\in\ofT\set\to\set\to\bool,\tb\elem\ \s:=\subtype\set{\lambdaFun x\set\x\in\s}
%\]
%
%We can use DHOL-functions $\set\to\set$ as set-theoretical functions between sets $\s$ and $\tm$ without a change in representation: we define the type
%$\Functs\,\s\,\tm:= \quot{\subtype{(\set\to\set)}{\p}}\r$
%where
%$\p\ \termf=\univQuant x\set \x\in\s\implC(\termf\ \x)\in\tm$ and
%$\r\ \termf\ \termg=\univQuant x\set \x\in\s\implC(\termf\ \x)\termEquals\set(\termg\ \x)$.

Consequently, theorem proving in typed set theory becomes very strong because a large share of the proving workload can be outsourced into typing-obligations.
For example, the property that the composition of functions $f:\Functs\ \s\ \tm$ and $g:\Functs\ \tm\ \termC{u}$ has type $\Functs\ \s\ \termC u$ becomes
\[
\assertL \univQuant{s,t,u}{\set}\univQuant{f}{\Functs\ \varname{s}\ \varname{t}}\univQuant{g}{\Functs\ \varname{t}\ \varname{u}}\univQuant{x}{\set}\x\in\s\implC((\varname{g}\ \circ\ \varname{f})\ \x)\in \varname{u}
\]
which yields in HOL the conjecture below that current HOL ATPs solve easily.
\newcommand{\setrel}{\ensuremath{\termC{set\_rel}}\xspace}
\newcommand{\termh}{\ensuremath{\termC{h}}\xspace}
\begin{align*}
	&\assert \univQuant{s}{\set}\setrel\ \s\ \s\implC \univQuant{t}{\set}\setrel\ \tm\ \tm\implC \univQuant{u}{\set}\setrel\ \termC{u}\ \termC{u}\implC\\
	&\qquad \univQuant{f}{\set\to\set}
	\univQuant{x}{\set}x\in \s\implC \setrel (\termf\ \x) (\termf\ \x)%
	\land (\univQuant{x}{\set}x\in \s\implC (\termf\ \x)\in \tm) \implC\\
	&\qquad \univQuant{g}{\set\to\set} %
	\univQuant{x}{\set}x\in \tm\implC \setrel (\termg\ \x) (\termg\ \x)%
	\land (\univQuant{x}{\set}x\in \tm\implC (\termg\ \x)\in  \termC{u}) \implC\\
	&\qquad \univQuant{x}{\set}\setrel\  \termC{u}\implC%
	\x\in\s\implC((\termg\ (\termf\ \x))\in \varname{u})
\end{align*}
%Writing this translation into a TPTP problem file yields a TPTP problem easily solved by all 7 ATPs recommended for the problem on the System on TPTP website (and at least 5 others)\protect\footnote{\url{https://tptp.org/cgi-bin/SystemOnTPTP}}. 
%12 out of 15 of the HOL ATPs supported at the Systems on TPTP website (all but CVC-5-SAT, Vampire-FMo and Isabelle, including plain CVC-5 and Vampire can solve the problem, this also includes all recommended systems for the problem on that website).
Below is the HOL translation of the conjecture that function composition is associative. It is similarly easily proved by current ATPs:
\begin{align*}
	&\assert \univQuant{s}{\set}\setrel\ \s\ \s\implC \univQuant{t}{\set}\setrel\ \tm\ \tm\implC \univQuant{u}{\set}\setrel\ \termC{u}\ \termC{u}\implC \univQuant{v}{\set}\setrel\ \termC{v}\ \termC{v}\implC\\
	&\qquad \univQuant{f}{\set\to\set}
	\univQuant{x}{\set}x\in \s\implC \setrel (\termf\ \x) (\termf\ \x)%
	\land (\univQuant{x}{\set}x\in \s\implC (\termf\ \x)\in \tm) \implC\\
	&\qquad \univQuant{g}{\set\to\set} %
	\univQuant{x}{\set}x\in \tm\implC \setrel (\termg\ \x) (\termg\ \x)%
	\land (\univQuant{x}{\set}x\in \tm\implC (\termg\ \x)\in  \termC{u}) \implC\\
	&\qquad \univQuant{h}{\set\to\set} %
	\univQuant{x}{\set}x\in \termC{u}\implC \setrel (\termh\ \x) (\termh\ \x)%
	\land (\univQuant{x}{\set}x\in \tm\implC (\termh\ \x)\in  \termC{v}) \implC\\
	&\qquad \univQuant{x}{\set}\setrel\ \termC{x}\ \termC{x}\implC%
	\x\in\s\implC( \setrel\ (\termh\ (\termg\ (\termf\ \x))) \ (\termh\ (\termg\ (\termf\ \x))))
\end{align*}
The corresponding TPTP files are available at
\url{https://gl.mathhub.info/MMT/LATIN2/-/tree/master/source/casestudies/2025-FroCos}.

\section{Conclusion and Future Work}\label{sec:conc}
DHOL combines higher-order logic with dependent types, obtaining an intuitive and expressive language, albeit with undecidable typing.
We double down on this design by elegantly extending DHOL with two practically important type constructors that thrive in that setting: refinement and quotient types.
Like dependent function types, these two require terms occurring in types.
Both are near-impossible to add as an afterthought to a type theory with decidable typing.

We translate the resulting logic to HOL, obtaining a practical automated theorem proving workflow for DHOL with refinement and quotient types.
Our main result is the proof of soundness and completeness of this translation.
%Critically, this translation maps every DHOL-type to a HOL-type with a partial equivalence relation (PER) on it.
%Because PERs are closed under refinements and quotients, it is feasible to adapt the existing translation as well as the soundness/completeness proof to the corresponding results for our extended DHOL.

We used an extensional subtyping approach, where $\A\subtyping\typB$ holds iff all $\A$-terms also have type $\typB$.
This allows combining typed representations and efficient reasoning.
We established all the expected variance and normalization laws except for function types with refined domains.
Future work must investigate how to improve on this to make normalizing types and thus subtype-checking simpler.

%Extensions of DHOL with choice operators have already been introduced in \cite{LPAR2024}.
%In the presence of refinement types and choice operators sigma types become definable. This extension also remains for future work.

We also want to carry our results for DHOL over to existing refinement/quotient type systems for programming languages like Quotient Haskell, where DHOL-like axioms are used as lightweight specifications.
%Like in our work on DHOL, they use a translation to obtain ATP support.
%But unlike those systems which typically prioritizes proof obligations that are efficiently checkable by SMTs, DHOL focuses on rigorously working out the general case. Furthermore, our soundness proof enables proof reconstruction and checking, whereas the trusted codebase of refinement type systems typically includes an entire SMT solver.
%A combination of these advantages is so far lacking.

\bibliographystyle{splncs04}

%\bibliography{biblio,pub_rabe,rabe,systems}

\newpage
\begin{appendix}
	\renewcommand{\namedRule}[3]{\refnamedRule{#1}{#2}{#3}}
\renewcommand{\rnamedRule}[4]{\refrnamedRule{#1}{#2}{#3}{#4}}
\renewcommand{\snamedRule}[3]{\refsnamedRule{#1}{#2}{#3}}
\setbool{inAppendix}{true}

\section{Summary of logics and translations}\label{appednix:logics}
In this section we collect the inference rules of the logics and the definition of the overall translation.
We name the rules and enumerate the cases in the definition of the translation for reference in the proofs in the subsequent appendices.

\subsection{\hol{} rules}
The rules for HOL are given in Fig.~\ref{fig:holrulesAppendix}.

\begin{figure}[hpt]
	{\ifbool{inAppendix}{\scriptsize}{\small}
		Theories and contexts:
		\[
		\snamedRule{thyEmpty}{\ded\Thy{\emptyThy}}{}\tb
		\snamedRule{thyType}{\ded\Thy{\concatThy{\theorycolor{T}}{\Type{\A}}}}{\ded\Thy{\theorycolor{T}}}\tb
		\snamedRule{thyConst}{\ded\Thy{\concatThy{\theorycolor{T}}{\constname{c}\ofT\A}}}{\dedT \Type{\A}}\tb
		\snamedRule{thyAxiom}{\ded\Thy{\concatThy{\theorycolor{T}}{\namedax{c}{\termF}}}}{\dedT \termF\ofT \bool}
		\]
		\[
		\snamedRule{ctxEmpty}{\dedT\Ctx{\emptyCtx}}{\ded\Thy{\theorycolor{T}}}\tb
		\snamedRule{ctxVar}{\dedT\Ctx{\concatCtx{\ctx}{\x\ofT \A}}}{\ctx\dedT \Type{\A}}\tb
		\snamedRule{ctxAssume}{\dedT\Ctx{\concatCtx{\contextcolor{\ctx}}{\namedass{x}{\termF}}}}{\contextcolor{\ctx}\dedT \termF\ofT\bool}
		\]\\
		Lookup in theory and context:
		\[
		\snamedRule{type}{\ctx \dedT \Type{\A}}{\thyIn{\A\ofT\type}{\theorycolor{T}} \tb \dedT\Ctx{\ctx}}\tb
		\rnamedRule{const}{\ctx\dedT \constname{c}\ofT \A}{\thyIn{\constname{c}\ofT \Ap}{\theorycolor{T}} \tb\ctx\dedT \Ap\typeEquals \A}{const}\tb
		\snamedRule{axiom}{\ctx\dedT F}{\thyIn{\namedax{c}{\typeC{F}}}{\theorycolor{T}} \tb  \dedT\Ctx{\ctx}}
		\]
		\[
		\rnamedRule{var}{\ctx\dedT \x\ofT \A}{\ctxIn{\x\ofT \Ap}{\ctx} \tb\ctx\dedT \Ap\typeEquals \A}{var} \tb
		\snamedRule{assume}{\ctx\dedT \termF}{\ctxIn{\namedass{x}{\termF}}{\ctx} \tb  \dedT\Ctx{\ctx}}
		\]\\
		Well-formedness and equality of types:
		\[
		\snamedRule{bool}{\ctx\dedT \Type{\bool}}{\dedT\Ctx\ctx}\tb
		\snamedRule{arrow}{\ctx\dedT \Type{\A\to \typB}}{\ctx\dedT \Type{\A} \tb \ctx\dedT \Type{\typB}}
		\]
		\[
		\snamedRule{congBase}{\ctx\dedT \A\typeEquals \A}{\Ctx{\ctx}\quad \thyIn{\tpdeclname{a}:\type}{\theorycolor{T}}} \tb
		\rnamedRule{cong$\bool$}{\ctx\dedT \bool\typeEquals \bool}{\dedT\Ctx{\ctx}}{congB} \tb
		\rnamedRule{cong$\to$}{\ctx\dedT \A\to \typB\typeEquals \Ap\to \Bp}{\ctx\dedT \A\typeEquals \Ap \tb  \ctx\dedT \typB\typeEquals \Bp}{congTo}
		\]\\
		Typing: 
		\[
		\snamedRule{lambda}{\ctx\dedT (\lambdaFun{x}{\A} \tm)\ofT \A\to \typB}{\concatCtx{\ctx}{\x\ofT \A}\dedT \tm\ofT \typB}\tb
		\snamedRule{appl}{\ctx\dedT \termf\ \tm\ofT\typB }{\ctx\dedT \termf\ofT \A\to \typB \tb \ctx\dedT \tm\ofT \A}\tb
		\rnamedRule{$=$type}{\ctx\dedT \s\termEquals{\A}\typeC{t}\ofT\bool}{\ctx\dedT \s\ofT\A\tb \ctx\dedT \tm\ofT \A}{eqType}
		\]\\
		Term equality, congruence, reflexivity, symmetry, $\beta$, $\eta$:
		\[
		\namedRule{cong$\lambda$ (xi)}{\ctx\dedT \lambdaFun{x}{\A} \tm\termEquals{\A\to \typB} \lambdaFun{x}{\Ap} \typeC{t'}}{\ctx\dedT \A\typeEquals \Ap \tb\concatCtx{\ctx}{\x\ofT \A}\dedT \tm\termEquals{B} \termtp}\rulelabelAppendix{congLam}{cong$\lambda$}
		\tb
		\snamedRule{congAppl}{\ctx\dedT \termf\ \tm\termEquals{B} \termfp\ \termtp}{\ctx\dedT \tm\termEquals{\A} \termtp\tb  \ctx\dedT \termf\termEquals{\A\to \typB} \termfp}
		\]
		\[
		\snamedRule{refl}{\ctx\dedT \tm\termEquals{\A} \tm}{\ctx\dedT \tm\ofT \A}\tb
		\snamedRule{sym}{\ctx\dedT \s\termEquals{\A} \tm}{\ctx\dedT \tm \termEquals{\A}\s}\tb
		\snamedRule{beta}{\ctx\dedT (\lambdaFun{x}{\A} \s)\ \tm \termEquals{B} \subst{\s}{x}{\tm} }{\ctx\dedT (\lambdaFun{x}{\A} \s)\ \tm\ofT \typB}\tb
		\snamedRule{eta}{\ctx\dedT \tm\termEquals{\A\to \typB} \lambdaFun{x}{\A}\tm\ \x}{\ctx\dedT \tm\ofT \A\to \typB\quad\ctxIn{\x\text{ not}}{\ctx}}
		\]\\
		Rules for implication:
		\[
		\rnamedRule{$\impl$type}{\ctx\dedT \termF\Rightarrow \termC{G}\ofT\bool}{\ctx\dedT \termF\ofT\bool\tb \ctx\dedT \termC{G}\ofT\bool}{implType}\tb
		\rnamedRule{$\impl$I}{\ctx\dedT \termF\Rightarrow \termC{G}}{\ctx\dedT \termF\ofT\bool \tb \concatCtx{\ctx}{\namedass{x}{\termF}}\dedT \termC{G}}{implI}\tb
		\rnamedRule{$\impl$E}{\ctx\dedT \termC{G}}{\ctx\dedT \termF\Rightarrow \termC{G}\tb \ctx\dedT \termF}{implE}
		\]
		%\[
		%\snamedRule{the}{\concatCtx{\ctx}{\y\ofT \A}\dedT \choiceOp{x}{A}{p}\ofT \A}{\ctx\dedT \p\ofT \A\to\bool}\tb
		%\snamedRule{theE}{\ctx\dedT \p\ \left(\choiceOp{x}{A}{p}\right)}{\ctx\dedT \left(\choiceOp{x}{A}{p}\right)\ofT \A}
		%\]
		Congruence for validity, Boolean extensionality, and non-emptiness of types:
		\[
		\rnamedRule{cong$\ded$}{\ctx\dedT \termF}{\ctx\dedT \termF\termEquals{\bool} \termC{F'}\quad \ctx\dedT \termC{F'}}{congDed}\tb
		\snamedRule{boolExt}{\ctx,\x\ofT \bool\dedT \p\ \x}{\ctx\dedT \p\ \T \quad \ctx\dedT \p\ \F}\tb
		\snamedRule{nonempty}{\ctx\dedT \termF}{\ctx\dedT \termF\ofT\bool\quad \ctx,\x\ofT \A\dedT \termF}
		\]
		
		In the soundness proof, we will occasionally use the existence of a HOL term of given type $\A$ (whose existence follows from rule \ruleRef{nonempty}), so we denote this term by $\defaultTerm{\A}$.
	}
	\caption{\hol{} Rules}
	\ifbool{inAppendix}{\vspace*{-.5cm}\label{fig:holrulesAppendix}}{\label{fig:holrulesPaper}}
\end{figure}

\clearpage

\subsection{Derived rules}
Using the rules given in Figure~\ref{fig:holrulesAppendix} we can derive a number of additional useful rules.
%\subsection{Admissible rules for HOL}
\label{sec:meta-thm:1}
%In the proofs in the appendix (especially but not only) in this part, it will be useful to name the inference rules for HOL, DHOL and \dhole. 
%We therefore collect them again, here:

The following lemma collects a few routine meta-theorems that we make use of later on:
\begin{lemma}\label{meta-thm:1}
	Given the inference rules for HOL (cfg. Figure~\ref{fig:holrulesAppendix}), the following rules are admissible:
	{\small
		\[
		\snamedRule{ctxThy}{\ded \Thy{T}}{\dedT \Ctx{\ctx}}\tb
		\snamedRule{tpCtx}{\dedT \Ctx{\ctx}}{\ctx\dedT \Type{\A}}\tb
		\snamedRule{typingTp}{\ctx\dedT \Type{\A}}{\ctx\dedT \tm\ofT\A}\tb
		\snamedRule{validTyping}{\ctx\dedT \termF\ofT\bool}{\ctx\dedT \termF}\tb
		\]
		\[
		\rnamedRule{constS}{\ctx\dedT \constname{c}\ofT\A}{\thyIn{\constname{c}\ofT\A}{T}}{constS}\tb
		\rnamedRule{varS}{\ctx\dedT \x\ofT\A}{\ctxIn{\x\ofT\A}{\ctx}}{varS}
		\]
		\[
		\rnamedRule{$\typeEquals$refl}{\ctx\dedT \A\typeEquals \A}{\ctx\dedT \Type{\A}}{tpEqRefl}\tb
		\rnamedRule{$\typeEquals$sym}{\ctx\dedT \Ap\typeEquals \A}{\ctx\dedT \A\typeEquals \Ap}{tpEqSym}\tb 
		\rnamedRule{$\typeEquals$trans}{\ctx\dedT \A\typeEquals \typeC{A''}}{\ctx\dedT \A\typeEquals \Ap\tb \ctx\dedT \Ap\typeEquals \typeC{A''}}{tpEqTrans}
		\]
		\[
		\snamedRule{eqTyping}{\ctx\dedT \s\ofT \A}{\ctx\dedT \s\termEquals{\A}\tm}\tb
		\snamedRule{implTypingL}{\ctx\dedT \termF\ofT\bool}{\ctx\dedT \termF\implC \termC{G}}\tb
		\snamedRule{implTypingR}{\ctx\dedT \termC{G}\ofT\bool}{\ctx\dedT \termF\implC \termC{G}}\]
		\[
		\snamedRule{typesUnique}{\ctx\dedT \A\typeEquals \Ap}{\ctx\dedT \s\ofT\A\tb \ctx\dedT \s\ofT\A'}
		%\]
		\tb
		%\[
		\rnamedRule{typingWf}{\ctx\dedT \tm\ofT\A}{\ctx\dedT \termf\ \tm\ofT\typB \quad \ctx\dedT \termf\ofT\A\to \typB}{typingWellformed}
		\]
		\[
		\snamedRule{applType}{\ctx\dedT \termf\ofT\A\to \typB}{\ctx\dedT \tm\ofT\A\tb \ctx\dedT \termf\ \tm\ofT\typB}\tb
		\snamedRule{rewriteTyping}{\ctx\dedT \subst{\s}{x}{\tm}\ofT\A}{\concatCtx{\ctx}{\x\ofT\typB}\dedT \s\ofT\A\quad \ctx\dedT \tm\ofT\typB}
		\]
		\[
		\rnamedRule{monotonic$\ded$}{\concatCtx{\ctx}{\namedass{ass}{\termF}}\dedT \termC{G}}{\ctx\dedT \termF\ofT\bool\quad\ctx\dedT \termC{G}}{assDed}
		\quad
		\rnamedRule{var$\ded$}{\concatCtx{\ctx}{\x\ofT\A}\dedT J}{\ctx\dedT \Type{\A}\quad\ctx\dedT J\tb \text{ for any statement }\dedT J}{varDed}
		\]
		\[
		\rnamedRule{$\forall$type}{\ctx\dedT \univQuant{x}{\A}\termF\ofT\bool}{\concatCtx{\ctx}{\x\ofT\A}\dedT \termF\ofT\bool}{forallType}\tb
		\rnamedRule{$\forall$I}{\ctx\dedT \univQuant{x}{\A}\termF}{\concatCtx{\ctx}{\x\ofT\A}\dedT \termF}{forallI}\tb
		\rnamedRule{$\forall$E}{\ctx\dedT \subst{\termF}{x}{\tm}}{\ctx\dedT \univQuant{x}{\A}\termF \quad \ctx\dedT \tm\ofT\A}{forallE}
		\]
		\[	
		\snamedRule{assTyping}{\ctx\dedT \termF\ofT\bool}{\Ctx{\ctx}\quad \ctxIn{\termF}{\ctx}}\tb
		\tb
		\rnamedRule{cong$\ofT\!$}{\ctx\dedT \termtp\ofT\A'}{\ctx\dedT \tm\termEquals{\A} \termtp \quad  \ctx\dedT \A\typeEquals \Ap \quad  \ctx\dedT \tm\ofT\A'}{congColon}
		\]
		\[
		\rnamedRule{$=\T$}{\judg \termF}{\judg \termF\termEqB \T}{eqTrue}\tb
		\rnamedRule{$\T=$}{\judg \termF\termEqB \T}{\judg \termF}{trueEq}\tb
		\snamedRule{propExt}{\ctx\dedT \termF\termEquals{\bool}\termC{G}}{\concatCtx{\ctx}{\namedass{ass}{\termF}}\dedT \termC{G}\quad \concatCtx{\ctx}{\namedass{ass_G}{\termC{G}}}\dedT \termF}\]
		\[\snamedRule{extensionality}{\ctx\dedT f\termEquals{\A\to \typB}f'}{\concatCtx{\ctx}{\x\ofT\A}\dedT \termf\ \x\termEquals{B} \termfp\ \x\tb \judg \termf\ofT\A\to \typB\tb \judg \termfp\ofT\A\to \typB}
		\]
		\[\snamedRule{trans}{\ctx\dedT \s\termEquals{\A}\termC{u}}{\ctx\dedT \s\termEquals{\A}\tm\tb \ctx\dedT \tm\termEquals{\A}\termC{u}}\tb
		\rnamedRule{$\termEquals{}$cong}{\ctx\dedT (\s\termEquals{\A}\tm) \termEqB (\sp\termEquals{\A}\termtp)}{\ctx\dedT \s\termEquals{\A}\sp\tb \ctx\dedT \tm\termEquals{\A}\termtp}{termEqCong}\tb
		\]
		\[
		\rnamedRule{$\forall$cong}{\ctx\dedT\univQuant{x}{\A}\termF\termEquals{\bool}\univQuant{x}{A'}\termC{F'}}{\ctx\dedT \A\typeEquals \Ap\tb \concatCtx{\ctx}{\x\ofT\A}\dedT \termF\termEquals{\bool}\termC{F'}}{forallCong}\tb
		\rnamedRule{$\implC$cong}{\ctx\dedT \termF\Rightarrow \termC{G}\termEquals{\bool} \termC{F'}\Rightarrow \termC{G'}}{\ctx\dedT \termF\termEquals{\bool}\termC{F'}\tb \ctx\dedT \termC{G}\termEquals{\bool}\termC{G'}}{implCong}
		\]
		\[
		\rnamedRule{$\forall\implC$}{\ctx\dedT\univQuant{x}{\A}\termF\implC \univQuant{x}{\A}\termC{G}}{\concatCtx{\ctx}{\x\ofT\A}\dedT \termF\implC \termC{G}}{forallImpl}\tb
		\rnamedRule{$\implC$Funct}{\ctx\dedT (\termF\implC \termC{G})\implC \left(\termC{F'}\implC \termC{G'}\right)}{\ctx\dedT \termC{G}\implC \termC{G'}\tb \ctx\dedT \termC{F'}\implC \termF}{implFunctorial}
		\]
		\[
		\rnamedRule{$\ded$cong}{\ctx\dedT \termC{F'}}{\ctx\dedT \termF\termEquals{\bool} \termC{F'}\quad \ctx\dedT \termF}{dedCong}\tb
		\snamedRule{rewrite}{\ctx\dedT \subst{\termF}{x}{\termtp}}{\ctx\dedT \subst{\termF}{x}{\tm}\tb \ctx\dedT \tm\termEquals{\A}\termtp\tb \concatCtx{\ctx}{\x\ofT\A}\dedT \termF\ofT\bool}
		\]
	}
\end{lemma}
This Lemma~\ref{meta-thm:1} is already proven for the version of \hol in the paper\cite{RRB:dhol:23} that originally introduced DHOL. 
Definitions for the existential quantifier and the connectives are also given in \cite{RRB:dhol:23}.
Due to it's prominent appearance, we will repeat the definition of $\forall$:
$$
\univQuant{x}{\A}\termF := \lambdaFun{x}{\A}\termF\termEquals{\piType{x}{\A}\bool}\lambdaFun{x}{\A}\T
$$

%Our version of \hol only add a choice operator. 
%A few of the admissible rules are proven by induction on derivations, but the inductive proofs for those rules can easily be extended with cases for the inference rules for choice operators.

Furthermore, using the definitions of the connectives and quantifiers we can prove the rules:
\[
\rnamedRule{$\land$}{\ctx\dedT \termF\land\termC{G}\ofT\bool}{\ctx\dedT \termF\ofT\bool\tb \ctx\dedT \termC{G}\ofT\bool}{andTp}\tb
\rnamedRule{$\land$Cong}{\ctx\dedT (\termF\land\termC{G})\termEqB (\termC{F'}\land\termC{G'})}{\ctx\dedT \termF\termEqB\termC{F'}\tb \ctx\dedT \termG\termEqB\termC{G'}}{andCong}\]
%\[\rnamedRule{$\lor$Il}{\ctx\dedT \termF\lor \termC{G}}{\ctx\dedT \termF}{orIl}\tb 
%\rnamedRule{$\lor$Ir}{\ctx\dedT \termF\lor \termC{G}}{\ctx\dedT \termC{G}}{orIr}\]
\[
\rnamedRule{$\land$I}{\ctx\dedT \termF\land \termC{G}}{\ctx\dedT \termF\tb \ctx\dedT \termC{G}}{andI}\tb
\rnamedRule{$\land$El}{\ctx\dedT \termF}{\ctx\dedT \termF\land \termC{G}}{andEl}\tb
\rnamedRule{$\land$Er}{\ctx\dedT \termC{G}}{\ctx\dedT \termF\land \termC{G}}{andEr}\tb
%\rnamedRule{$\exists$I}{\ctx\dedT\existQuant{x}{\A} \termF}{\ctx\dedT \tm\ofT\A\tb \ctx\dedT \subst{\termF}{x}{\tm}}{existI}
\]
and similar rules for the other boolean connectives.

\begin{rem}\label{rem:analogueDerivedRulesDHOLP}
	Observe that many of the rules derived for HOL in Lemma~\ref{meta-thm:1} still hold in \dhole. 
	In particular, the rules \ruleRef{ctxThy}, \ruleRef{tpCtx}, \ruleRef{typingTp} and \ruleRef{validTyping} can be proven by the same method. 
	The rules \ruleRef{assDed}, \ruleRef{varDed}, \ruleRef{forallType}, \ruleRef{forallE}, \ruleRef{forallI}, \ruleRef{assTyping}, \ruleRef{eqTrue}, \ruleRef{trueEq}, \ruleRef{propExt}, \ruleRef{extensionality}, \ruleRef{forallCong}, \ruleRef{forallImpl}, \ruleRef{implFunctorial}, \ruleRef{dedCong}, \ruleRef{rewrite} and the introduction and elimination rules for the (dependent) conjunction can be derived in \dhole with the same proofs.
	Also the rules \ruleRef{tpEqRefl} and \ruleRef{tpEqSym} can be proven easily in \dhole by induction on the type equality rules. 
	%(for the rule \ruleRef{tpEqSym} the proof easily generalizes to DHOL but proving it in \dhole is harder). 
\end{rem}

\subsection{\dhol rules}
{\small
	Theories and contexts:
	\[
	\rnamedRule{thyEmpty}{\ded\Thy{\emptyThy}}{}{thyEmpty'}\tb
	\snamedRule{thyType'}{\ded\Thy{\concatThy{\theorycolor{T}}{\a\ofT \piType{x_1}{\typeC{A_1}}\ldots\piType{x_n}{\typeC{A_n}}\type}}}
	{\dedT \Ctx{\varname{x_1}\ofT \typeC{A_1},\,\ldots,\varname{x_n}\ofT \typeC{A_n}}}
	\]
	\[
	\rnamedRule{thyConst}{\ded\Thy{\concatThy{\theorycolor{T}}{\constname{c}\ofT\A}}}{\dedT \Type{\A}}{thyConst'}\tb
	\rnamedRule{thyAxiom}{\ded\Thy{\concatThy{\theorycolor{T}}{\namedax{c}{\termF}}}}{\dedT \termF\ofT \bool}{thyAxiom'}
	\]
	\[
	\rnamedRule{ctxEmpty}{\dedT\Ctx{\emptyCtx}}{\ded\Thy{\theorycolor{T}}}{ctxEmpty'}\tb
	\rnamedRule{ctxVar}{\dedT\Ctx{\concatCtx{\ctx}{\x\ofT \A}}}{\ctx\dedT \Type{\A}}{ctxVar'}\tb
	\rnamedRule{ctxAssume}{\dedT\Ctx{\concatCtx{\contextcolor{\ctx}}{\namedass{x}{\termF}}}}{\contextcolor{\ctx}\dedT \termF\ofT\bool}{ctxAssume'}
	\]\\
	Well-formedness and equality of types:
	\[
	\snamedRule{type'}{\ctx \dedT\ \a\ \termC{t_1}\ \ldots\ \termC{t_n}\;\type}{\begin{array}{c}\a\ofT \piType{x_1}{\typeC{A_1}}\ldots\piType{x_n}{\typeC{A_n}}\type \text{ in }T\\
			\ctx \dedT \termC{t_1}\ofT \typeC{A_1} \ \;\ldots\ \; \ctx \dedT \termC{t_n}\ofT \subst{\typeC{A_n}}{x_1}{\termC{t_1}}\ldots\substOp{x_{n-1}}{\termC{t_{n-1}}}\end{array}}
	\]
	\[
	\rnamedRule{$\subtype{}{p}\type$}{\ctx\dedT \subtype{\A}{\p}\; \type}{\ctx\dedT \p\ofT \piType{x}{\A}\bool}{psubType}\tb
	\rnamedRule{bool}{\ctx\dedT \Type{\bool}}{\dedT\Ctx\ctx}{bool'}\tb
	\snamedRule{pi}{\Gamma\dedT \piType{x}{\A}\typB\ \type}{\ctx\dedT \Type{\A} \quad \concatCtx{\ctx}{\x\ofT \A\dedT \Type{\typB}}}\tb
	\]
	\[
	\rnamedRule{Q}{\ctx\dedT \Type{\quot{\A}{\r}}}{\ctx\dedT \Type{\A}\tb \ctx\dedT \r\ofT \piType{x_1}{\A}\piType{x_2}{\A}\bool\tb \ctx\dedT \isEqRel{\r}}{Qtype}
	\]
	\\Type equality:
	\[
	\snamedRule{congBase'}{\ctx\dedT\a\ \termC{s_1}\ \ldots\ \termC{s_n}\typeEquals \a\ \termC{t_1}\ \ldots \termC{t_n}}
	{\begin{array}{c}\a\ofT \piType{x_1}{\typeC{A_1}}\:\ldots\ \piType{x_n}{\typeC{A_n}}\type\text{ in }\theorycolor{T}\\
			\ctx \dedT \termC{s_1}\termEquals{\typeC{A_1}}\termC{t_1} \ \ldots\ \ctx \dedT \termC{s_n}\termEquals{\subst{\typeC{A_n}}{x_1}{\termC{t_1}}\;\ldots\;\substOp{x_{n-1}}{\termC{t_{n-1}}}}\termC{t_n}\end{array}}\]
	\[
		\rnamedRule{\STequiv}{\ctx\dedT \A\typeEquals\typB}{\ctx\ded \A\subtyping\typB\tb \ctx\ded \typB\subtyping\A}{subtE}
	\]
	\[
	\rnamedRule{$\typeEquals\bool$}{\ctx\dedT \Type{\bool\typeEquals\bool}}{\dedT\Ctx\ctx}{congBool}\tb
	\rnamedRule{cong$\Pi$}{\ctx\dedT\piType{x}{\A} \typB\typeEquals \piType{x}{\Ap}\Bp}{\ctx\dedT \A\typeEquals \Ap \quad  \ctx,\x\ofT \A\dedT \typB\typeEquals \Bp}{congPi}
	\]
	%\rnamedRule{$\subtype{}{p}\!\!\typeEquals$}{\ctx\dedT \subtype{\A}{\p} \typeEquals \subtype{\Ap}{\pp}}{\ctx\dedT \A\typeEquals \Ap \tb \ctx\dedT \p\termEquals{\piType{x}{\A}\bool} \termC{p'}}{psubEq}\tb
	%\rnamedRule{Q$\typeEquals$}{\ctx\dedT \quot{\A}{\r}\typeEquals\quot{\Ap}{\rp}}{\ctx\dedT \A\typeEquals\typB\tb \ctx\dedT \r\termEquals{\piType{x_1}{\A}\piType{x_2}{\A}\bool}\rp}{QTEq}
	%\]
	\\Typing: 
	\[
	\rnamedRule{const'}{\ctx\dedT \constname{c}\ofT \A}{c\ofT \Ap\thyIn{}{T}\tb \ctx\dedT \Ap\typeEquals \A}{const''}\tb
	\rnamedRule{var'}{\ctx\dedT \x\ofT \A}{\x\ofT \Ap\ctxIn{}{\ctx}\tb \ctx\dedT \Ap\typeEquals \A}{var''}
	\]
	\[
	\snamedRule{lambda'}{\ctx\dedT (\lambdaFun{x}{\A} \tm)\ofT  \piType{x}{\Ap}\typeC{ B}}{\concatCtx{\ctx}{\x\ofT \A}\dedT \tm\ofT \typB\tb \ctx\dedT \A\typeEquals\Ap}\tb
	\snamedRule{appl'}{\ctx\dedT \termf\,\tm\ofT\subst{\typB}{x}{\typeC{t}}}{\ctx\dedT \termf\ofT \piType{x}{\A} \typB \tb \ctx\dedT \tm\ofT \A}
	\]
	\[
	\rnamedRule{$\impl$type'}{\ctx\dedT \termF\Rightarrow \termC{G}\ofT\bool}{\ctx\dedT \termF\ofT\bool\tb \concatCtx{\ctx}{\namedass{x}{\termF}}\dedT \termC{G}\ofT\bool}{implType'}\tb
	\rnamedRule{$=$type}{\ctx\dedT \s\termEquals{\A}\typeC{t}\ofT\bool}{\ctx\dedT \s\ofT\A\tb \ctx\dedT \tm\ofT \A}{eqType'}
	\]
	\[\rnamedRule{$\subtype{}{p}$I}{\ctx\dedT \tm\ofT \subtype{\A}{\p}}{\ctx\dedT \tm\ofT \A \tb  \ctx\dedT \p\ \tm}{psubI}\tb
	\rnamedRule{$\subtype{}{\p}$E1}{\ctx\dedT \tm\ofT\A}{\ctx\dedT \tm\ofT \subtype{\A}{\p}}{psubE1}\tb
	\snamedRule{QI}{\ctx\dedT \tm\ofT \quot{\A}{\r}}{\ctx\dedT \tm\ofT\A\tb \ctx\dedT \isEqRel{\r}}
	\]
	\[
	\rnamedRule{QE}{\ctx\dedT \tm\substOp\x\s \ofT\typB\substOp\x\s}{
		\ctx\dedT \s\ofT\quot\A\r\quad
		\ctx,\,\x\ofT\A,\,\assertL\x\termEquals{\quot\A\r}\s\dedT \tm\ofT \typB\quad
		\ctx,\,\x\ofT\A,\,\xp\ofT\A,\,\assertL\x\termEquals{\quot\A\r}\s,\,\assertL\xp\termEquals{\quot\A\r}\s\dedT \tm\termEquals{\typB} \tm\substOp\x\xp
	}{quotE}\]
	\\
	Term equality; congruence, reflexivity, symmetry, $\beta$, $\eta$:
	\[
	\rnamedRule{cong$\lambda$'}{\ctx\dedT \lambdaFun{x}{\A} t\termEquals{\piType{x}{\A}\!\!\typB} \lambdaFun{x}{\Ap} t'}{\ctx\dedT \A\typeEquals \Ap \ \;\concatCtx{\Gamma}{\x\ofT \A} \dedT \tm\termEquals{\typB} \termtp}{congLam'}
	\quad	
	\snamedRule{congAppl'}{\ctx\dedT \termf\ \tm\termEquals{\typB} \termfp\ \termtp}{\ctx\dedT \tm\termEquals{\A} \termtp\quad \ctx\dedT \termf\termEquals{\piType{x}{\A} \typB} \termfp}
	\]
	\[
	\rnamedRule{refl}{\ctx\dedT \tm\termEquals{\A} \tm}{\ctx\dedT \tm\ofT \A}{refl'}\tb
	\rnamedRule{sym}{\ctx\dedT \s\termEquals{\A} \tm}{\ctx\dedT \tm \termEquals{\A}\s}{sym'}\tb
	\rnamedRule{beta}{\ctx\dedT (\lambdaFun{x}{\A} \s)\ \tm \termEquals{B} \subst{\s}{x}{\tm} }{\ctx\dedT (\lambdaFun{x}{\A} \s)\ \tm\ofT \typB}{beta'}\tb
	\snamedRule{etaPi}{\ctx\dedT \tm\termEquals{\piType{x}{\A} \!\!\typB} \lambdaFun{x}{\A}\tm\ \x}{\ctx\dedT \tm\ofT \piType{x}{\A} \typB}
	\]
	\[
	\rnamedRule{$\subtype{}{\p}$Eq}{\ctx\dedT \s\termEquals{\subtype{\A}{\p}} \tm}{\tb\ctx\dedT \s\termEquals{\A}\tm\tb \ctx\dedT \p\ \s}{psubEq}\tb
	\rnamedRule{Q$\termEquals{}$}{\ctx\dedT (\s\termEquals{\quot{\A}{\r}}\tm)\termEquals\bool(\r\ \s\ \tm)}	{\ctx\dedT \s\ofT\A\quad \ctx\dedT \tm\ofT\A \quad \ctx\dedT\r\ofT\A\to\A\to\bool\tb\isEqRel{\r}}{QEq}
	\]
	\\
	Rules for validity:
	\[
	\rnamedRule{axiom}{\ctx\dedT \termF}{\thyIn{\namedax{c}{\typeC{F}}}{\theorycolor{T}} \tb  \dedT\Ctx{\ctx}}{axiom'}\tb
	\rnamedRule{assume}{\ctx\dedT \termF}{\ctxIn{\namedass{x}{\termF}}{\ctx} \tb  \dedT\Ctx{\ctx}}{assume'}
	\]
	\[
	\rnamedRule{$\impl$I}{\ctx\dedT \termF\Rightarrow \termC{G}}{\ctx\dedT \termF\ofT\bool \tb \concatCtx{\ctx}{\namedass{x}{\termF}}\dedT \termC{G}}{implI'}\tb
	\rnamedRule{$\impl$E}{\ctx\dedT \termC{G}}{\ctx\dedT \termF\Rightarrow \termC{G}\tb \ctx\dedT \termF}{implE'}
	\]
	\[
	\rnamedRule{cong$\ded$}{\ctx\dedT \termF}{\ctx\dedT \termF\termEquals{\bool} \termC{F'}\tb \ctx\dedT \termC{F'}}{congDed'}\tb
	\rnamedRule{boolExt}{\ctx,\x\ofT \bool\dedT \p\ \x}{\ctx\dedT \p\ \T \tb \ctx\dedT \p\ \F}{boolExt'}
	\]
	\[
	%\snamedRule{repr}{\ctx\dedT \existQuant{repr}{\piType{x}{\quot{\A}{\r}}\A} \univQuant{x}{\quot{\A}{\r}}\varname{repr}\ \x\termEquals{\quot{\A}{\r}}\x}{\ctx\dedT \r\ofT \piType{x_1}{\A}\piType{x_2}{\A}\bool\tb\ctx\dedT \Type{\quot{\A}{\r}}}\tb
	\rnamedRule{$\subtype{}{\p}$E2}{\ctx\dedT \p\ \tm}{\ctx\dedT \tm\ofT \subtype{\A}{\p}}{psubE}%\tb
	%\snamedRule{quotL}{\ctx\dedT \existQuant{y}{\A} \quotM{\y}{\r}\termEquals{\quot{\A}{\r}} \tm}{\ctx\dedT \tm\ofT \quot{\A}{\r}}
	\]
}

We also have the axiom from Definition~\ref{ax:quotcod}.

Finally, we modify the rule \ruleRef{nonempty} for the non-emptiness of types: we allow the existence of empty dependent types and only require that for each HOL type in the image of the translation there exists one non-empty DHOL type translated to it (rather than requiring all dependent types translated to it to be non-empty).
Observe that either restricting to the fragment HOL of DHOL or translating to it then yields the non-emptyness assumptions for HOL types.

\section{The translation from \dhole into HOL}\label{appendix:translation}
Before actually going into the soundness and completeness proofs, we repeat and enumerate the cases in the definition of the translation, so we can reference them in the following. 
\begin{definition}[Translation]
	We define a translation from DHOL to HOL syntax by induction on the Grammar.
	
	We use the notation $\overrightarrow{\x\ofT \A}, \overrightarrow{\piType{x}{\A}}, \overrightarrow{\A}$ and $\overrightarrow{\x}$ 
	to denote $\x\ofT \typeC{A_1}, \ldots, \varname{x_n}\ofT \typeC{A_n}$, $\piType{x_1}{\typeC{A_1}}\ldots \piType{x_n}{\typeC{A_n}}$, $\typeC{A_1}\to\ldots\to \typeC{A_n}$ and $\varname{x_1}\ \ldots\ \varname{x_n}$ 
	% and $\univQuant{x_1}{\PhiAppl{A_1}}\PredPhi{A_1}{x_1}\impl \ldots \univQuant{x_n}{\PhiAppl{A_n}}\PredPhi{A_n}{x_n}\impl$ 
	respectively.
	
	The cases for theories and contexts are:
	{\small
		\begin{align*}
			\PhiAppl{\emptyThy}:=&\emptyThy \plabel{PTemptyThy}\\
			\PhiAppl{\concatThy{\theorycolor{T}}{D}}:=&\concatThy{\PhiAppl{\theorycolor{T}}}{\PhiAppl{D}}&&\text{where}\\
			\PhiAppl{\;\a\ofT \overrightarrow{\piType{x}{\A}} \type} :=& \a\ofT \type,\\& \PredPhiName{a}\ofT \overrightarrow{\PhiAppl{\A}} \to \a\to \a\to \bool,\\
			& \namedax{a_{trans}}{}\forall \overrightarrow{\x\!:\!\PhiAppl{\A}}.~ \univQuant{u,v,w}{\a}\termEqT{\left(\a\ \overrightarrow{x}\right)}{\termC{u}}{\varname{v}}\impl \left(\termEqT{\left(\a\ \overrightarrow{x}\right)}{\varname{v}}{\varname{w}}\impl\termEqT{\left(\a\ \overrightarrow{\x}\right)}{\varname{u}}{\varname{w}}\right),\\
			& \namedax{a_{sym}}{}\forall\overrightarrow{\x\!:\!\PhiAppl{\A}}.~ \univQuant{u,v}{\a}\termEqT{\left(\a\ \overrightarrow{\x}\right)}{\varname{u}}{\varname{v}}\impl\termEqT{\left(\a\ \overrightarrow{\x}\right)}{\varname{v}}{\varname{u}},\\
			&\namedax{a_{PER}}{} \forall\overrightarrow{\x\!:\!\PhiAppl{\A}}.~ \univQuant{u,v}{\a}\PredPhi{\left(\a\ \overrightarrow{\x}\right)}{\varname{v}}\impl \termEqT{(\a\ \overrightarrow{\x})}{\varname{u}}{\varname{v}}\termEqB \varname{u}\termEquals{\a}\varname{v}
			%\\&\namedax{a_{choice}}{}\forall \overrightarrow{\x\!:\!\PhiAppl{\A}}.~
			%\univQuant{\p}{\PhiAppl{\a}}
			%\PredPhi{\left(\a\ \overrightarrow{\x}\right)}{\left(\choiceOp{x}{\PhiAppl{\A}}{\p}\right)}
			\plabel{PTTpConstr}\\
			\PhiAppl{\constname{c}\ofT \A}:=&\constname{c}:\PhiAppl{\A},\quad \namedax{\typingAxName{c}}{\PredPhi{A}{\constname{c}}} \plabel{PTtermDecl}\\
			\PhiAppl{\namedax{ax}{\termF}}:=&\namedass{ax}{\PhiAppl{\termF}} \plabel{PTax}\\[0.5cm]
			\PhiAppl{\emptyCtx}:=&\emptyCtx \plabel{PTemptyCtx}\\
			\PhiAppl{\ctx,\x\ofT\A}:=&\PhiAppl{\Gamma},\;\x\ofT\PhiAppl{\A},\namedass{\typingAssName{x}}{\PredPhi{\A}{\x}} \plabel{PTctxVar}\\
			\PhiAppl{\Gamma,\namedass{ass}{\termF}}:=&\PhiAppl{\ctx},\;\namedass{ass}{\PhiAppl{\termF}} \plabel{PTctxAss}
		\end{align*}
	}
	The case of $\PhiAppl{\A}$ and $\termEqT{A}{\s}{\tm}$ for types $\A$ are:
	\begin{align*}
		\PhiAppl{(\a\ \termC{t_1}\ \ldots \ \termC{t_n})}&:=\a \plabel{PTTpAppl}\\
		\termEqT{(\a\ \termC{t_1}\ \ldots \ \termC{t_n})}{\s}{\tm}&:=\PredPhiName{a}\ \PhiAppl{\termC{t_1}}\ \ldots\ \PhiAppl{\termC{t_n}}\ \s\ \tm \plabel{PTTpPredAppl}\\
		\PhiAppl{\piType{x}{\A}\typB} &:= \PhiAppl{\A} \to \PhiAppl{\typB} \plabel{PTPitype}\\
		\termEqT{(\piType{x}{\A}\typB)}{\termf}{\termC{g}} &:= \univQuant{x,y}{\PhiAppl{\A}}%\univQuant{y}{\PhiAppl{\A}}
		\termEqT{\A}{\x}{\y}\impl \termEqT{\typB}{\left(\termf\ \x\right)}{\left(\termC{g}\ \y\right)} \plabel{PTPipred}\\
		%\termEqT{\subtype{\A}{\p}}{s}{t} &:= \termEqT{\A}{s}{t}\land \PhiAppl{p}\ s\land \PhiAppl{p}\ t\\
		\PhiAppl{\bool}&:=\bool \plabel{PTTpBool}\\
		\termEqT{\bool}{\s}{\tm}&:=\s\termEqB \tm \plabel{PTTpBoolPred}\\
			\ifbool{inAppendix}{\\	
				\PhiAppl{\subtype{\A}{\p}} &:= \PhiAppl{\A}\plabel{PTPStype}\\
				\termEqT{\left(\subtype{\A}{\p}\right)}{\s}{\tm} &:= \termEqT{\A}{\s}{\tm}\land \PhiAppl{\p}\ \s\land \PhiAppl{\p}\ \tm\plabel{PTPSpred}\\
				\PhiAppl{\quot{\A}{\r}} &:=
				\PhiAppl{\A}\plabel{PTQtype}\\
				\termEqT{\left(\quot{\A}{\r}\right)}{\s}{\tm} &:=\PhiAppl{\r}\ \s\ \tm\land \PredPhi{A}{\s}\land \PredPhi{A}{\tm}\plabel{PTQpred}}{}
		\end{align*}
	The cases for terms are:
	\begin{align*}
		\PhiAppl{\constname{c}} &:= \constname{c} \plabel{PTmConst}\\
		\PhiAppl{\x} &:= \x\plabel{PTmVar}\\
		\PhiAppl{\lambdaFun{x}{\A} \tm} &:= \lambdaFun{x}{\PhiAppl{\A}} \PhiAppl{\tm}\plabel{PTLam}\\
		\PhiAppl{\termf\ \tm} &:= \PhiAppl{\termf}\ \PhiAppl{\tm}\plabel{PTappl}\\
		\PhiAppl{\termF\Rightarrow \termC{G}}&:=\PhiAppl{\termF} \impl \PhiAppl{\termC{G}}\plabel{PTImpl}\\
		\PhiAppl{\s\termEquals{\A}\tm}&:=\termEqT{A}{\PhiAppl{\s}}{\PhiAppl{\tm}}\plabel{PTEq}\\
		%\PhiAppl{\choiceOp{x}{\A}{\p}} &:= \choiceOp{x}{\PhiAppl{\A}}{\lambdaFun{x}{\PhiAppl{\A}}
		%	\PredPhi{\A}{\x}\land \PhiAppl{\p}\ \x\land \univQuant{y}{\PhiAppl{\A}} \PredPhi{\A}{\y}\land \PhiAppl{\p}\ \y \implC \termEqT{A}{\x}{\y}} \plabel{PTchoice}
	\end{align*}
\end{definition}

\section{Completeness proof}\label{appendix:Complete}
To simplify the inductive arguments, we will actually prove the following slightly stronger version of the theorem:

\begin{theorem}[Completeness]\label{thm:complete}
	We have 
	\begin{align}
		\phantom{\ctx}&\ded \Thy{\theorycolor{T}}  &&\text{ implies } \phantom{\ctx}\vdash_{\phantom{T}}\Thy{\PhiAppl{\theorycolor{T}}}&& \label{correct:theoremhood}\\
		\phantom{\ctx}&\dedT \Ctx{\ctx} &&\text{ implies } \phantom{\ctx}\dedPT\Ctx{\ctxT}&& \label{correct:contexthood}\\
		\ctx&\dedT \Type{\A} &&\text{ implies } \ctxT\dedPT\Type{\PhiAppl{\A}} &&\text{ and }\ 
		\ctxT\dedPT\PredPhiName{A}\ofT \PhiAppl{\A}\to\PhiAppl{\A}\to \bool&& \label{correct:typehood}\\
		\ctx&\dedT \A \typeEquals \typB &&\text{ implies } \ctxT\dedPT\PhiAppl{\A} \typeEquals \PhiAppl{\typB} 
		&& \text{ and }\  \concatCtx{\ctxT}{\x \ofT\PhiAppl{\A}}\dedPT\PredPhi{A}{\x} \termEqB \PredPhi{B}{\x}\label{correct:typeEq}\\
		\ctx&\dedT \tm\ofT\A &&\text{ implies } \ctxT\dedPT\PhiAppl{\tm}\ofT\PhiAppl{\A} &&\text{ and }\  \ctxT\dedPT\PredPhi{A}{\PhiAppl{\tm}}\label{correct:typing}\\
	%	\intertext{In case of $\subtyping$ we strengthen the first claim of $$\ctxT,\x\ofT\PhiA,\typingAss{\A}{\x}\dedPT \x\ofT\PhiAppl{\typB}$$ to $\ctxT\dedPT \PhiAppl{\A}\typeEquals\PhiAppl{\typB}$ yielding:}
		\ctx&\dedT \A \subtyping \typB&&\text{ implies } \ctxT\dedPT \PhiAppl{\A} \typeEquals \PhiAppl{\typB} &&\text{ and }\ \concatCtx{\ctxT}{\x ,\y\ofT\PhiAppl{\typB}}\dedPT \termEqT{A}{\x}{\y}\impl \termEqT{B}{\x}{\y}\label{correct:subtyping}\\
		\ctx&\dedT \termF &&\text{ implies } \ctxT\dedPT\PhiAppl{\termF} &&\label{correct:validity}
		\intertext{In case of term equality, we strengthen the claim to:}
		\ctx&\dedT \tm \termEquals{\A} \termtp &&\text{ implies } \ctxT\dedPT\termEqT{A}{\PhiAppl{\tm}}{\PhiAppl{\termtp}}\ &&\text{ and }\ \ctxT\dedPT \PhiAppl{\tm}\ofT\PhiAppl{\A}\quad \text{and}\ \ctxT\dedPT \PhiAppl{\termtp}\ofT\PhiAppl{\A}\label{correct:termEq}
	\end{align}
	
	Additionally, the typing relations $\PredPhiName{A}$ are symmetric and transitive on all well-formed types $A$:
	\begin{align}
		\ctx\dedT \Type{\A} &\Impl \ctxT\dedPT 
		\univQuant{x,y}{\PhiAppl{\A}}
		\termEqT{A}{\x}{\y}\impl \termEqT{A}{\y}{\x}
		\label{correct:relatSym}\\
		\ctx\dedT \Type{\A} &\Impl \ctxT\dedPT 
		\univQuant{\x ,\y,\z}{\PhiAppl{\A}}\termEqT{A}{\x}{\y}\impl \left(\termEqT{A}{\y}{\z}\impl \termEqT{A}{\x}{\z}\right)\label{correct:relatTrans}
	\end{align}
	
	Moreover, the following substitution lemma holds, i.e.,
	\begin{align}
		\concatCtx{\ctx}{\x \ofT\A}\dedT \tm\ofT\typB\Mand\ctx\ded \termC{u}\ofT\A &\Impl \ctxT\dedPT\PhiAppl{\subst{\tm}{x}{\termC{u}}}\termEquals{\PhiAppl{\typB}} \subst{\PhiAppl{\tm}}{x}{\PhiAppl{\termC{u}}}\label{correct:substTerm}\\
		\concatCtx{\ctx}{\x \ofT\A}\dedT \typB\ \type\Mand\ctx\dedT \termC{u}\ofT\typB &\Impl \ctxT\dedPT\PhiAppl{\subst{\typB}{x}{\termC{u}}}\typeEquals \subst{\PhiAppl{\typB}}{x}{\PhiAppl{\termC{u}}}\label{correct:substType}
	\end{align}
	In the following line, we assume that if $\tm=\lambdaFun{y}{\typC}\s$ for $\s$ of type $\typeC{D}$, then $\typB=\piType{y}{\typC}\typeC{D}$ (this is enough in practice and we cannot easily show more). In the presence of nontrivial subtyping we need this assumption that types "fit together exactly" to prove this part of the substitution lemma.
	%Similarly, if $\tm=\termf\ t$ for $\termf$ of type $\piType{y}{C}\typeC{D}$ and $\s$ of type $C$, we assume that $B=D$.\ednote{Second sentence doesn't make sense.}
	\begin{align}
		\concatCtx{\ctx}{\x \ofT\A}\dedT \tm\ofT\typB &\Impl \concatCtx{\concatCtx{\ctxT}{\x ,\xp \ofT\PhiAppl{\A}}}{\namedass{xRx'}{\termEqT{A}{\x}{\xp }}} \dedPT \termEqT{B}{\PhiAppl{\tm}}{\subst{\PhiAppl{\tm}}{x}{\xp }}\label{correct:substRelatTerms}
		%\\\ctx\dedT \tm\ofT\typB \Mand \ctx\dedT \s\ofT\typB\Mand \ctx\dedT \termC{u}\ofT\typB&\Impl \ctxT \dedPT \termEqT{B}{\PhiAppl{\tm}}{\PhiAppl{\s}}\impl \left(\termEqT{B}{\PhiAppl{\s}}{\PhiAppl{u}}\impl \termEqT{B}{\PhiAppl{\tm}}{\PhiAppl{u}}\right)\label{correct:relatTrans}
	\end{align}
\end{theorem}

Here Case~\ref{correct:typeEq} looks weaker than in the original statement, but is easily seen to be equivalent.
The equivalence proof uses induction on the shape of the types (reducing the claim to base types), propositional extensionality and the PER axioms. 
Furthermore, the statement that $\PredPhiName{\A}$ is a PER no longer appears in the statement for well-typedness since the fact that the PERs are actually PERs is proved separately.

\renewcommand{\subparagraph}[1]{\paragraph*{#1}}
\begin{proof}[Proof of Theorem~\ref{thm:complete}]
	Firstly, we will prove the substitution lemma by induction on the grammar, i.e. by induction on the shape of the terms and types.
	
	Afterwards, we will prove completeness of the translation w.r.t. all \dhole judgements by induction on the derivations. This means that we consider the inference rules of \dhole and prove that if completeness holds for the assumptions of a \dhole inference rule, then it also holds for the conclusion of the rule. 
	For the inductive steps for some typing rules, namely \ruleRef{eqType}, we also require the fact that for any (well-formed) type $\A$ in \dhole we have $\PredPhiName{A}\ofT\PhiAppl{\A}\to\PhiAppl{\A}\to\bool$. 
	This follows directly from how the $\PredPhiName{A}$ are generated/defined in the translation.
	
	\subsection{Substitution lemma and symmetry and transitivity of the typing relations}\label{subsec:substLemma}
	Since the translation of types commutes with the type productions of the grammar (\ref{correct:substType}) is obvious.
	
	We show (\ref{correct:substTerm}) by induction on the grammar of \dhole. 
	If $\x$ is not a free variable in $\tm$, then $\PhiAppl{\subst{\tm}{x}{\termC{u}}}=\PhiAppl{\tm}=\subst{\PhiAppl{\tm}}{x}{\PhiAppl{\termC{u}}}$ and the claim (\ref{correct:substTerm}) follows by rule \ruleRef{refl}. 
	So assume that $\x$ is a free variable of $\tm$.
	
	If $\tm$ is a variable, then by assumption (that $\x$ is a free variable in $\tm$) we have $\tm=\x$ and thus $\PhiAppl{\subst{\tm}{x}{\termC{u}}}=\PhiAppl{\termC{u}}=\subst{\PhiAppl{\tm}}{x}{\PhiAppl{\termC{u}}}$ so the claim follows by rule \ruleRef{refl}. 
	
	If $\tm$ is a $\lambda$-term $\lambdaFun{y}{\A} \s$, then by induction hypothesis we yield $\concatCtx{\ctxT}{\y\ofT \PhiAppl{\A}}\dedPT \PhiAppl{\subst{\s}{x}{\termC{u}}}\termEquals{\PhiAppl{\A}}\subst{\PhiAppl{\s}}{x}{\PhiAppl{\termC{u}}}$, where $\A$ is the type of $\s$. Using rule \ruleRef{congLam}, the claim of $\ctxT\dedPT \PhiAppl{\subst{\lambdaFun{y}{\A} \s}{x}{\termC{u}}}\termEquals{\PhiAppl{\typB}}\subst{\PhiAppl{\lambdaFun{y}{\A} \s}}{x}{\PhiAppl{\termC{u}}}$ follows.
	
	If $\tm$ is a function application $\termf\ \s$, the induction hypothesis yields $\ctxT\dedPT \PhiAppl{\subst{\s}{x}{\termC{u}}}\termEquals{\PhiAppl{\A}}\subst{\PhiAppl{\s}}{x}{\PhiAppl{\termC{u}}}$ and $\ctxT\dedPT \PhiAppl{\subst{\termf}{x}{\termC{u}}}\termEquals{\PhiAppl{\A}\to\PhiAppl{\typB}}\subst{\PhiAppl{\termf}}{x}{\PhiAppl{\termC{u}}}$, where $\A$ is the type of $\s$. From rule \ruleRef{congAppl}, the claim of $\ctxT\dedPT \PhiAppl{\subst{\left(\termf\ \s\right)}{x}{\termC{u}}}\termEquals{\PhiAppl{\typB}}\subst{\PhiAppl{\termf\ \s}}{x}{\PhiAppl{\termC{u}}}$ follows.
	
	If $\tm$ is an equality $\s\termEquals{\A}\sp$, the induction hypothesis yields $\ctxT\dedPT \PhiAppl{\subst{\s}{x}{\termC{u}}}\termEquals{\PhiAppl{\A}}\subst{\PhiAppl{\s}}{x}{\PhiAppl{\termC{u}}}$ and $\ctxT\dedPT \PhiAppl{\subst{\sp}{x}{\termC{u}}}\termEquals{\PhiAppl{\A}}\subst{\PhiAppl{\sp}}{x}{\PhiAppl{\termC{u}}}$, where $\A$ is the type of $\s$ and $\sp$. 
	The claim of $\ctxT\dedPT \PhiAppl{\subst{\left(\s\termEquals{\A}\sp\right)}{x}{\termC{u}}}\termEqB\subst{\left(\PhiAppl{\s\termEquals{\A}\termC{\sp}}\right)}{x}{\PhiAppl{\termC{u}}}$ follows using rule \ruleRef{termEqCong}.
	
	%Finally, if $\tm$ is an equivalence class $\quotM{\s}{\r}$, then by the translation definition $\PhiAppl{\tm}=\PhiAppl{\s}$ and $\PhiAppl{\subst{\tm}{x}{\termC{u}}}=\PhiAppl{\subst{\s}{x}{\termC{u}}}$ and $\subst{\PhiAppl{\tm}}{x}{\PhiAppl{\termC{u}}}=\subst{\PhiAppl{\s}}{x}{\PhiAppl{\termC{u}}}$.
	%By induction hypothesis we have $\ctx\dedPT \PhiAppl{\subst{\s}{x}{\termC{u}}}\termEquals{\PhiAppl{\A}}\subst{\PhiAppl{\s}}{x}{\PhiAppl{\termC{u}}}$, so the claim follows.
	
	Before we can show (\ref{correct:substRelatTerms}), we first need to prove the symmetry and transitivity of the typing relations:
	We can prove both by induction on the type $\A$. The base cases are base types and equality, for which the claims follow by construction.
%	\ednote{CR@FR: This proof is straightforward, but very tedious and practically unmaintainable. So I commented out the details here.}
	%\input{perhood}
	
	We show (\ref{correct:substRelatTerms}) by induction on the grammar:
		Without loss of generality we may assume that $\typB=:\subtype{\Bp}{\p}$ for $\Bp$ either a quotient--, a base-- or a $\Pi$-type. This is due to the fact that quotinet--, base-- and $\Pi$-types $\Bp$ can be written as $\subtype{\Bp}{\lambdaFun{x}{\Bp}\T}$ and types of the form $\subtype{\subtype{\typeC{B''}}{\p}}{\q}$ can be rewritten as $\subtype{\typeC{B''}}{\lambdaFun{x}{\typeC{B''}}\p\ \x \land \q\ \x}$.
	
		If $\tm$ is a constant or variable then $\subst{\PhiAppl{\tm}}{x}{\xp }=\PhiAppl{\tm}$ and by case (\ref{PTctxVar}) resp. by case (\ref{PTax}) in the definition of the translation, we have $\termEqT{A}{\PhiAppl{\tm}}{\PhiAppl{\tm}}$. So the claim holds.
		
		If $\tm$ is a $\lambda$-term $\lambdaFun{y}{\typC}\s$ and $\Bp=\piType{z}{\typC}\typeC{D}$, then by induction hypothesis we have \[\concatCtx{\concatCtx{\ctxT}{\x ,\xp \ofT\PhiAppl{\A}}}{\namedass{xRx'}{\termEqT{A}{\x}{\xp }}} \dedPT \termEqT{D}{\PhiAppl{\s}}{\subst{\PhiAppl{\s}}{x}{\xp }}.\]
		By the rules \ruleRef{forallI}, \ruleRef{implI}, we yield \[\ctxT\dedPT \univQuant{x,y}{\PhiAppl{\A}}\termEqT{A}{\x}{\y}\impl\termEqT{D}{\PhiAppl{\s}}{\subst{\PhiAppl{\s}}{x}{\xp }}. \]
		By definition (\ref{PTPipred}) this is exactly \[
		\concatCtx{\concatCtx{\ctxT}{\x ,\xp \ofT\PhiAppl{\A}}}{\namedass{xRx'}{\termEqT{A}{\x}{\xp }}} \dedPT \termEqT{B'}{\PhiAppl{\tm}}{\subst{\PhiAppl{\tm}}{x}{\xp }}.\]
		Since $\tm$ is a $\lambda$-term, by assumption we have that $\typB \typeEquals \Bp=\subtype{\typB}{\lambdaFun{z}{\typB}\T}$, so the claim follows trivially.
		
		If $\tm$ is a function application $\termf\ \s$ with $\termf$ of type $\piType{z}{\typC}\typeC{D}$ and $\s$ of type $\typC$, then by assumption $\typB=\typeC{D}\typeEquals \Bp=\subtype{\typB}{\lambdaFun{z}{\typB}\T}$, so it suffices to prove that \[ \concatCtx{\concatCtx{\ctxT}{\x ,\xp \ofT\PhiAppl{\A}}}{\namedass{xRx'}{\termEqT{A}{\x}{\xp }}} \dedPT \termEqT{D}{\PhiAppl{\termf\ \s}}{\subst{\PhiAppl{\termf\ \s}}{\x}{\xp }}.\]
		By induction hypothesis and (\ref{correct:substTerm}) we then have:
		\[
		\concatCtx{\concatCtx{\ctxT}{\x ,\xp \ofT\PhiAppl{\A}}}{\namedass{xRx'}{\termEqT{A}{\x}{\xp }}} \dedPT \termEqT{\left(\piType{z}{\typC}\typeC{D}\right)}{\PhiAppl{\termf}}{\subst{\PhiAppl{\termf}}{\x}{\xp }}
		\] and
		\begin{equation}
			\concatCtx{\concatCtx{\ctxT}{\x ,\xp \ofT\PhiAppl{\A}}}{\namedass{xRx'}{\termEqT{A}{\x}{\xp }}} \dedPT \termEqT{C}{\PhiAppl{\s}}{\subst{\PhiAppl{\s}}{\x}{\xp }}\label{substLemAppl2}
			.\end{equation}
		By definition (\ref{PTPipred}), we can unpack the former to:
		\begin{equation}
			\concatCtx{\concatCtx{\ctxT}{\x ,\xp \ofT\PhiAppl{\A}}}{\namedass{xRx'}{\termEqT{A}{\x}{\xp }}} \dedPT
			\univQuant{z,z'}{\PhiAppl{\typC}}\termEqT{C}{\z}{\zp}\impl \termEqT{\left(\piType{z}{\typC}\typeC{D}\right)}{\PhiAppl{\termf}\ \z}{\subst{\PhiAppl{\termf}}{\x}{\xp }\ \subst{\zp}{x}{\xp }}\label{substLemAppl1}
		\end{equation}
		Using the rules \ruleRef{forallE} and \ruleRef{implE} (using (\ref{substLemAppl1})) to plug in $\PhiAppl{\s}$ resp. $\subst{\PhiAppl{\s}}{x}{\xp }$ for $\termC{z}, \termC{z'}$ in (\ref{substLemAppl1}), we yield:
		\[
		\concatCtx{\concatCtx{\ctxT}{\x ,\xp \ofT\PhiAppl{\A}}}{\namedass{xRx'}{\termEqT{A}{\x}{\xp }}} \dedPT \termEqT{\left(\piType{z}{\typC}\typeC{D}\right)}{\PhiAppl{\termf}\ \PhiAppl{\s}}{\subst{\PhiAppl{\termf}}{\x}{\xp }\ \subst{\PhiAppl{\s}}{\x}{\xp }}
		\] which is exactly the desired result.
		
		By definition (\ref{PTTpBoolPred}), the typing relation for type $\bool$ is ordinary equality, so the cases of $\tm$ being an implication or Boolean equality are in fact special cases of (\ref{correct:substTerm}), which is already proven above. 
		It remains to consider the case of $\tm$ being an equality $\s\termEquals{C}\sp$ for $\typC\not\typeEquals\bool$. 
		In this case, the induction hypothesis implies that 
		\begin{equation}
			\concatCtx{\concatCtx{\ctxT}{\x ,\xp \ofT\PhiAppl{\A}}}{\namedass{xRx'}{\termEqT{A}{\x}{\xp }}} \dedPT \termEqT{C}{\PhiAppl{\s}}{\subst{\PhiAppl{\s}}{x}{\xp }}\label{substLemEq1}
		\end{equation} and \begin{equation}
			\concatCtx{\concatCtx{\ctxT}{\x ,\xp \ofT\PhiAppl{\A}}}{\namedass{xRx'}{\termEqT{A}{\x}{\xp }}} \dedPT \termEqT{\typC}{\PhiAppl{\sp}}{\subst{\PhiAppl{\sp}}{\x}{\xp }}\label{substLemEq2}
		\end{equation}
		We need to prove
		\[
		\concatCtx{\concatCtx{\ctxT}{\x ,\xp \ofT\PhiAppl{\A}}}{\namedass{xRx'}{\termEqT{A}{\x}{\xp }}} \dedPT
		\termEqT{C}{\PhiAppl{\s}}{\PhiAppl{s'}} \termEqB \termEqT{C}{\subst{\PhiAppl{\s}}{x}{\xp }}{\subst{\PhiAppl{\sp}}{x}{\xp }}.\]
		
		If we can show  
		\[\concatCtx{\concatCtx{\concatCtx{\ctxT}{\x ,\xp \ofT\PhiAppl{\A}}}{\namedass{xRx'}{\termEqT{A}{\x}{\xp }}}}{\namedass{sRs'}{\termEqT{C}{\PhiAppl{\s}}{\PhiAppl{s'}}}} \dedPT \termEqT{C}{\subst{\PhiAppl{\s}}{x}{\xp }}{\subst{\PhiAppl{\sp}}{x}{\xp }}\] and similarly also 
		\[\concatCtx{\concatCtx{\concatCtx{\ctxT}{\x ,\xp \ofT\PhiAppl{\A}}}{\namedass{xRx'}{\termEqT{A}{\x}{\xp }}}}{\namedass{subrel}{\termEqT{C}{\subst{\PhiAppl{\s}}{x}{\xp }}{\subst{\PhiAppl{\sp}}{x}{\xp }}}} \dedPT \termEqT{C}{\PhiAppl{\s}}{\PhiAppl{s'}},\]
		then the claim follows by rule \ruleRef{propExt}.
		
		Both follows from the transitivity (\ref{correct:relatTrans}) of the typing relation $\PredPhiName{\typC}$. 

	\subsection{Proof of remaining soundness theorem by induction on DHOL derivations}
	\subsubsection{Well-formedness of theories}
	Well-formedness of \dhol theories can be shown using the rules %\ruleRef{ctxThy},
	\ruleRef{thyEmpty'}, \ruleRef{thyType'}, \ruleRef{thyConst'} and \ruleRef{thyAxiom'}:
	\subparagraph{\ruleRef{thyEmpty'}:}
	\begin{align}
		\NDLine{}{\Thy{\emptyThy}}{\ruleRef{thyEmpty}}\label{thyEmpty1}\\
		\NDLineH{}{\Thy{\PhiAppl{\emptyThy}}}{\ruleRef{thyEmpty}}\nonumber
	\end{align}
	
	\subparagraph{\ruleRef{thyType'}:}
	\begin{align}
		\NDLineT{}{\Ctx{\varname{x_1}\ofT\typeC{A_1},\ldots,\varname{x_n}\ofT\typeC{A_n}}}{by assumption}\label{thyType1}\\
		\NDLinePT{}{\Ctx{\varname{x_1}\ofT\PhiAppl{\typeC{A_1}},\PredPhi{A_1}{\varname{x_1}},\ldots,\varname{x_n}\ofT\PhiAppl{\typeC{A_n}},\typingAss{A_n}{\varname{x_n}}}}{\IH,(\ref{thyType1})}\label{thyType1a}\\
		\NDLineH{}{\Thy{\PhiAppl{\theorycolor{T}}}}{\ruleRef{ctxThy},(\ref{thyType1a})}\label{thyType2}\\
		\NDLineH{}{\Thy{\concatThy{\PhiAppl{\theorycolor{T}}}{\a\ofT\type}}}{\ruleRef{thyType},(\ref{thyType2})}\label{thyType3}\\
		\NDLineH{}{
			\Thy{\PhiAppl{\concatThy{T}{\a\ofT\piType{x_1}{A_1}\ldots\piType{x_n}{A_n}\type}}}}{\ref{PTTpConstr},(\ref{thyType3})}\nonumber
	\end{align}
	
	\subparagraph{\ruleRef{thyConst'}:}
	\begin{align}
		\NDLineT{}{\Type{\A}}{by assumption}\label{thyConst2}\\
		\NDLinePT{}{\PhiAppl{\A}\:\type}{\IH,(\ref{thyConst2})}\label{thyConst4}\\
		\NDLineH{}{\Thy{\concatThy{\PhiAppl{T}}{\constname{c}\ofT\PhiAppl{\A}}}}{\ruleRef{thyConst},(\ref{thyConst4})}\label{thyConst5}\\
		\NDLineH{}{\Thy{\PhiAppl{\concatThy{T}{c\ofT\A}}}}{\ref{PTtermDecl},(\ref{thyConst5})}
	\end{align}
	
	\subparagraph{\ruleRef{thyAxiom'}:}
	\begin{align}
		\NDLineT{}{\termF\ofT\bool}{by assumption}\label{thyAxiom1}\\
		\NDLinePT{}{\PhiAppl{\termF}\ofT\bool}{\IH,(\ref{thyAxiom1})}\label{thyAxiom2}\\
		\NDLineH{}{\Thy{\concatThy{\thyT}{\namedax{ax}{\PhiAppl{\termF}}}}}{\ruleRef{thyAxiom},(\ref{thyAxiom2})}\label{thyAxiom3}\\
		\NDLineH{}{\Thy{\PhiAppl{\concatThy{\thyT}{\namedax{ax}{\termF}}}}}{\ref{PTax},(\ref{thyAxiom3})}
	\end{align}
	
	\subsubsection{Well-formedness of contexts}
	Well-formedness of contexts can be concluded using the rules \ruleRef{ctxEmpty'}, \ruleRef{ctxVar'} and \ruleRef{ctxAssume'}:
	\subparagraph{\ruleRef{ctxEmpty'}:}
	\begin{align}
		\NDLine{}{\Thy{\thy}}{by assumption}\label{ctxEmpty1}\\
		\NDLineH{}{\Thy{\thyT}}{\IH,(\ref{ctxEmpty1})}\label{ctxEmpty2}\\
		\NDLinePT{}{\Ctx{\emptyCtx}}{\ruleRef{ctxEmpty},(\ref{ctxEmpty2})}\label{ctxEmpty3}\\
		\NDLinePT{}{\Ctx{\PhiAppl{\emptyCtx}}}{\ref{PTemptyCtx},(\ref{ctxEmpty3})}
	\end{align}
	
	\subparagraph{\ruleRef{ctxVar'}:}
	\begin{align}
		\NDLineTG{\Type{\A}}{\byAss}\label{ctxVar1}\\
		\NDLinePTG{\Type{\PhiAppl{\A}}}{\IH,(\ref{ctxVar1})}\label{ctxVar2}
		\\\NDLinePTG{\PredPhiName{A}\ofT\PhiAppl{\A}\to\PhiAppl{\A}\to \bool}{\IH,(\ref{ctxVar1})}\label{ctxVarPATp}
		\\\NDLinePT{}{\Ctx{\concatCtx{\ctxT}{\x \ofT\PhiAppl{\A}}}}{\ruleRef{ctxVar},(\ref{ctxVar2})}\label{ctxVar3}
		\\\NDLinePT{\concatCtx{\ctxT}{\x \ofT\PhiAppl{\A}}}{\PredPhiName{A}\ofT\PhiAppl{\A}\to \PhiAppl{\A}\to \bool}{\ruleRef{varDed},(\ref{ctxVar2}),(\ref{ctxVarPATp})}\label{ctxVar3a}
		\\\NDLinePT{\concatCtx{\ctxT}{\x \ofT\PhiAppl{\A}}}{\PredPhi{A}{\x}\ofT\bool}{\ruleRef{appl},(\ref{ctxVar3a}),\ruleRef{varS}}\label{ctxVar4}
		\\\NDLinePT{}{\Ctx{\concatCtx{\concatCtx{\ctxT}{\x \ofT\PhiAppl{\A}}}{\PredPhi{A}{\x}}}}{\ruleRef{ctxAssume},(\ref{ctxVar4})}\label{ctxVar5}
		\\\NDLinePT{}{\Ctx{\PhiAppl{\concatCtx{\ctx}{ \x \ofT\A}}}}{\ref{PTctxVar},(\ref{ctxVar5})}
	\end{align}
	
	\subparagraph{\ruleRef{ctxAssume'}:}
	\begin{align}
		\NDLineTG{\termF\ofT\bool}{by assumption}\label{ctxAssume1}
		\\\NDLinePTG{\PhiAppl{\termF}\ofT\bool}{\IH,(\ref{ctxAssume1})}\label{ctxAssume2}
		\\\NDLinePT{}{\Ctx{\concatCtx{\ctxT}{\namedass{ass}{\PhiAppl{\termF}}}}}{\ruleRef{ctxAssume},(\ref{ctxAssume2})}\label{ctxAssume3}
		\\\NDLinePT{}{\Ctx{\PhiAppl{\concatCtx{\ctx}{\namedass{ass}{\termF}}}}}{\ref{PTctxAss},(\ref{ctxAssume3})}
	\end{align}
	
	\subsubsection{Well-formedness of types}
	Well-formedness of types can be shown in DHOL using the rules
	\ruleRef{type'}, \ruleRef{bool'}, \ruleRef{pi}, \ruleRef{Qtype} and \ruleRef{psubType}:
	\subparagraph{\ruleRef{type'}:}
	\begin{align}
		&\thyIn{\a\ofT\piType{x_1}{\typeC{A_1}}\ldots\piType{x_n}{\typeC{A_n}}\type}{T}&&\text{\byAss}\label{type0}\\
		%\NDLineTG{\Ctx{\ctx}}{\byAss}\label{typeGCtx}\\
		%\NDLinePT{}{\Ctx{\ctxT}}{\IH,(\ref{typeGCtx})}\label{typePGCtx}\\
		\NDLineTG{\termC{t_1}\ofT\typeC{A_1}}{\byAss}\label{typeAfirst}\\
		\vdots\nonumber\\
		\NDLineTG{\termC{t_n}\ofT\typeC{A_n}\substOp{x_1}{\termC{t_1}}\ldots\substOp{x_n}{\termC{t_n}}}{\byAss}\label{typeAlast}\\
		\NDLinePTG{\PhiAppl{\termC{t_1}}\ofT\PhiAppl{\typeC{A_1}}}{\IH,(\ref{typeAfirst})}\label{typeIHfirst}\\
		\vdots\nonumber\\
		\NDLinePTG{\PhiAppl{\termC{t_n}}\ofT\PhiAppl{\typeC{A_n}}}{\IH,(\ref{typeAlast})}\label{typeIHlast}\\
		&\thyIn{\a\ofT\type}{\PhiAppl{T}}&&\text{\ref{PTTpConstr},(\ref{type0})}\label{typeT1}\\
		&\thyIn{\PredPhiName{a}\ofT\PhiAppl{\typeC{A_1}}\to\ldots\PhiAppl{\typeC{A_n}}\to\a\to \a\to\bool}{\PhiAppl{T}}&&\text{\ref{PTTpConstr},(\ref{type0})}\label{typeT2}\\
		%\NDLinePT{\phantom{\ctxT}}{\Type{a}}{\ruleRef{type},(\ref{typeT1}),\ruleRef{ctxEmpty'}}\label{type'pre1}\\
		\NDLinePTG{\PredPhiName{a}\ofT\PhiAppl{\typeC{A_1}}\to\ldots\PhiAppl{\typeC{A_n}}\to\a\to \a\to\bool}{\ruleRef{constS},(\ref{typeT2})}\label{typeT3}\\
		\NDLinePTG{\PredPhiName{a}\ \PhiAppl{\termC{t_1}}\ofT\PhiAppl{\typeC{A_2}}\to\ldots\PhiAppl{\typeC{A_n}}\to\a\to \a\to\bool}{\ruleRef{appl},(\ref{typeT3}),(\ref{typeIHfirst})}\label{typeT3first}\\
		\vdots\nonumber\\
		\NDLinePTG{\PredPhiName{a}\ \PhiAppl{\termC{t_1}}\ \ldots\ \PhiAppl{\termC{t_n}}\ofT\a\to \a\to\bool}{\ruleRef{appl},previous line,(\ref{typeIHlast})}\label{typeT3last}\\
		\NDLinePT{}{\Ctx{\ctxT}}{\ruleRef{tpCtx},\ruleRef{typingTp},(\ref{typeT3})}\label{typePGCtx}\\
		\NDLinePTG{\Type{\a}}{\ruleRef{type},(\ref{typeT1}),(\ref{typePGCtx})}\nonumber\\
		\NDLinePTG{\PredPhiName{(\a\ \termC{t_1}\ \ldots\ \ \termC{t_n})}\ofT\PhiAppl{\a}\to\PhiAppl{\a}\to\bool}{\ref{PTappl},(\ref{typeT3last})}\nonumber
	\end{align}
	
	\subparagraph{\ruleRef{bool'}:}
	\begin{align}
		\NDLineT{}{\Ctx{\ctx}}{\byAss}\label{bool'A}\\
		\NDLinePT{}{\Ctx{\ctxT}}{\IH,(\ref{bool'A})}\label{bool'IH}\\
		\NDLinePT{}{\Type{\bool}}{\ruleRef{bool'},(\ref{bool'IH})}\label{bool'pre}\\
		\NDLinePT{}{\Type{\PhiAppl{\bool}}}{\ref{PTTpBool},(\ref{bool'pre})}\nonumber
	\end{align}
	$\pbool$ is just a notation of $\termEqB$ which is of type $\bool\to\bool\to\bool$ in \hol, as desired.
	
	\subparagraph{\ruleRef{pi}:}
	\begin{align}
		\NDLineTG{\Type{\A}}{\byAss}\label{pi1}\\
		\NDLineT{\concatCtx{\ctx}{\x \ofT\A}}{\Type{B}}{\byAss}\label{pi2}\\
		\NDLinePTG{\PhiAppl{\A}\;\type}{\IH,(\ref{pi1})}\label{pi3}\\
		\concatCtx{\concatCtx{\ctxT}{\x \ofT&\PhiAppl{\A}}}{\typingAss{\A}{\x}}\dedPT\PhiAppl{\typB}\;\type\QQQNegSp&&\QQuad\text{\IH,(\ref{pi2})}\label{pi4pre}\\
		\NDLinePTG{\PhiAppl{\typB}\;\type}{\QNegSp HOL types context independent,(\ref{pi4pre})}\label{pi4}\\
		\NDLinePTG{\PhiAppl{\A}\to \PhiAppl{\typB}\:\type}{\ruleRef{arrow},(\ref{pi3}),(\ref{pi4})}\label{pi5}\\
		\NDLinePTG{\PredPhiName{A}\ofT\PhiAppl{\A}\to\PhiAppl{\A}\to\bool}{\IH,(\ref{pi1})}\label{pi6}\\
		\NDLinePTG{\PredPhiName{B}\ofT\PhiAppl{\typB}\to\PhiAppl{\typB}\to\bool}{\IH,(\ref{pi2})}\label{pi7}\\
		\NDLinePTG{\PhiAppl{\piType{x}{\A}\typB}\;\type}{\ref{PTPitype},(\ref{pi5})}\nonumber\\
		\NDLinePTG{\PredPhiName{\left(\piType{x}{\A}\typB\right)}\ofT(\PhiAppl{\piType{x}{\A}\typB})\to(\PhiAppl{\piType{x}{\A}\typB})\to\bool\QQQQNegSp}{\QQQQuad\ref{PTPipred},(\ref{pi6}),(\ref{pi7})}\nonumber
	\end{align}
	
		\subparagraph{\ruleRef{Qtype}:}
		\small{
		\begin{align}
			\NDLineTG{\Type{\A}}{\byAss}\label{QA1}\\
			\NDLineTG{\r\ofT \piType{x_1}{\A}\piType{x_2}{\A}\bool}{\byAss}\label{QA2}\\
			%\NDLineTG{\r\text{ is equivalence}}{\byAss}\label{QA3}\\
			\NDLinePTG{\Type{\PhiAppl{\A}}}{\IH,(\ref{QA1})}\label{QIH1}\\
			\NDLinePTG{\PredPhiName{A}\ofT \PhiAppl{\A}\to \PhiAppl{\A}\to\bool}{\IH,(\ref{QA1})}\label{QIH2}\\
			\NDLinePTG{\PhiAppl{\r}\ofT \PhiAppl{\A}\to \PhiAppl{\A}\to\bool}{\IH,(\ref{QA2})}\label{QIH3}\\
			\NDLinePTG{\Type{\PhiAppl{\quot{\A}{\r}}}}{\ref{PTQtype},(\ref{QIH1})}\nonumber\\
			\NDLinePT{\concatCtx{\ctxT}{\x ,\y\ofT\PhiAppl{\A}}}{\PredPhi{\A}{\x}\ofT\bool}{\ruleRef{appl},\ruleRef{appl},(\ref{QIH2}),\ruleRef{var},\ruleRef{var}}\label{QApXX}\\
			\NDLinePT{\concatCtx{\ctxT}{\x ,\y\ofT\PhiAppl{\A}}}{\PredPhi{\A}{\y}\ofT\bool}{\ruleRef{appl},\ruleRef{appl},(\ref{QIH2}),\ruleRef{var},\ruleRef{var}}\label{QApYY}\\
			%\NDLinePT{\concatCtx{\ctxT}{\x ,\y\ofT\PhiAppl{\A}}}{\termEqT{\A}{\x}{\y}\ofT\bool}{\ruleRef{appl},\ruleRef{appl},(\ref{QIH2}),\ruleRef{var},\ruleRef{var}}\label{QApXY}\\
			\NDLinePT{\concatCtx{\ctxT}{\x ,\y\ofT\PhiAppl{\A}}}{\PhiAppl{\r}\ \x \ \x \ofT\bool}{\ruleRef{appl},\ruleRef{appl},(\ref{QIH3}),\ruleRef{var},\ruleRef{var}}\label{Qrxy}\\
			\NDLinePT{\concatCtx{\ctxT}{\x ,\y\ofT\PhiAppl{\A}}}{\PhiAppl{\r}\ \x \ \y \land \PredPhi{\A}{\x}\land \PredPhi{\A}{\y}\ofT \bool\QQNegSp}{\ruleRef{andTp},(\ref{Qrxy}),\ruleRef{andTp},(\ref{QApXX}),(\ref{QApYY})}\label{Qconjs}\\
			%\NDLinePT{\concatCtx{\ctxT}{\x ,\y\ofT\PhiAppl{\A}}}{\termEqT{A}{\x}{\x}\lor\left(\PhiAppl{\r}\ \x \ \y \land \PredPhi{\A}{\x}\land \PredPhi{\A}{\y}\right) \ofT\bool\QQQNegSp}{\QQuad definition of $\lor$,(\ref{QApXY}),(\ref{Qconjs})}\label{Qrxy'}\\
			\ctxT\dedPT& \lambdaFun{x,y}{\PhiAppl{\A}}\PhiAppl{\r}\ \x \ \y \land \PredPhi{\A}{\x}\land \PredPhi{\A}{\y} \ofT \PhiAppl{\A}\to \PhiAppl{\A}\to \bool\QQQQNegSp&&\nonumber\\&&&\text{\ruleRef{lambda},\ruleRef{lambda},(\ref{Qconjs})}\label{QTppre}\\
			\NDLinePTG{\PredPhiName{\left(\quot{\A}{\r}\right)}\ofT \PhiAppl{\left(\quot{\A}{\r}\right)}\to \PhiAppl{\left(\quot{\A}{\r}\right)}\to \bool}{\ref{PTQtype},\ref{PTQpred},(\ref{QTppre})}\nonumber
		\end{align}
	}
		
		\subparagraph{\ruleRef{psubType}:}
		\begin{align}
			%\NDLineTG{\Type{\A}}{\byAss}\label{psubType1}\\
			\NDLineTG{\p\ofT\piType{x}{\A}\bool}{\byAss}\label{psubTypeA}\\
			\NDLinePTG{\PhiAppl{\p}\ofT \PhiAppl{\A}\to\bool}{\IH,(\ref{psubTypeA})}\label{psubTypeIH}\\
			\NDLinePTG{\Type{\PhiAppl{\A}\to\bool}}{\ruleRef{typingTp},(\ref{psubTypeIH})}\label{psubTypepTpWf}
			\intertext{Since statements of shape $\ded \Type{\typB\to\typC}$ only provable using rule \ruleRef{arrow}:}
			\NDLinePTG{\Type{\PhiAppl{\A}}}{see above,(\ref{psubTypepTpWf})}\label{psubTypepre1}\\
			\NDLinePTG{\Type{\PhiAppl{\subtype{\typeC{A}}{\p}}}}{\ref{PTPStype},(\ref{psubTypepre1})}\nonumber\\
			\NDLinePTG{\univQuant{x,y}{\PhiAppl{\A}}\termEqT{A}{\x}{\y}\implC \PhiAppl{\p}\ \x \termEqB \PhiAppl{\p}\ \y}{\IH,\ref{PTPipred},(\ref{psubTypeA})}\label{psubTypeIH2}\\
			\NDLinePT{\concatCtx{\ctxT}{\varname{x,y}\ofT \PhiAppl{\A}}}{\termEqT{A}{\x}{\y}\implC \PhiAppl{\p}\ \x \termEqB \PhiAppl{\p}\ \y}{\ruleRef{assDed},\ruleRef{forallE},\ruleRef{assDed},}\nonumber\\&&&\text{\ruleRef{forallE},\ruleRef{assDed},(\ref{psubTypeIH2}),\ruleRef{var},\ruleRef{var}}\label{psubTypeReverseForallI}\\
		\NDLinePT{\concatCtx{\ctxT}{\varname{x,y}\ofT \PhiAppl{\A}}}{\termEqT{A}{\x}{\y}\ofT\bool}{\ruleRef{implTypingL},(\ref{psubTypeReverseForallI})}\label{psubTypeArelApplWt}\\
		\NDLinePT{\concatCtx{\ctxT}{\varname{x,y}\ofT \PhiAppl{\A}}}{\PredPhiName{A}\ofT \PhiAppl{\A}\to\PhiAppl{\A}\to\bool}{\ruleRef{applType},\ruleRef{var},\ruleRef{applType},\ruleRef{var},(\ref{psubTypeArelApplWt})}\label{psubTypeFakeIH2C}
		\end{align}
		Since $\x ,\y$ don't occur in $\PredPhiName{A}$ and \hol types are context independent:
		\begin{align}
		\NDLinePTG{\PredPhiName{A}\ofT\PhiAppl{\A}\to\PhiAppl{\A}\to\bool}{see above,(\ref{psubTypeFakeIH2C})}\label{psubTypeFakeIH2}\\
		\NDLinePTG{\PredPhiName{\left(\subtype{\A}{\p}\right)}\ofT\PhiAppl{\A}\to\PhiAppl{\A}\to\bool}{\ref{PTPSpred},(\ref{psubTypeFakeIH2})}\nonumber
		\end{align}
	
	\subsubsection{Type-equality}
	Type-equality can be shown using the rules \ruleRef{congBase'},\ruleRef{subtE}, \ruleRef{congPi} and \ruleRef{congBool}:
	Observe that by the rules \ruleRef{varDed}, \ruleRef{congAppl}, \ruleRef{var}, instead of proving $\concatCtx{\ctxT}{\x,\y\ofT\PhiAppl{\A}}\dedPT \termEqT{A}{\x}{\y}\termEqB\termEqT{A}{\x}{\y}$ we may simply prove $\ctxT\dedPT \PredPhiName{A}\termEquals{\PhiAppl{\A}\to\PhiAppl{\A}\to\bool}\termEqT{A'}{\x}{\y}$.

	\subparagraph{\ruleRef{congBase'}:}
	\begin{align}
		&\thyIn{\a\ofT\piType{x_1}{\typeC{A_1}}\ldots\piType{x_n}{\typeC{A_n}}\type}{\theorycolor{T}}&&\text{by assumption}\label{congBase0}\\
		\NDLineTG{\termC{s_1}\termEquals{\typeC{A_1}}\termC{t_1}}{by assumption}\label{congBase1}\\
		\vdots\nonumber\\
		\NDLineTG{\termC{s_n}\termEquals{\subst{\typeC{A_n}}{x_1}{\termC{t_1}}\ldots\substOp{x_{n-1}}{\termC{t_{n-1}}}}\termC{t_n}}{\byAss}\label{congBasen}\\
		%\NDLineT{}{\Ctx{\ctx}}{\byAss}\label{congBaseGCtx}\\
		&\thyIn{\a\ofT\type}{\PhiAppl{\theorycolor{T}}}&&\text{\ref{PTTpConstr},(\ref{congBase0})}\label{congBaseATpl}\\
		&\thyIn{\PredPhiName{a}\ofT\PhiAppl{\typeC{A_1}}\to\ldots \to \PhiAppl{\typeC{A_n}}\to\PhiAppl{\a}\to\PhiAppl{\a}\to\bool}{\PhiAppl{\theorycolor{T}}}\quad&&\text{\ref{PTTpConstr},(\ref{congBase0})}\label{congBaseAPred}\\
		\NDLinePTG{\PhiAppl{\termC{s_1}}\termEquals{\PhiAppl{\typeC{A_1}}}\PhiAppl{\termC{t_1}}}{\IH,(\ref{congBase1})}\label{congBase1T}\\
		\vdots\nonumber\\	
		\NDLinePTG{\PhiAppl{\termC{s_n}}\termEquals{\PhiAppl{\typeC{A_n}}}\PhiAppl{\termC{t_n}}}{\IH,(\ref{congBasen})}\label{congBasenT}\\
		\NDLinePT{}{\Ctx{\ctxT}}{\ruleRef{tpCtx},\ruleRef{typingTp},\ruleRef{eqTyping},(\ref{congBase1T})}\label{congBasePGCtx}\\
		%\NDLinePT{}{\Ctx{\ctxT}}{\IH,(\ref{congBaseGCtx})}\label{congBasePGCtx}\\
		%\NDLinePTG{\a\ofT\;\type}{\ruleRef{type},(\ref{congBaseATpl}),(\ref{congBasePGCtx})}\label{congBaseATp}\\
		\NDLinePTG{\a\typeEquals \a}{\ruleRef{congBase},(\ref{congBasePGCtx}),(\ref{congBaseATpl})}\label{congBaseT}\\
		\NDLinePTG{\PredPhiName{a}\termEquals{\PhiAppl{\typeC{A_1}}\to\ldots \to \PhiAppl{\typeC{A_n}}\to\PhiAppl{\a}\to\PhiAppl{\a}\to\bool}\PredPhiName{a}}{\ruleRef{refl},\ruleRef{constS},(\ref{congBaseAPred}),(\ref{congBasePGCtx})}\label{congBaseApredEqApred}\\
		\NDLinePTG{\PredPhiName{a}\ \PhiAppl{\termC{s_1}}\termEquals{\PhiAppl{\typeC{A_2}}\to\ldots \to \PhiAppl{\typeC{A_n}}\to\PhiAppl{\a}\to\PhiAppl{\a}\to\bool}\PredPhiName{a}\ \PhiAppl{\termC{t_1}}}{\ruleRef{congAppl},(\ref{congBase1T}),(\ref{congBaseApredEqApred})}\label{congBasePredApp1}\\
		\vdots\nonumber\\	
		\NDLinePTG{\PredPhiName{a}\ \PhiAppl{\termC{s_1}}\ \ldots\ \PhiAppl{\termC{s_n}}\termEquals{\PhiAppl{\a}\to\PhiAppl{\a}\to\bool}\PredPhiName{a}\ \PhiAppl{\termC{t_1}}\ \ldots\ \PhiAppl{\termC{t_n}}}{\ruleRef{congAppl},(\ref{congBasenT}),previous line}\label{congBasePredAppn}\\
		\NDLinePTG{\PhiAppl{\a\ \termC{s_1}\ \ldots\ \termC{s_n}}\typeEquals \PhiAppl{\a\ \termC{t_1}\ \ldots\ \termC{t_n}}}{\ref{PTTpAppl},(\ref{congBaseT})}\nonumber\\
		\NDLinePTG{\dedPT\PredPhiName{(\a\ \termC{s_1}\ \ldots\ \termC{s_n})}\termEquals{\PhiAppl{\a}\to\PhiAppl{\a}\to\bool}\PredPhiName{(\a\ \termC{t_1}\ \ldots\ \termC{t_n})}\QQNegSp}{\ref{PTTpPredAppl},(\ref{congBasePredAppn})}\nonumber
	\end{align}
	
	\subparagraph{\ruleRef{subtE}:}
	\begin{align}
		\NDLineTG{\A\subtyping\Ap}{\byAss}\label{subtEA1}\\
		\NDLineTG{\Ap\subtyping\A}{\byAss}\label{subtEA2}\\
		\NDLinePTG{\PhiAppl{\A}\typeEquals\PhiAppl{\Ap}}{\IH,(\ref{subtEA1})}\label{subtEIH1}\\
		\NDLinePT{\concatCtx{\ctx}{\x ,\y\ofT \PhiAppl{\A}}}{\termEqT{A}{\x}{\y}\implC \termEqT{A'}{\x}{\y}}{\IH,(\ref{subtEA1})}\label{subtEIH2}\\
		\NDLinePT{\concatCtx{\ctxT}{\x \ofT\PhiAppl{\A}}}{\termEqT{A}{\x}{\x}\implC \termEqT{A'}{\x}{\x}}{\ruleRef{forallE},\ruleRef{forallI},(\ref{subtEIH2}),\ruleRef{var}}\label{subtEIH2'}\\
		\NDLinePT{\concatCtx{\ctx}{\x ,\y\ofT \PhiAppl{\Ap}}}{\termEqT{A'}{\x}{\y}\implC \termEqT{A}{\x}{\y}}{\IH,(\ref{subtEA2})}\label{subtEIH3}\\
		\NDLinePT{\concatCtx{\ctxT}{\x \ofT\PhiAppl{\Ap}}}{\termEqT{A'}{\x}{\x}\implC \termEqT{A}{\x}{\x}}{\ruleRef{forallE},\ruleRef{forallI},(\ref{subtEIH3}),\ruleRef{var}}\label{subtEIH3'}\\
		\ctxT\dedPT\univQuant{x}{\PhiAppl{\A}}&\termEqT{A'}{\x}{\x}\implC \termEqT{A}{\x}{\x}&&\text{\ruleRef{congDed},\ruleRef{forallCong},\ruleRef{tpEqTrans},(\ref{subtEIH1}),\ruleRef{refl},\ruleRef{forallI},(\ref{subtEIH3'})}\label{subtEIH3''}\\
		\NDLinePT{\concatCtx{\ctxT}{\x \ofT\PhiAppl{\A}}}{\termEqT{A'}{\x}{\x}\implC \termEqT{A}{\x}{\x}}{\ruleRef{forallE},\ruleRef{varDed},(\ref{subtEIH3''}),\ruleRef{var}}\label{subtEIH3'''}\\
		\NDLinePT{\concatCtx{\ctxT}{\x \ofT\PhiAppl{\A}}}{\termEqT{A}{\x}{\x}\termEqB \termEqT{A'}{\x}{\x}}{\ruleRef{propExt},(\ref{subtEIH2'}),(\ref{subtEIH3'''})}\nonumber
	\end{align}
	
	\subparagraph{\ruleRef{congPi}:}
	{\small
		\begin{align}
			\NDLineTG{\A\typeEquals \Ap}{\byAss}\label{congPi1}\\
			\NDLineT{\concatCtx{\ctx}{\x \ofT\A}}{\typB \typeEquals \Bp}{\byAss}\label{congPi2}\\
			\NDLinePTG{\PhiAppl{\A}\typeEquals\PhiAppl{\Ap}}{\IH,(\ref{congPi1})}\label{congPi3}\\
			\NDLinePTG{\PredPhiName{A}\termEquals{\PhiAppl{\A}\to\PhiAppl{\A}\to\bool}\PredPhiName{A'}}{\IH,(\ref{congPi1})}\label{congPi7}\\
			\NDLinePT{\concatCtx{\concatCtx{\ctxT}{\x \ofT\PhiAppl{\A}}}{\typingAss{\A}{\x}}}{\PhiAppl{\typB}\typeEquals\PhiAppl{\Bp}}{\IH,(\ref{congPi2})}\label{congPi4}
		\end{align}
		Since $\typeEquals$ is context independent in HOL:
		\begin{align}
			\NDLinePTG{\PhiAppl{\typB}\typeEquals\PhiAppl{\Bp}}{explanation,(\ref{congPi4})}\label{congPi5}\\
			\NDLinePTG{\PhiAppl{\A}\to\PhiAppl{\typB}\typeEquals\PhiAppl{\Ap}\to\PhiAppl{\Bp}}{\ruleRef{congTo},(\ref{congPi3}),(\ref{congPi5})}\label{congPi6}\\
			\NDLinePTG{\PhiAppl{\piType{x}{\A}\typB}\typeEquals\PhiAppl{\piType{x}{A'}\Bp}}{\ref{PTPitype},(\ref{congPi6})}\nonumber\\
			\NDLinePT{\concatCtx{\concatCtx{\ctxT}{\x \ofT\PhiAppl{\A}}}{\typingAss{\A}{\x}}}{\PredPhiName{B}\termEquals{\PhiAppl{\typB}\to\PhiAppl{\typB}\to\bool}\PredPhiName{B'}}{\QNegSp\IH,(\ref{congPi2})}\label{congPi8}\\
			\concatCtx{\concatCtx{\ctxT}{\termf\ofT\PhiAppl{\A}\to\PhiAppl{\typB}}}{\x\ofT&\PhiAppl{\A}}\dedPT\PredPhi{A}{\x}\impl \PredPhi{B}{(\termf\ \x )}&&\nonumber\\&\termEqB \PredPhi{A'}{\x}\impl \PredPhi{B}{(\termf\ \x )}&&\text{\ruleRef{rewrite},\ruleRef{refl},(\ref{congPi7})}\label{congPi9}\\
			\concatCtx{\concatCtx{\ctxT}{\termf\ofT\PhiAppl{\A}\to\PhiAppl{\typB}}}{\x\ofT&\PhiAppl{\A}}\dedPT\PredPhi{A}{\x}\impl \PredPhi{B}{(\termf\ \x )}&&\nonumber\\&\termEqB \PredPhi{A'}{\x}\impl \PredPhi{B'}{(\termf\ \x )}&&\text{\ruleRef{rewrite},(\ref{congPi9}),(\ref{congPi8})}\label{congPi10}\\
			\NDLinePT{\concatCtx{\ctxT}{\termf\ofT\PhiAppl{\A}\to\PhiAppl{\typB}}}{\univQuant{x}{\PhiAppl{\A}} \PredPhi{A}{\x}\impl (\PredPhi{B}{(\termf\ \x )})\termEqB&&\nonumber\\&
				\univQuant{x}{\PhiAppl{\Ap}} \PredPhi{A'}{\x} \impl (\PredPhi{B'}{(\termf\ \x )})}{\ruleRef{forallCong},(\ref{congPi3}),(\ref{congPi10})}\label{congPi11}\\
			\NDLinePTG{\PredPhiName{\left(\piType{x}{\A}\typB\right)}\termEquals{\PhiAppl{\A}\to\PhiAppl{\A}\to\bool}\PredPhiName{\left(\piType{x}{A'}\Bp\right)}}{\ref{PTLam},\ruleRef{congLam},(\ref{congPi11})}\nonumber
		\end{align}
	}
	
	\subparagraph{\ruleRef{congBool}:}
	\begin{align}
		\NDLineT{}{\Ctx{\ctx}}{\byAss}\label{congBoolA}\\
		\NDLinePT{}{\Ctx{\ctxT}}{\IH,(\ref{congBoolA})}\label{congBoolIH}\\
		\NDLinePTG{\Type{\bool}}{\ruleRef{congB},(\ref{congBoolIH})}\nonumber
	\end{align}
	$\pbool$ is a well-typed relation on $\bool$ by definition.
	
	\subsubsection{Subtyping}
	Subtyping can be shown using the axiom \ruleRef{ax:quotcod}.
	
	\subparagraph{\ruleRef{ax:quotcod}:}
	We need to check that the translation of axiom \ruleRef{ax:quotcod} holds in HOL. \ruleRef{ax:quotcod} states that whenever either side is well-formed, we have:
	\[\ded
	\piType{x}\A\quot\typB\r \subtyping
	\quot{(\piType x\A\typB)}
	{\lambdaFun{f,g}{\piType{x}{\A}\typB}\univQuant{x}{\A}\r\ \left(\termf\ \x\right)\left(\termg\ \x\right)}
	\]
	If either side is well-formed it follows that $\A, \typB$ are well-formed and $\r$ is equivalence relation on $\A$.
	We then need to prove that $\PhiAppl{\piType{x}\A\quot\typB\r}\typeEquals\PhiAppl{\quot{(\piType x\A\typB)}
		{\lambdaFun{f,g}{\piType{x}{\A}\typB}\univQuant{x}{\A}\r\ \left(\termf\ \x\right)\left(\termg\ \x\right)}}$ holds, which is immediate from the definition of the translation (both sides are just $\PhiAppl{\A}\to\PhiAppl{\typB}$) and that in a context containing $\x,\y\ofT\PhiAppl{\typB}$ we have $$\termEqT{\left(\piType{x}\A\quot\typB\r\right)}{\x}{\y}\implC \termEqT{\left(\quot{(\piType x\A\typB)}
		{\lambdaFun{f,g}{\piType{x}{\A}\typB}\univQuant{x}{\A}\r\ \left(\termf\ \x\right)\left(\termg\ \x\right)}\right)}{\x}{\y}.$$
	However, we have already shown in Example~\ref{exam:persQuotCod} that both PER applications reduce to the same formula, so the implication must be valid in HOL.
	
	\subsubsection{Typing}
	Typing can be shown using the rules \ruleRef{const''}, \ruleRef{var''}, \ruleRef{quotE},
	\ruleRef{lambda'}, \ruleRef{appl'}, \ruleRef{implType'}, \ruleRef{eqType'}, \ruleRef{psubI}, \ruleRef{psubE1}, \ruleRef{QI}:
	
	\subparagraph{\ruleRef{const''}:}
	\begin{align}
		&\thyIn{\c\ofT\Ap}{T}&&\text{\byAss}\label{const1}\\
		\NDLineTG{\Ap \typeEquals \A}{\byAss}\label{const2}\\
		&\thyIn{\c\ofT\PhiAppl{\Ap}}{\PhiAppl{T}}&&\text{\ref{PTtermDecl},(\ref{const1})}\label{const3}\\
		&\thyIn{\typingAx{A'}{\c}}{\PhiAppl{T}}&&\text{\ref{PTtermDecl},(\ref{const1})}\label{const6}\\
		\NDLinePTG{\PhiAppl{\Ap}\typeEquals\PhiAppl{\A}}{\IH,(\ref{const2})}\label{const4}\\
		\NDLinePT{\concatCtx{\ctxT}{\x \ofT\PhiAppl{\Ap}}}{\PredPhi{A'}{\x}\termEqB \PredPhi{A}{\x}}{\IH,(\ref{const2})}\label{const4a}\\
		\NDLinePTG{\univQuant{x}{\PhiAppl{\A}}\PredPhi{A'}{\x}\termEqB \PredPhi{A}{\x}}{\ruleRef{forallI},\ruleRef{forallI},(\ref{const4a})}\label{const4b}\\
		\NDLinePTG{\constname{c}\ofT\PhiAppl{\A}}{\ruleRef{const},(\ref{const3}),(\ref{const4})}\label{const5}\\
		\NDLinePTG{\PhiAppl{\c}\ofT\PhiAppl{\A}}{\ref{PTtermDecl},(\ref{const5})}\nonumber\\
		\NDLinePTG{\PredPhi{A'}{ \PhiAppl{\c}}}{\ref{PTtermDecl},\ruleRef{axiom},(\ref{const6})}\label{const7}\\
		\NDLinePTG{\PredPhi{A}{\PhiAppl{\c}}}{\ruleRef{congDed},\ruleRef{forallE},(\ref{const4b}),(\ref{const7})}\nonumber
	\end{align}
	
	\subparagraph{\ruleRef{var''}:}
	\begin{align}
		&\ctxIn{\x \ofT\Ap}{\ctx}&&\text{\byAss}\label{var1}\\
		\NDLineTG{\Ap \typeEquals \A}{\byAss}\label{var2}\\
		&\ctxIn{\x\ofT\PhiAppl{\Ap}}{\ctxT}&&\text{\ref{PTtermDecl},(\ref{var1})}\label{var3}\\
		&\ctxIn{\typingAss{A'}{\x}}{\ctxT}&&\text{\ref{PTtermDecl},(\ref{var1})}\label{var6}\\
		\NDLinePTG{\PhiAppl{\Ap}\typeEquals\PhiAppl{\A}}{\IH,(\ref{var2})}\label{var4}\\
		\NDLinePT{\concatCtx{\ctxT}{\x \ofT\PhiAppl{\Ap}}}{\PredPhi{A'}{\x}\termEqB \PredPhi{A}{\x}}{\IH,(\ref{var2})}\label{var4a}\\
		\NDLinePTG{\univQuant{x}{\PhiAppl{\A}}\PredPhi{A'}{\x}\termEqB \PredPhi{A}{\x}}{\ruleRef{forallI},(\ref{var4a})}\label{var4b}\\
		\NDLinePTG{\x\ofT\PhiAppl{\A}}{\ruleRef{var},(\ref{var3}),(\ref{var4})}\label{var5}\\
		\NDLinePTG{\PhiAppl{\x}\ofT\PhiAppl{\A}}{\ref{PTtermDecl},(\ref{var5})}\nonumber\\
		\NDLinePTG{\PredPhi{A'}{ \PhiAppl{\x}}}{\ref{PTtermDecl},\ruleRef{assume},(\ref{var6})}\label{var7}\\
		\NDLinePTG{\PredPhi{A}{\PhiAppl{\x}}}{\ruleRef{congDed},\ruleRef{forallE},(\ref{var4b}),(\ref{var7})}\nonumber
	\end{align}

	\subparagraph{\ruleRef{quotE}:}
	\begin{align}
		\NDLineTG{\s\ofT\quot\A\r}{\byAss}\label{QEA1}\\
		\NDLineT{\ctx,\,\x\ofT\A,\,\namedass{xrs}{\x\termEquals{\quot\A\r}\s}}{\tm\ofT \typB}{\byAss}\label{QEA2}\\
		\NDLineT{\ctx,\,\x\ofT\A,\,\xp\ofT\A,\,\namedass{xrs}{\x\termEquals{\quot\A\r}\s},\,\namedass{xprs}{\xp\termEquals{\quot\A\r}\s}}{\tm\termEquals{\typB} \tm\substOp\x\xp\QQNegSp}{\QQuad\byAss}\label{QEA3}\\
		\NDLinePTG{\PhiAppl{\s}\ofT\PhiAppl{\A}}{\IH,(\ref{QEA1})}\label{QEIH1}\\
		\NDLinePT{\ctxT,\,\x\ofT\PhiAppl{\A},\,\typingAss{\A}{\x},\,\namedass{xrs}{\termEqT{\left(\quot\A\r\right)}{\PhiAppl{\x}}{\PhiAppl{\s}}}}{\PhiAppl{\tm}\ofT\PhiAppl{\typB}}{\IH,(\ref{QEA2})}\label{QEIH2}\\
		%\NDLinePT{\ctxT,\,\x\ofT\PhiAppl{\A},\,\typingAss{\A}{\x},\,\namedass{xrs}{\termEqT{\left(\quot\A\r\right)}{\PhiAppl{\x}}{\PhiAppl{\s}}}}{\PredPhi{\typB}{\PhiAppl{\tm}}}{\IH,(\ref{QEA2})}\label{QEIH3}\\
		\ctxT,\,\x\ofT\A,\,\typingAss{\A}{\x},\,\xp\ofT\A,\,\typingAss{\A}{\xp},\,&\nonumber\\\namedass{xrs}{\termEqT{\left(\quot\A\r\right)}{\x}{\PhiAppl{\s}}},\,\namedass{xprs}{\termEqT{\left(\quot\A\r\right)}{\xp}{\PhiAppl{\s}}}\dedPT& 
		\termEqT{B}{\PhiAppl{\tm}}{\subst{\PhiAppl{\tm}}{x}{\xp}}\QQNegSp&&\QQuad\text{\IH,(\ref{QEA3})}\label{QEIH4}
	\end{align}
	Since typing is context independent in HOL:
	\begin{align}
		\NDLinePT{\ctxT,\,\x\ofT\PhiAppl{\A}}{\PhiAppl{\tm}\ofT\PhiAppl{\typB}}{explanation,(\ref{QEIH2})}\label{QEIH2'}\\
		\NDLinePTG{\PhiAppl{\tm}\substOp{\x}{\PhiAppl{\s}}\ofT \PhiAppl{\typB}}{\ruleRef{rewriteTyping},(\ref{QEIH2'}),(\ref{QEIH1})}\label{QEpre1}\\
		\NDLinePTG{\PhiAppl{\tm}\substOp{\x}{\PhiAppl{\s}}\ofT \PhiAppl{\typB\substOp\x\s}}{HOL types are simple,(\ref{QEpre1})}\nonumber%\label{QEconcl1}
	\end{align}
	Since $\PredPhiName{B}$ is transitive we can simplify (\ref{QEIH4}) to:
	\begin{align}
		\ctxT,\,\x\ofT\A,\,\typingAss{\A}{\x},\,&\nonumber\\\namedass{xrs}{\termEqT{\left(\quot\A\r\right)}{\x}{\PhiAppl{\s}}}\dedPT& 
		\termEqT{B}{\PhiAppl{\tm}}{\subst{\PhiAppl{\tm}}{x}{\PhiAppl{\s}}}&&\text{explanation,(\ref{QEIH4})}\label{QEIH4'}
	\end{align}
	By symmetry and transitivity of $\PredPhiName{B}$, we yield also $\termEqT{B}{\subst{\PhiAppl{\tm}}{x}{\PhiAppl{\s}}}{\subst{\PhiAppl{\tm}}{x}{\PhiAppl{\s}}}$ in the same context. Since this formula no longer depends on $\x$ and an (known to be well-typed) equality assumption with an otherwise unused variable on one side is not useful for proving in HOL, the same must also be derivable in context $\ctxT$.
	\begin{align}
	\NDLinePTG{\termEqT{B}{\subst{\PhiAppl{\tm}}{x}{\PhiAppl{\s}}}{\subst{\PhiAppl{\tm}}{x}{\PhiAppl{\s}}}}{explanation,(\ref{QEIH4'})}\nonumber
	\end{align}
	
	\subparagraph{\ruleRef{lambda'}:}	
	{\small
		\begin{align}
			\NDLineT{\concatCtx{\ctx}{\x \ofT\PhiAppl{\A}}}{\tm\ofT\typB}{\byAss}\label{lambdap1}\\
			\NDLineTG{\A\typeEquals\Ap}{\byAss}\label{lambdapA2}\\
			\NDLinePT{\concatCtx{\concatCtx{\ctx}{\x \ofT\PhiAppl{\A}}}{\typingAss{\A}{\x}}}{\PhiAppl{\tm}\ofT\PhiAppl{\typB}}{\IH,\ref{PTctxVar},(\ref{lambdap1})}\label{lambdap2}\\
			\NDLinePTG{\PhiAppl{\A}\typeEquals\PhiAppl{\Ap}}{\IH,(\ref{lambdapA2})}\label{lambdapIH2}\\
			\NDLinePTG{\PredPhiName{\A}\termEquals{\PhiAppl{\A}\to\PhiAppl{\A}\to\bool}\PredPhiName{\Ap}}{\IH,(\ref{lambdapA2})}\label{lambdapIH3}\\
			\NDLinePT{\concatCtx{\concatCtx{\ctx}{\x \ofT\PhiAppl{\A}}}{\typingAss{\A}{\x}}}{\PredPhi{B}{\PhiAppl{\tm}}}{\IH,\ref{PTctxVar},(\ref{lambdap1})}\label{lambdap5}\\
			\NDLinePT{\concatCtx{\concatCtx{\ctx}{\x ,\y\ofT\PhiAppl{\A}}}{\namedass{xRy}{\termEqT{A}{\x}{\y}}}}{\termEqT{B}{\PhiAppl{\tm}}{\subst{\PhiAppl{\tm}}{\x}{\y}}}{(\ref{correct:substRelatTerms}),(\ref{lambdap5})}\label{lambdap5b}\\
			\ctx\dedT \univQuant{x,y}{\PhiAppl{\A}}&\termEqT{A}{\x}{\y}\impl \termEqT{B}{\PhiAppl{\tm}}{\subst{\PhiAppl{\tm}}{\x}{\y}} \QQQNegSp&&\QQuad \text{\ruleRef{forallI},\ruleRef{implI},(\ref{lambdap5b})}\label{lambdap6}\\
			\ctx\dedT \univQuant{x,y}{\PhiAppl{\A}}&\termEqT{A'}{\x}{\y}\impl \termEqT{B}{\PhiAppl{\tm}}{\subst{\PhiAppl{\tm}}{\x}{\y}} \QQQNegSp&&\QQuad \text{\ruleRef{rewrite},\ruleRef{lambdap6},(\ref{lambdapIH3})}\label{lambdap6'}\\
			\NDLinePT{\concatCtx{\ctx}{\x \ofT\A}}{\PhiAppl{\tm}\ofT\PhiAppl{\typB}}{\QQNegSp\QNegSp typing independent of assumptions,(\ref{lambdap2})}\label{lambdap3}\\
			\NDLinePTG{(\lambdaFun{x}{\PhiAppl{\A}}\PhiAppl{\tm})\ofT\PhiAppl{\A}\to\PhiAppl{\typB}}{\ruleRef{lambda},(\ref{lambdap3})}\label{lambdap4}
			\intertext{Since in HOL equal types are necessarily identical, it follows:}
			\NDLinePTG{(\lambdaFun{x}{\PhiAppl{\A}}\PhiAppl{\tm})\ofT\PhiAppl{\Ap}\to\PhiAppl{\typB}}{explanation,(\ref{lambdap4}),(\ref{lambdapIH2})}\label{lambdappre1}\\
			\NDLinePTG{\PhiAppl{\lambdaFun{x}{\A}t}\ofT\PhiAppl{\piType{x}{\Ap}\typB}\QNegSp}{\QQuad\ref{PTLam},\ref{PTPitype},(\ref{lambdappre1})}\nonumber\\
			\NDLinePTG{\PredPhi{(\piType{x}{\Ap}\typB)}\ \PhiAppl{\lambdaFun{x}{\A}t}\QQNegSp}{\QQuad\ref{PTPipred},(\ref{lambdap6'})}\nonumber
		\end{align}
	}
	
	\subparagraph{\ruleRef{appl'}:}
	\begin{align}
		\NDLineTG{\termf\ofT\piType{x}{\A}\typB}{\byAss}\label{applp1}\\
		\NDLineTG{\tm\ofT\A}{\byAss}\label{applp2}\\
		\NDLinePTG{\PhiAppl{\termf}\ofT\PhiAppl{\A}\to\PhiAppl{\typB}}{\IH,\ref{PTPitype},(\ref{applp1})}\label{applp3}\\
		\NDLinePTG{\PredPhi{(\piType{x}{\A}\typB)}{\PhiAppl{\termf}}}{\IH,\ref{PTPitype},(\ref{applp1})}\label{applp6}\\
		\NDLinePTG{\univQuant{x}{\PhiAppl{\A}}\univQuant{y}{\PhiAppl{\A}}\nonumber\\&
			\termEqT{A}{\x}{\y}\impl \termEqT{B}{(\PhiAppl{\termf}\ \x )}{(\PhiAppl{\termf}\ \y)}}{\ref{PTPipred},(\ref{applp6})}\label{applp7}\\
		\NDLinePTG{\PhiAppl{\tm}\ofT\PhiAppl{\A}}{\IH,(\ref{applp2})}\label{applp4}\\
		\NDLinePTG{\PredPhi{A}{\PhiAppl{\tm}}}{\IH,(\ref{applp2})}\label{applp8}\\
		\NDLinePTG{\PredPhi{A}{\PhiAppl{\tm}}\impl \PredPhi{B}{(\PhiAppl{\termf}\ \PhiAppl{\tm})}}{\ruleRef{forallE},\ruleRef{forallE},(\ref{applp7}),(\ref{applp4}),(\ref{applp4})}\label{applp9}\\
		\NDLinePTG{\PredPhi{B}{(\PhiAppl{\termf}\ \PhiAppl{\tm})}}{\ruleRef{implE},(\ref{applp9}),(\ref{applp8})}\label{applp10}\\
		\NDLinePTG{\PhiAppl{\termf}\ \PhiAppl{\tm}\ofT\PhiAppl{\typB}}{\ruleRef{appl},(\ref{applp3}),(\ref{applp4})}\label{applp5}\\
		\NDLinePTG{\PhiAppl{\termf\ \tm}\ofT\PhiAppl{\typB}}{\ref{PTappl},(\ref{applp5})}\nonumber\\
		\NDLinePTG{\PredPhi{B}{\PhiAppl{\termf\ \tm}}}{\ref{PTappl},(\ref{applp9})}\nonumber
	\end{align}
	
	\subparagraph{\ruleRef{implType'}:}
	\begin{align}
		\NDLineTG{\termF\ofT\bool}{\byAss}\label{implType1}\\
		\NDLineT{\concatCtx{\ctx}{\namedass{ass}{\termF}}}{\termC{G}\ofT\bool}{\byAss}\label{implType2}\\
		\NDLinePTG{\PhiAppl{\termF}\ofT\bool}{\IH,(\ref{implType1})}\label{implType3}\\
		%\NDLinePTG{\boolPred{\PhiAppl{\termF}}}{\IH,(\ref{implType1})}\label{implType7}\\
		\NDLinePT{\concatCtx{\ctxT}{\PhiAppl{\termF}}}{\PhiAppl{\termC{G}}\ofT\bool}{\IH,(\ref{implType2})}\label{implType4}\\
		%\NDLinePT{\concatCtx{\ctxT}{\PhiAppl{\termF}}}{\boolPred{\PhiAppl{\termC{G}}}}{\IH,(\ref{implType2})}\label{implType8}\\
		\NDLinePTG{\PhiAppl{\termC{G}}\ofT\bool}{typing is independent of assumptions,(\ref{implType4})}\label{implType5}\\
		%\NDLinePTG{\PhiAppl{\termF}\impl (\boolPred{\PhiAppl{\termC{G}}})}{\ruleRef{implI},(\ref{implType8})}\label{implType9}\\
		\NDLinePTG{\PhiAppl{\termF}\impl \PhiAppl{\termG}\ofT\bool}{\ruleRef{implType},(\ref{implType3}),(\ref{implType5})}\label{implType6}\\
		\NDLinePTG{\PhiAppl{\termF\impl \termG}\ofT\bool}{\ref{PTImpl},(\ref{implType6})}\label{implType7}\\
		\NDLinePTG{\boolPred{\PhiAppl{\termF \impl \termG}}}{(\ref{PTTpBoolPred}),\ruleRef{refl},(\ref{implType7})}\nonumber
	\end{align}
	
	\subparagraph{\ruleRef{eqType'}:}
	\begin{align}
		\NDLineTG{\s\ofT\A}{\byAss}\label{eqType1}\\
		\NDLineTG{\tm\ofT\A}{\byAss}\label{eqType2}\\
		\NDLinePTG{\PhiAppl{\s}\ofT\PhiAppl{\A}}{\IH,(\ref{eqType1})}\label{eqType3}\\
		\NDLinePTG{\PhiAppl{\tm}\ofT\PhiAppl{\A}}{\IH,(\ref{eqType2})}\label{eqType4}\\
		\NDLinePTG{\PredPhi{A}{\PhiAppl{\s}}}{\IH,(\ref{eqType1})}\label{eqType5}\\
		\NDLinePTG{\PredPhi{A}{\PhiAppl{\s}}\ofT\bool}{\ruleRef{validTyping},(\ref{eqType5})}\label{eqType5Tp}\\	
		%\NDLinePTG{\PredPhi{A}{\PhiAppl{\tm}}}{\IH,(\ref{eqType2})}\label{eqType6}\\
		%\NDLinePTG{\PredPhi{A}{\PhiAppl{\s}}\land \PredPhi{A}{\PhiAppl{\tm}}}{\ruleRef{andI},(\ref{eqType5}),(\ref{eqType6})}\label{eqType9}\\
		%\NDLinePTG{\PredPhi{A}{\PhiAppl{\s}}\ofT\bool}{\ruleRef{validTyping},(\ref{eqType5})}\label{eqType7}\\	
		%\NDLinePTG{\PredPhi{A}{\PhiAppl{\tm}}\ofT\bool}{\ruleRef{validTyping},(\ref{eqType6})}\label{eqType8}\\
		%\NDLinePTG{\PhiAppl{\s}\termEquals{\PhiAppl{\A}}\PhiAppl{\tm}\ofT\bool}{\ruleRef{eqType},(\ref{eqType3}),(\ref{eqType4})}\label{eqType10}
		\NDLinePTG{\PredPhiName{A}\ofT\PhiAppl{\A}\to\PhiAppl{\A}\to\bool}{\ruleRef{applType},(\ref{eqType5}),\ruleRef{applType},(\ref{eqType5}),(\ref{eqType5Tp})}\label{eqType12}\\
		\NDLinePTG{\termEqT{A}{\PhiAppl{\s}}{\PhiAppl{\tm}}\ofT\bool}{\ruleRef{appl},\ruleRef{appl},(\ref{eqType12}),(\ref{eqType3}),(\ref{eqType4})}\label{eqType13}\\
		\NDLinePTG{\boolPred{\left(\termEqT{A}{\PhiAppl{\s}}{\PhiAppl{\tm}}\right)}}{(\ref{PTTpBoolPred}),\ruleRef{refl},(\ref{eqType13})}\nonumber
	\end{align}

	\subparagraph{\ruleRef{psubI}:}
	\begin{align}
		\NDLineTG{\tm\ofT\A}{\byAss}\label{psubI1}\\
		\NDLineTG{\p\ t}{\byAss}\label{psubI2}\\
		\NDLinePTG{\PhiAppl{\tm}\ofT\PhiAppl{\A}}{\IH,(\ref{psubI1})}\label{psubI3}\\
		\NDLinePTG{\PredPhi{A}{\PhiAppl{\tm}}}{\IH,(\ref{psubI1})}\label{psubI4}\\
		\NDLinePTG{\PhiAppl{p}\ \PhiAppl{\tm}}{\IH,(\ref{psubI2})}\label{psubI5}\\
		\NDLinePTG{\PredPhi{\left(\subtype{\A}{\p}\right)}{\PhiAppl{\tm}}}{\ref{PTPSpred},\ruleRef{andI},(\ref{psubI4}),\ruleRef{andI},(\ref{psubI5}),(\ref{psubI5})}\nonumber\\
		\NDLinePTG{\PhiAppl{\tm}\ofT\PhiAppl{\subtype{\A}{\p}}}{\ref{PTPStype},(\ref{psubI3})}\nonumber
	\end{align}
	\subparagraph{\ruleRef{psubE1}:}
	\begin{align}
		\NDLineTG{\tm\ofT\subtype{\A}{\p}}{\byAss}\label{psubE1A}\\
		\NDLinePTG{\PhiAppl{\tm}\ofT\PhiAppl{\A}}{\IH,(\ref{psubE1A})}\nonumber\\
		\NDLinePTG{\PredPhi{A}{\PhiAppl{\tm}}\land\PhiAppl{\p}\ \PhiAppl{\tm}}{\IH,(\ref{psubE1A})}\label{psubE1IH2}\\
		\NDLinePTG{\PredPhi{A}{\PhiAppl{\tm}}}{\ruleRef{andEl},(\ref{psubE1IH2})}\nonumber
	\end{align}
	
	\subparagraph{\ruleRef{QI}:}
	\begin{align}
		\NDLineTG{\tm\ofT\A}{\byAss}\label{QIA1}\\
		\NDLineTG{\isEqRel{\r}}{\byAss}\label{QIA2}\\
		\NDLinePTG{\PhiAppl{\tm}\ofT\PhiAppl{\A}}{\IH,(\ref{QIA1})}\label{QIIH1}\\
		\NDLinePTG{\PhiAppl{\tm}\ofT\PhiAppl{(\quot{\A}{\r})}}{\ref{PTPSpred},(\ref{QIIH1})}\nonumber\\
		\NDLinePTG{\PredPhi{A}{\PhiAppl{\tm}}}{\IH,(\ref{QIA1})}\label{QIIH2}
		\intertext{As shown as Subsection~\ref{subsec:substLemma} (\ref{QIA2}) implies that $\PhiAppl{\r}$ is an equivalence on terms \x satisfying $\PredPhi{\A}{\x}$. It follows that $\PhiAppl{\r}\ \PhiAppl{\tm}\ \PhiAppl{\tm}$ holds.}
		%\NDLinePTG{\left(\PhiAppl{\r}\ \PhiAppl{\tm}\ \PhiAppl{\tm}\land \PredPhi{A}{\PhiAppl{\tm}}\land\PredPhi{A}{\PhiAppl{\tm}}\right)}{definition of $\land$,(\ref{QIIH1}),(\ref{QIIH2})}\label{QIIHs}\\
		\NDLinePTG{\PhiAppl{\r}\ \PhiAppl{\tm}\ \PhiAppl{\tm}}{explanation,(\ref{QIIH2})}\label{QIrtt}\\
		\NDLinePTG{\PhiAppl{\r}\ \PhiAppl{\tm}\ \PhiAppl{\tm} \land \PredPhi{A}{\PhiAppl{\tm}}\land\PredPhi{A}{\PhiAppl{\tm}}}{definition of $\land$,(\ref{QIrtt}),(\ref{QIIH2})}\label{QIpre}\\
		%\NDLinePTG{\left(\PhiAppl{\r}\ \PhiAppl{\tm}\ \PhiAppl{\tm}\lor \PredPhi{A}{\PhiAppl{\tm}}\right)\land\PredPhi{A}{\PhiAppl{\tm}}\land\PredPhi{A}{\PhiAppl{\tm}}}{\ruleRef{andI},\ruleRef{andI},(\ref{QI4}),(\ref{QIIH2}),(\ref{QIIH2})}\label{QIpre}\\
		\NDLinePTG{\PredPhi{\left(\quot{\A}{\r}\right)}{\PhiAppl{\tm}}}{\ref{PTQpred},(\ref{QIpre})}\nonumber
	\end{align}

	\subsubsection{Term equality}\label{par:soundnessTmEq}
	Fix a context. 
	By rule \ruleRef{rewrite}, if we can show for two \dhole terms $\s,\tm\ofT\A$ that $\PhiAppl{\s}\termEquals{\PhiAppl{\A}}\PhiAppl{\tm}$ and additionally that $\PredPhi{A}{\PhiAppl{\s}}$, then $\PredPhi{A}{\PhiAppl{\tm}}$ and $\termEqT{A}{\PhiAppl{\s}}{\PhiAppl{\tm}}$ follow. 
	By rule \ruleRef{eqTyping} and rule \ruleRef{sym} we further yield $\PhiAppl{\s}\ofT\PhiAppl{\A}$ and $\PhiAppl{\tm}\ofT\PhiAppl{\A}$.
	This reduces the completeness claim for a term-equality $\s\termEquals{\A}t$ to showing $\PhiAppl{\s}\termEquals{\PhiAppl{\A}}\PhiAppl{\tm}$ and $\PredPhi{A}{\PhiAppl{\s}}$.
	
	Term equality can be shown using the rules  \ruleRef{congLam'}, \ruleRef{congAppl'}, \ruleRef{refl'}, \ruleRef{sym'}, \ruleRef{beta'}, \ruleRef{etaPi} and \ruleRef{QEq} in DHOL.
	
	\subparagraph{\ruleRef{congLam'}}
	This case will use (\ref{correct:substRelatTerms}). 
	
	{\small
		\begin{align}
			\NDLineTG{\A\typeEquals \Ap}{\byAss}\label{congLamp1}\\
			\NDLineT{\concatCtx{\ctx}{\x \ofT\A}}{\tm\termEquals{\typB}\termtp}{\byAss}\label{congLamp2}\\
			\NDLinePTG{\PhiAppl{\A}\typeEquals\PhiAppl{\Ap}}{\IH,(\ref{congLamp1})}\label{congLamp3}\\	
			\NDLinePT{\concatCtx{\concatCtx{\ctxT}{\x \ofT\PhiAppl{\A}}}{\typingAss{\A}{\x}}}{\termEqT{B}{\PhiAppl{\tm}}{\PhiAppl{\termtp}}}{\IH,(\ref{congLamp2})}\label{congLamp4}\\
			\NDLinePT{\concatCtx{\concatCtx{\ctxT}{\z \ofT\PhiAppl{\A}}}{\typingAss{\A}{\z}}}{\termEqT{B}{\subst{\PhiAppl{\tm}}{\x}{\z}}{\subst{\PhiAppl{\termtp}}{x}{\z}}}{$\alpha$-renaming,(\ref{congLamp4})}\label{congLamp5}\\
			\NDLinePT{\concatCtx{\ctxT}{\z \ofT\PhiAppl{\A}}}{\PredPhi{A}{\z}\impl\nonumber\\&\termEqT{B}{\subst{\PhiAppl{\tm}}{x}{\z}}{\subst{\PhiAppl{\termtp}}{x}{\z}}}{\ruleRef{implI},(\ref{congLamp5})}\label{congLamp5a}\\
			\NDLinePTG{\univQuant{z}{\PhiAppl{\A}} \PredPhi{A}{\z}\impl\nonumber\\&\termEqT{B}{\subst{\PhiAppl{\tm}}{\x}{\z}}{\subst{\PhiAppl{\termtp}}{x}{\z}}}{\ruleRef{forallI},(\ref{congLamp5a})}\label{congLamp5b}\\
			\NDLinePT{\concatCtx{\ctxT}{\x ,\y\ofT\PhiAppl{\A}}}{\univQuant{z}{\PhiAppl{\A}} \PredPhi{A}{\z}\impl\nonumber\\&\termEqT{B}{\subst{\PhiAppl{\tm}}{x}{\z}}{\subst{\PhiAppl{\termtp}}{x}{\z}}}{\ruleRef{varDed},\ruleRef{varDed},(\ref{congLamp5b})}\label{congLamp5c}\\
			\NDLinePT{\concatCtx{\ctxT}{\x ,\y\ofT\PhiAppl{\A}}}{\PredPhi{A}{\x}\impl\termEqT{B}{\PhiAppl{\tm}}{\PhiAppl{t'}}}{\ruleRef{varDed},\ruleRef{varDed},(\ref{congLamp5c})}\label{congLamp5d}\\
			\NDLinePT{\concatCtx{\concatCtx{\ctxT}{\x ,\y\ofT\PhiAppl{\A}}}{\namedass{xRy}{\termEqT{A}{\x}{\y}}}}{\PredPhi{A}{\x}\impl\termEqT{B}{\PhiAppl{\tm}}{\PhiAppl{\termtp}}}{\ruleRef{assDed},(\ref{congLamp5d})}\label{congLamp5e}\\
			\NDLinePT{\concatCtx{\concatCtx{\ctxT}{\x ,\y\ofT\PhiAppl{\A}}}{\namedass{xRy}{\termEqT{A}{\x}{\y}}}}{\PredPhi{A}{\x}}{(\ref{correct:substRelatTerms}),\ruleRef{assume},\ruleRef{assume}}\label{congLamp6}\\
			\NDLinePT{\concatCtx{\concatCtx{\ctxT}{\x ,\y\ofT\PhiAppl{\A}}}{\namedass{xRy}{\termEqT{A}{\x}{\y}}}}{\termEqT{B}{\PhiAppl{\tm}}{\PhiAppl{\termtp}}}{\ruleRef{implE},(\ref{congLamp5e}),(\ref{congLamp6})}\label{congLamp7}\\
			\NDLinePT{\concatCtx{\concatCtx{\ctxT}{\x ,\y\ofT\PhiAppl{\A}}}{\namedass{xRy}{\termEqT{A}{\x}{\y}}}}{\termEqT{B}{\PhiAppl{\tm}}{\subst{\PhiAppl{\termtp}}{\x}{\y}}}{(\ref{correct:substRelatTerms}),(\ref{congLamp7}),\ruleRef{assume}}\label{congLamp8}\\
			\NDLinePT{\concatCtx{\ctxT}{\x ,\y\ofT\PhiAppl{\A}}}{\termEqT{A}{\x}{\y}\impl \termEqT{B}{\PhiAppl{\tm}}{\subst{\PhiAppl{\termtp}}{\x}{\y}}}{\ruleRef{implI},(\ref{congLamp8})}\label{congLamp9}\\
			\NDLinePTG{\univQuant{x}{\PhiAppl{\A}}\univQuant{y}{\PhiAppl{\A}}\termEqT{A}{\x}{\y}\nonumber\\&\impl \termEqT{B}{\PhiAppl{\tm}}{\subst{\PhiAppl{\termtp}}{\x}{\y}}}{\ruleRef{forallI},\ruleRef{forallI},(\ref{congLamp9})}\label{congLamp10}\\
			\NDLinePT{\concatCtx{\concatCtx{\ctxT}{\x \ofT\PhiAppl{\A}}}{\typingAss{\A}{\x}}}{\PhiAppl{\tm}\ofT\PhiAppl{\typB}}{\IH,(\ref{congApplp2})}\label{congApplp11}\\
			\NDLinePT{\concatCtx{\concatCtx{\ctxT}{\x \ofT\PhiAppl{\A}}}{\typingAss{\A}{\x}}}{\PhiAppl{t'}\ofT\PhiAppl{\typB}}{\IH,(\ref{congApplp2})}\label{congApplp12}\\
			\intertext{Since in HOL typing is independent of context assumptions:}
			\NDLinePT{\concatCtx{\ctxT}{\x \ofT\PhiAppl{\A}}}{\PhiAppl{\tm}\ofT\PhiAppl{\typB}}{explanation,(\ref{congApplp11})}\label{congApplp13}\\
			\NDLinePT{\concatCtx{\ctxT}{\x \ofT\PhiAppl{\A}}}{\PhiAppl{\termtp}\ofT\PhiAppl{\typB}}{explanation,(\ref{congApplp12})}\label{congApplp14}\\
			\NDLinePTG{\lambdaFun{x}{\PhiAppl{\A}}\PhiAppl{\tm}\ofT\PhiAppl{\A}\to\PhiAppl{\typB}}{\ruleRef{lambda},(\ref{congApplp13})}\label{congApplp15}\\
			\NDLinePTG{\lambdaFun{x}{\PhiAppl{\A}}\PhiAppl{\termtp}\ofT\PhiAppl{\A}\to\PhiAppl{\typB}}{\ruleRef{lambda},(\ref{congApplp14})}\label{congApplp16}\\
			\NDLinePTG{\PhiAppl{\lambdaFun{x}{\A}\tm\termEquals{\piType{x}{\A}\typB}\lambdaFun{x}{\Ap}\termtp}\QNegSp}{\Quad(\ref{PTEq}),(\ref{congLamp10})}\nonumber\\
			\NDLinePTG{\PhiAppl{\lambdaFun{x}{\A}\tm}\ofT\PhiAppl{\piType{x}{\A}\typB}}{\Quad(\ref{PTPitype}),(\ref{PTLam}),(\ref{congApplp15})}\nonumber\\
			\NDLinePTG{\PhiAppl{\lambdaFun{x}{\A}\termtp}\ofT\PhiAppl{\piType{x}{\A}\typB}}{\Quad(\ref{PTPitype}),(\ref{PTLam}),(\ref{congApplp16})}\nonumber
		\end{align}
	}
	
	\subparagraph{\ruleRef{congAppl'}:}
	\begin{align}
		\NDLineTG{\tm\termEquals{\A}\termtp}{\byAss}\label{congApplp1}\\
		\NDLineTG{\termf\termEquals{\piType{x}{\A}\typB}\termfp}{\byAss}\label{congApplp2}\\
		\NDLinePTG{\termEqT{A}{\PhiAppl{\tm}}{\PhiAppl{\termtp}}}{\IH,(\ref{congApplp1})}\label{congApplp3}\\
		\NDLinePTG{\univQuant{x}{\PhiAppl{\A}}\univQuant{y}{\PhiAppl{\A}}\termEqT{A}{\x}{\y}\impl\nonumber\\& \termEqT{(\piType{z}{\A}\typB)}{\PhiAppl{\termf}\ \x}{\PhiAppl{\termfp}\ \y}}{\IH,(\ref{congApplp2})}\label{congApplp4}\\
		\NDLinePTG{\PhiAppl{\tm}\ofT\PhiAppl{\A}}{\IH,(\ref{congApplp1})}\label{congApplp5}\\
		\NDLinePTG{\PhiAppl{\termtp}\ofT\PhiAppl{\A}}{\IH,(\ref{congApplp1})}\label{congApplp6}\\
		\NDLinePTG{\termEqT{A}{\PhiAppl{\tm}}{\PhiAppl{\termtp}}\impl \termEqT{(\piType{z}{\A}\typB)}{\PhiAppl{\termf}\ \PhiAppl{\tm}}{\PhiAppl{\typeC{f'}}\ \PhiAppl{\termtp}}}{\ruleRef{forallE},\ruleRef{forallE},(\ref{congApplp4}),(\ref{congApplp5}),(\ref{congApplp6})}\label{congApplp7}\\
		\NDLinePTG{\PhiAppl{\termf}\ofT\PhiAppl{\A}\to\PhiAppl{\typB}}{\IH,(\ref{congApplp2})}\label{congApplp8}\\
		\NDLinePTG{\PhiAppl{\termfp}\ofT\PhiAppl{\A}\to\PhiAppl{\typB}}{\IH,(\ref{congApplp2})}\label{congApplp9}\\
		\NDLinePTG{\termEqT{(\piType{z}{\A}\typB)}{\PhiAppl{\termf}\ \PhiAppl{\tm}}{\PhiAppl{\termfp}\ \PhiAppl{\termtp}}}{\ruleRef{implE},(\ref{congApplp7}),(\ref{congApplp3})}\nonumber\\
		\NDLinePTG{\PhiAppl{\termf}\ \PhiAppl{\tm}\ofT\PhiAppl{\typB}}{\ruleRef{appl},(\ref{congApplp8}),(\ref{congApplp5})}\nonumber\\
		\NDLinePTG{\PhiAppl{\termfp}\ \PhiAppl{\termtp}\ofT\PhiAppl{\typB}}{\ruleRef{appl},(\ref{congApplp8}),(\ref{congApplp5})}\nonumber
	\end{align}
	
	\subparagraph{\ruleRef{refl'}:}
	\begin{align}
		\NDLineTG{\tm\ofT\A}{\byAss}\label{refl1}\\
		\NDLinePTG{\PhiAppl{\tm}\ofT\PhiAppl{\A}}{\IH,(\ref{refl1})}\label{refl2}\\
		\NDLinePTG{\PhiAppl{\tm}\termEquals{\PhiAppl{\A}}\PhiAppl{\tm}}{\ruleRef{refl},(\ref{refl2})}\nonumber\\
		\NDLinePTG{\PredPhi{A}{\PhiAppl{\tm}}}{\IH,(\ref{refl1})}\nonumber
	\end{align}
	
	\subparagraph{\ruleRef{sym'}:}
	\begin{align}
		\NDLineTG{\s\termEquals{\A}\tm}{\byAss}\label{sym1}\\
		\NDLinePTG{\termEqT{A}{\PhiAppl{\s}}{\PhiAppl{\tm}}}{\IH,(\ref{sym1})}\label{sym2}\\
		\NDLinePTG{\PhiAppl{\tm}\ofT\PhiAppl{\A}}{\IH,(\ref{sym1})}\label{sym3}\\
		\NDLinePTG{\PhiAppl{\s}\ofT\PhiAppl{\A}}{\IH,(\ref{sym1})}\label{sym4}\\
		\NDLinePTG{\termEqT{A}{\PhiAppl{\tm}}{\PhiAppl{\s}}}{\ruleRef{forallE},\ruleRef{forallE},\ruleRef{implE},(\ref{correct:relatSym}),(\ref{sym2}),(\ref{sym3}),(\ref{sym4})}\nonumber
	\end{align}
	
	\subparagraph{\ruleRef{beta'}:}
	\begin{align}
		\NDLineTG{(\lambdaFun{x}{\A}\s)\ \tm\ofT\typB}{\byAss}\label{beta1}\\
		\NDLinePTG{(\lambdaFun{x}{\PhiAppl{\A}}\PhiAppl{\s})\ \PhiAppl{\tm}\ofT\PhiAppl{\typB}}{\IH,\ref{PTLam},(\ref{beta1})}\label{beta2}\\
		\NDLinePTG{(\lambdaFun{x}{\PhiAppl{\A}}\PhiAppl{\s})\ \PhiAppl{\tm}\termEquals{\PhiAppl{\typB}}\subst{\PhiAppl{\s}}{x}{\PhiAppl{\tm}}}{\ruleRef{beta},(\ref{beta2})}\label{betaPConc}\\
		\NDLinePTG{(\lambdaFun{x}{\PhiAppl{\A}}\PhiAppl{\s})\ \PhiAppl{\tm}\termEquals{\PhiAppl{\typB}}\PhiAppl{\subst{s}{x}{\tm}}}{(\ref{correct:substTerm}),(\ref{betaPConc})}\label{betaConc}\\
		\NDLinePTG{\PhiAppl{(\lambdaFun{x}{\termC{A}}\s)\ \tm}\termEquals{\PhiAppl{\typB}}\PhiAppl{\subst{\s}{x}{\tm}}\QQNegSp}{\QQuad\ref{PTLam},\ref{PTappl},(\ref{betaConc})}\nonumber\\
		\NDLinePTG{\PredPhi{\left(\piType{x}{\A}\typB\right)}{\left((\lambdaFun{x}{\A}\s)\ \tm\right)}\QQNegSp}{\QQuad\IH,(\ref{beta1})}\nonumber
	\end{align}
	
	\subparagraph{\ruleRef{etaPi}:}
	\begin{align}
		\NDLineTG{\tm\ofT\piType{x}{\A}\typB}{\byAss}\label{etaPi1}\\
		\NDLinePTG{\PhiAppl{\tm}\ofT\PhiAppl{\A}\to\PhiAppl{\typB}}{\ref{PTPitype},\IH,(\ref{etaPi1})}\label{etaPi2}\\
		\NDLinePTG{\PhiAppl{\tm}\termEquals{\PhiAppl{\A}\to\PhiAppl{\typB}}\lambdaFun{x}{\PhiAppl{\A}}\PhiAppl{\tm}\ \x}{\ruleRef{eta},(\ref{etaPi2})}\label{etaPiConc}\\
		\NDLinePTG{\PhiAppl{\tm}\termEquals{\PhiAppl{\piType{x}{\A}\typB}}\PhiAppl{\lambdaFun{x}{\A}\tm\ \x}}{\ref{PTLam},\ref{PTPitype},(\ref{etaPiConc})}\nonumber\\
		\NDLinePTG{\PredPhi{\left(\piType{x}{\A}\typB\right)}{\PhiAppl{\tm}}}{\IH,(\ref{etaPi2})}\nonumber
	\end{align}

	\subparagraph{\ruleRef{psubEq}:}
	\begin{align}
		\NDLineTG{\s\termEquals{\A}\tm}{\byAss}\label{psubEqA1}\\
		\NDLineTG{\p\ \s}{\byAss}\label{psubEqA2}\\
		\NDLinePTG{\termEqT{A}{\PhiAppl{\s}}{\PhiAppl{\tm}}}{\IH,(\ref{psubEqA1})}\label{psubEqIH1}\\
		\NDLinePTG{\PhiAppl{\p}\ \PhiAppl{\s}}{\IH,(\ref{psubEqA2})}\label{psubEqIH2}
		\intertext{By (\ref{substLemAppl2}) it follows:}
		\NDLinePTG{\PhiAppl{\p}\ \PhiAppl{\tm}}{explanation,(\ref{psubEqA2}),(\ref{psubEqIH1})}\label{psubEqIH2'}\\
		\NDLinePTG{\termEqT{A}{\PhiAppl{\s}}{\PhiAppl{\tm}}\land\PhiAppl{\p}\ \PhiAppl{\s}\land\PhiAppl{\p}\ \PhiAppl{\tm}}{definition of $\land$,(\ref{psubEqIH1}),(\ref{psubEqIH2}),(\ref{psubEqIH2'})}\label{psubEqpre}\\
		\NDLinePTG{\termEqT{\left(\subtype{\A}{\p}\right)}{\s}{\tm}}{\ref{PTPSpred},(\ref{psubEqpre})}\nonumber\\
		\NDLinePTG{\PhiAppl{\s}\ofT\PhiAppl{(\quot{\A}{\p})}}{\ruleRef{validTyping},\ruleRef{andEl},\ruleRef{psubEqpre}}\nonumber
	\end{align}
	
	\subparagraph{\ruleRef{QEq}:}
	\newcommand{\NDLnQEq}[4]{\NDLinePT{\concatCtx{\ctxT}{#1}}{#2}{#3}\label{#4}}
	\begin{align}
		\NDLineTG{\s\ofT\A}{\byAss}\label{QEqA1}\\
		\NDLineTG{\tm\ofT\A}{\byAss}\label{QEqA2}\\
		\NDLineTG{\r\ofT\A\to\A\to\bool}{\byAss}\label{QEqA3}\\
		\NDLinePTG{\PhiAppl{\s}\ofT\PhiAppl{\A}}{\IH,(\ref{QEqA1})}\label{QEqIH1}\\
		\NDLinePTG{\PredPhi{A}{\PhiAppl{\s}}}{\IH,(\ref{QEqA1})}\label{QEqIH2}\\
		\NDLinePTG{\PhiAppl{\tm}\ofT\PhiAppl{\A}}{\IH,(\ref{QEqA2})}\label{QEqIH3}\\
		\NDLinePTG{\PredPhi{A}{\PhiAppl{\tm}}}{\IH,(\ref{QEqA1})}\label{QEqIH4}\\
		\NDLinePTG{\PhiAppl{\r}\ofT\PhiAppl{\A}\to\PhiAppl{\A}\to\bool}{\IH,(\ref{QEqA3})}\label{QEqIH5}\\
		\NDLinePTG{\PhiAppl{\r}\ \PhiAppl{\s}\ \PhiAppl{\tm}\ofT\bool}{\ruleRef{appl},\ruleRef{appl},(\ref{QEqIH5}),(\ref{QEqIH1}),(\ref{QEqIH3})}\label{QEqRstBool}\\
		\NDLinePTG{\PhiAppl{\r}\ \PhiAppl{\s}\ \PhiAppl{\tm}\termEqB \PhiAppl{\r}\ \PhiAppl{\s}\ \PhiAppl{\tm}}{\ruleRef{refl},(\ref{QEqRstBool})}\label{rstRefl}\\
		\NDLinePTG{\left(\PhiAppl{\r}\ \PhiAppl{\s}\ \PhiAppl{\tm}\land \PredPhi{A}{\PhiAppl{\s}}\land \PredPhi{A}{\PhiAppl{\tm}}\right)\termEqB \PhiAppl{\r}\ \PhiAppl{\s}\ \PhiAppl{\tm}}{definition of $\land$, \ruleRef{QEqIH2},(\ref{QEqIH4}),(\ref{rstRefl})}\label{QEqpre}\\
		\NDLinePTG{\PhiAppl{\s\termEquals{\left(\quot{\A}{\r}\right)}\tm}}{\ref{PTQpred},\ref{PTTpBoolPred},(\ref{QEqpre})}\nonumber\\
		\NDLinePTG{\PhiAppl{\s}\ofT\PhiAppl{\left(\quot{\A}{\r}\right)}}{\ref{PTQpred},(\ref{QEqIH3})}\nonumber
	\end{align}

	\subsubsection{Validity}
	Validity can be shown using the rules \ruleRef{axiom'}, \ruleRef{assume'}, \ruleRef{implI'}, \ruleRef{implE'}, \ruleRef{congDed'}, \ruleRef{boolExt'} and \ruleRef{psubE}.

	\subparagraph{\ruleRef{axiom'}}
	\begin{align}
		&\thyIn{\namedax{ax}{\termF}}{T}&&\text{\byAss}\label{axiom1}\\
		\NDLineT{}{\Ctx{\ctx}}{\byAss}\label{axiom2}\\
		&\thyIn{\namedax{ax}{\PhiAppl{\termF}}}{\PhiAppl{T}}&&\text{\ref{PTax},(\ref{axiom1}}\label{axiom3}\\
		\NDLinePT{}{\Ctx{\ctxT}}{\IH,\ref{axiom2}}\label{axiom4}\\
		\NDLinePTG{\PhiAppl{\termF}}{\ruleRef{axiom},(\ref{axiom3}),(\ref{axiom4})}\nonumber
	\end{align}
	
	\subparagraph{\ruleRef{assume'}}
	\begin{align}
		&\ctxIn{\namedass{ass}{\termF}}{\ctx}&&\text{\byAss}\label{assume1}\\
		\NDLineT{}{\Ctx{\ctx}}{\byAss}\label{assume2}\\
		&\ctxIn{\namedass{ass}{\PhiAppl{\termF}}}{\ctxT}&&\text{\ref{PTctxAss},(\ref{assume1}}\label{assume3}\\
		\NDLinePT{}{\Ctx{\ctxT}}{\IH,\ref{assume2}}\label{assume4}\\
		\NDLinePTG{\PhiAppl{\termF}}{\ruleRef{assume},(\ref{assume3}),(\ref{assume4})}\nonumber
	\end{align}
	
	\subparagraph{\ruleRef{implI'}}
	\begin{align}
		\NDLineTG{\termF\ofT\bool}{\byAss}\label{implI1}\\
		\NDLineT{\concatCtx{\ctx}{\namedass{ass}{\termF}}}{\termC{G}}{\byAss}\label{implI2}\\
		\NDLinePTG{\PhiAppl{\termF}\ofT\bool}{\IH,(\ref{implI1})}\label{implI3}\\
		\NDLinePT{\concatCtx{\ctxT}{\PhiAppl{\termF}}}{\PhiAppl{\termC{G}}}{\IH,\ref{PTctxAss},(\ref{implI2})}\label{implI4}\\
		\NDLinePTG{\PhiAppl{\termF}\impl\PhiAppl{\termC{G}}}{\ruleRef{implI},(\ref{implI3}),(\ref{implI4})}\label{implI5}\\
		\NDLinePTG{\PhiAppl{\termF\impl \termG}}{\ref{PTImpl},(\ref{implI5})}\nonumber
	\end{align}
	
	\subparagraph{\ruleRef{implE'}}
	\begin{align}
		\NDLineTG{\termF\impl \termG}{\byAss}\label{implE1}\\
		\NDLineTG{\termF}{\byAss}\label{implE2}\\
		\NDLinePTG{\PhiAppl{\termF}\impl \PhiAppl{\termG}}{\IH,\ref{PTImpl},(\ref{implE1})}\label{implE3}\\
		\NDLinePTG{\PhiAppl{\termF}}{\IH,(\ref{implE2})}\label{implE4}\\
		\NDLinePTG{\PhiAppl{\termG}}{\ruleRef{implE},(\ref{implE3}),(\ref{implE4})}\nonumber
	\end{align}
	
	\subparagraph{\ruleRef{congDed'}}
	\begin{align}
		\NDLineTG{\termF\termEqB \termFp}{\byAss}\label{congDed1}\\
		\NDLineTG{\termFp}{\byAss}\label{congDed2}\\
		\NDLinePTG{\PhiAppl{\termF}\termEqB\PhiAppl{\termFp}}{(\ref{PTTpBoolPred}),\IH,(\ref{congDed1})}\label{congDed3}\\
		\NDLinePTG{\PhiAppl{\termFp}}{\IH,(\ref{congDed2})}\label{congDed4}\\
		\NDLinePTG{\PhiAppl{\termF}}{\ruleRef{congDed},(\ref{congDed3}),(\ref{congDed4})}\nonumber
	\end{align}
	
	\subparagraph{\ruleRef{boolExt'}}
	\begin{align}
		\NDLineTG{\p\ \T}{\byAss}\label{boolExt1}\\
		\NDLineTG{\p\ \F}{\byAss}\label{boolExt2}\\
		\NDLinePTG{\PhiAppl{\p}\ \T}{\IH,\ref{PTappl},(\ref{boolExt1})}\label{boolExt3}\\
		\NDLinePTG{\PhiAppl{\p}\ \F}{\IH,\ref{PTappl},(\ref{boolExt2})}\label{boolExt4}\\
		\NDLinePTG{\univQuant{z}{\bool}\PhiAppl{\p}\ \z}{\ruleRef{boolExt},(\ref{boolExt3}),(\ref{boolExt4})}\label{boolExt5}\\
		\NDLinePT{\concatCtx{\ctxT}{\x \ofT\bool}}{\univQuant{z}{\bool}\PhiAppl{\p}\ \z}{\ruleRef{varDed},(\ref{boolExt5})}\label{boolExt6}\\
		\NDLinePT{\concatCtx{\ctxT}{\x \ofT\bool}}{\PhiAppl{\p}\ \x}{\ruleRef{forallE},(\ref{boolExt6}),\ruleRef{assume}}\label{boolExt7}\\
		%\NDLinePT{\concatCtx{\ctxT}{\x \ofT\bool}}{x\termEqB \x}{\ruleRef{refl},\ruleRef{varS}}\label{boolExt7a}\\
		%\concatCtx{\concatCtx{\ctxT}{\x \ofT\bool}}{x\termEqB& \x\impl \PhiAppl{p}\ \x}\dedT \PhiAppl{p}\ \x &&\text{\ruleRef{implE},\ruleRef{assume},(\ref{boolExt7a})}\label{boolExt7b}\\
		%\NDLinePT{\concatCtx{\concatCtx{\ctxT}{\x \ofT\bool}}{\PhiAppl{p}\ \x}}{x\termEqB \x\impl \PhiAppl{p}\ \x}{\ruleRef{implI},\ruleRef{assume}}\label{boolExt7c}\\
		%\NDLinePT{\concatCtx{\ctxT}{\x \ofT\bool}}{\left(x\termEqB \x\impl \PhiAppl{p}\ \x \right)\termEqB \PhiAppl{p}\ \x}{\ruleRef{propExt},(\ref{boolExt7b}),(\ref{boolExt7c})}\label{boolExt7d}\\
		%\NDLinePT{\concatCtx{\ctxT}{\x \ofT\bool}}{x\termEqB \x\impl\PhiAppl{p}\ \x}{\ruleRef{congDed},(\ref{boolExt7d}),\ruleRef{boolExt7}}\label{boolExt8}\\
		\concatCtx{\concatCtx{\ctxT}{\x ,\y\ofT&\bool}}{\namedass{xEqy}{\x \termEqB \y}}\dedT \PhiAppl{\p}\ \x &&\text{\ruleRef{assDed},\ruleRef{varDed},\ruleRef{boolExt7}}\label{boolExt9}\\	
		\concatCtx{\concatCtx{\ctxT}{\x ,\y\ofT&\bool}}{\namedass{xEqy}{\x \termEqB \y}}\dedT \PhiAppl{\p}\ \y&&\text{\ruleRef{rewrite},\ruleRef{boolExt9},\ruleRef{assume}}\label{boolExt10}\\	
		\concatCtx{\ctxT}{\x ,\y\ofT&\bool}\dedT \termEqT{\bool}{\x}{\y}\impl\PhiAppl{\p}\ \y&&\text{(\ref{PTTpBoolPred}),\ruleRef{implI},\ruleRef{boolExt10}}\label{boolExt11}\\
		\NDLinePTG{\univQuant{x}{\bool}\univQuant{y}{\bool} \nonumber\\&\termEqT{\bool}{\x}{\y}\impl\PhiAppl{\p}\ \y}{\ruleRef{forallI},\ruleRef{forallI},(\ref{boolExt11})}\label{boolExtPre}\\
		\NDLinePTG{\PhiAppl{\univQuant{x}{\bool}\p\ \x}}{(\ref{PTEq}),(\ref{PTPipred}),(\ref{boolExtPre})}\nonumber
	\end{align}
	
		\subparagraph{\ruleRef{psubE}}
		\begin{align}
			\NDLineTG{\tm\ofT\subtype{\A}{\p}}{\byAss}\label{psubE1C}\\
			\NDLinePTG{\PredPhi{\left(\subtype{\A}{\p}\right)}{\PhiAppl{\tm}}}{\IH,(\ref{psubE1C})}\label{psubE2a}\\		
			\NDLinePTG{\PhiAppl{\p}\ \PhiAppl{\tm}}{\ruleRef{andEr},\ruleRef{andEr},\ref{PTPSpred},(\ref{psubE2a})}\label{psubE2C}\\
			\NDLinePTG{\PhiAppl{\p\ t}}{\ref{PTappl},(\ref{psubE2C})}\nonumber
		\end{align}
\end{proof}

\section{Soundness proof}\label{appendix:Sound}
The idea of the soundness proof is to transform HOL-proofs into \dhol-proofs.

The key idea is to transform a HOL-proof of $\PhiAppl{\termC{F}}$ into one only using terms in the well-typed image of the translation, at which point we can read off a (partial) DHOL-proof of $\termC{F}$. 
This DHOL proof will be missing sub-proofs to showing the well-typedness of terms occurring in the proof. 
Using the (assumed) well-typedness of the DHOL problem and induction, we can show that all terms occurring in the partial DHOL proof are in fact well-typed, yielding the existence of a complete DHOL proof. 
We proceed in multiple steps:

\begin{enumerate}
	\item prove that the translation is injective for terms of equal given \dhol type,
	\item define quasi-preimages for terms not in image of translation,
	\item given valid HOL derivation of translation of well-typed validity conjecture, choose \dhol types of quasi-preimages of terms in it,
	\item modify derivation to make terms in it (almost) proper,
	\item lift modified HOL derivation to \dhol derivation.
\end{enumerate}

\subsection{Type-wise injectivity of the translation}\label{sec:injectivity}

\begin{definition}
	Let \tm be an ill-typed \dhole term with well-typed image $\PhiAppl{\tm}$ in \hol. In this case we will say that $\PhiAppl{\tm}$ is a \emph{spurious} term w.rt. its preimage $\tm$. If the preimage is unique or clear from the context we will simply say that $\PhiAppl{\tm}$ is spurious. 
	Similarly, a term $\PhiAppl{\s}$ in HOL that is the image of a well-typed term $\s$, will be called \emph{proper} w.r.t its preimage $\s$. A term $tm$ in HOL that is not the image of any (well-typed or not) term is said to be \emph{improper}. 
\end{definition}

\begin{lemma}\label{lem:termwise-inj}
	Let $\theorycolor{\Delta}$ be a \dhole context and let $\ctx$ denote its translation.
	Given two \dhole terms $\s, \tm$ of type $\A$ and assuming $\s$ and $\tm$ are not identical, it follows that $\PhiAppl{\s}$ and $\PhiAppl{\tm}$ are not identical.
\end{lemma}

\begin{proof}[Proof of Lemma~\ref{lem:termwise-inj}]
	We prove this by induction on the shape of the types both equalities are over\,---\,in case both terms are equalities\,---\,and by subinduction on the shape of the two translated terms otherwise.
	We observe that terms created using a different top-level production are non-identical and will remain that way in the image. 
	So we can go over the productions one by one and assuming type-wise injectivity for subterms show injectivity of applying them.
	Different constants are mapped to different constants and different variables to different variables, so in those cases there is nothing to prove. 
	If two function applications or implications differ in \dhole then one of the two pairs of corresponding arguments must differ as well. By induction hypothesis so will the images of the terms in that pair. Since function application and implication both commute with the translation, it follows that the images of the function applications or implications also differ. 
	%Analogously, if two equivalence classes (for quotients from the same \dhol type) differ, then either the equivalence relations or the term the classes are applied to are not identical and thus by induction hypothesis, so must their images.\protect\footnote{This nicely illustrates just how weak this theorem really is: After all while related terms are not identical and neither are their equivalence classes or the images of those classes, they will certainly be equal.}
	Since the translations of the terms on both sides of an equality also show up in the translation, the same argument also works for two equalities over the same type. Similarly for lambda functions
	% and choice operators 
	over the same type.
	
	Consider now two equalities over different types that get identified by dependency-erasure. 
	
	In case of equalities over different base types, the typing relations that are applied in the images are different, so the images of the equalities differ.
	For equalities over different $\Pi$-types either the domain type or the codomain type must differ by rule \ruleRef{congPi}. 
	If the domain types differ then the typing assumption after the two universal quantifiers of the translated equalities will differ. 
	If the codomain types are different then the applications of the typing relations on the right of the $\impl$ of the translated equalities are the translations of the equalities yielded by applying the functions on both sides of the equalities to a freshly bound variable of the domain type.
	The translations of the equalities are only identical if those "inner  equalities" are identical. Furthermore, the inner equalities are over types that are the codomain of the type the equalities are over. The claim then follows from the induction hypothesis. 
	
		Finally it remains to consider the case of equalities $\s\termEquals{\subtype{\A}{\p}}\tm$ and $\sp\termEquals{\subtype{A'}{\pp}}\termtp$ over non-identical refinement types $\subtype{\A}{\p}$ and $\subtype{A'}{\pp}$ where not both $\A=\Ap$ and $\p=\pp$. 
		If $\p\neq \pp$, then the translations have different subterms $\PhiAppl{p}\ \s$ and $\PhiAppl{\pp}\ \sp$ and thus differ.
		If $\A\neq \Ap$, then the first conjuncts in the translated equalities are the translations of equalities over the types $\A$ and $\Ap$ respectively, which by the induction hypothesis have different translations.
		So in any case, the equalities have different images.
		The case of equalities over quotient types works analogously.
\end{proof}

\subsection{Quasi-preimages for terms and validity statements in admissible HOL derivations}\label{sec:quasi-preimages}
Firstly, we will consider the preimage of a typing relations $\PredPhiName{A}$ to be the equality symbol $\lambdaFun{x}{\A}\lambdaFun{y}{\A} \x\termEquals{\A}\y$ (if equality is treated as a (parametric) binary predicate rather than a production of the grammar this eta reduces to the symbol $\termEquals{\A}$). 

Using this convention, we define the normalization of an improper HOL term, which is either a proper term or a spurious term.
The normalization of an improper HOL term is defined by:
\begin{definition}\label{defn:normalization}
	Let $\tm$ be an improper HOL term. Then we define the replacement $\norm{t}$ of $\tm$ by first matching line below:
	
	\begin{align*}
		\norm{\PhiAppl{\tm}}&:=\tm\plabel{PTnormProper}\\
		\norm{\norm{s}}&:=\norm{s}\plabel{PTnormNorm}\\	
		\norm{\PredPhiName{A}\ \s}&:=\lambdaFun{y}{\PhiAppl{\A}} \termEqT{A}{s}{y}\plabel{normRelAppl}\\
		\norm{\PredPhiName{A}}&:=\lambdaFun{x}{\PhiAppl{\A}}\lambdaFun{y}{\PhiAppl{\A}} \termEqT{A}{x}{y}\plabel{normRel}\\
		\norm{c}&:=\constname{c}\plabel{normConst}\\
		\norm{x}&:=\x\plabel{normVar}\\
		\norm{f\ \tm}&:=\norm{f}\ \norm{t}\plabel{normAppl}\\
		\norm{\lambdaFun{x}{\typC}\tm}&:=\lambdaFun{x}{\typC}\norm{t}\plabel{normLam}\\
		%\norm{\choiceOp{x}{\typC}{\p}}&:=\choiceOp{x}{\typC}{\norm{p}}\plabel{normChoice}
		\intertext{If $\termC{F}$ not of shape $\PredPhiName{\typeC{A}} \termC{x'}\ \x\impl \_ $ (for $\termC{x'}$ a variable bound in a universal quantifier whose body is this term) or of shape $\univQuant{x'}{\PhiAppl{\typeC{A}}}\termEqT{\typeC{A}}{x}{x'}\impl\_$:}
		\norm{\univQuant{x}{\PhiAppl{\typeC{A}}}\PredPhi{\A}{x}\implC \termF}&:=\univQuant{x,x'}{\PhiAppl{\typeC{A}}}\termEqT{\typeC{A}}{x}{x'}\implC \norm{\termF}\plabel{normAltTransUniv}\\
		\norm{\univQuant{x}{\PhiAppl{\typeC{A}}}\termF}&:=\univQuant{x,x'}{\PhiAppl{\typeC{A}}}\termEqT{\typeC{A}}{x}{x'}\implC \norm{\termF}\plabel{normUniv}
		\intertext{Otherwise $\univQuant{x}{\PhiAppl{\typeC{A}}}\termF$ is the translation of a universal quantifier and (\ref{PTnormProper}) applies.}
		\norm{\s\termEquals{\PhiAppl{\A}}\tm}&:=\termEqT{A}{s}{\tm} \plabel{normEq}\\
		\norm{s\impl \tm}&:=\norm{s}\impl \norm{t} \plabel{normImpl}
	\end{align*}
	For terms $\tm$ in the image of the translation, we define the normalization of $\tm$ be be $\tm$ itself.
\end{definition}

\begin{definition}	
	Assume a well-formed \dhole theory $\theorycolor{T}$. 
	
	We say that an HOL context $\theorycolor{\Delta}$ is \emph{proper} (relative to $\PhiAppl{\theorycolor{T}}$), iff there exists a well-formed HOL context  $\theorycolor{\Theta}$ (relative to $\PhiAppl{\theorycolor{T}}$), s.t. there is a well-formed \dhole context $\ctx$ (relative to $\theorycolor{T}$) with $\ctxT=\theorycolor{\Theta}$ and $\theorycolor{\Theta}$ can be obtained from $\theorycolor{\Delta}$ by adding well-typed typing assumptions.
	In this case, $\ctx$ is called a \emph{quasi-preimage} of $\theorycolor{\Delta}$.
	Inspecting the translation, it becomes clear that $\ctx$ is uniquely determined by the choices of the preimages of the types of variables without a typing assumption in $\theorycolor{\Delta}$.
	
	Given a proper HOL context $\theorycolor{\Delta}$ and a well-typed HOL formula $\termC{\phi}$ over $\theorycolor{\Delta}$, we say that $\termC{\phi}$ is \emph{quasi-proper} iff $\norm{\phi}=\PhiAppl{\termF}$ for $\ctx\dedT \termF:\bool$ 
	and $\ctx$ is a quasi-preimage of $\theorycolor{\Delta}$.
	In that case, we call $\termF$ a \emph{quasi-preimage} of $\termC{\phi}$.
	%It is intuitively clear that almost proper terms are quasi-proper, but we don't use this observation and therefore we don't prove it. 
	\par
	
	Finally, we call a validity judgement $\theorycolor{\Delta}\dedPT \termC{\phi}$ in HOL \emph{proper} iff\begin{enumerate}
		\item $\theorycolor{\Delta}$ is proper,
		\item $\termC{\phi}$ is quasi-proper in context $\theorycolor{\Delta}$
	\end{enumerate}
	In this case, we will call $\ctxT\dedPT \PhiAppl{\termF}$ a relativization of $\theorycolor{\Delta}\dedPT \termC{\phi}$ and $\ctx\dedT \termF$ a \emph{quasi-preimage} of the statement $\theorycolor{\Delta}\dedPT \termC{\phi}$, where $\ctx$ is a quasi-preimage of $\theorycolor{\Delta}$ and $\termF$ a quasi-preimage of $\termC{\phi}$. 
	Additionally, for HOL terms with preimages we consider these preimages to be quasi-preimages of the HOL term as well.
\end{definition}

\subsection{Transforming HOL derivations into admissible HOL derivations}\label{sec:prfTransform}
It will be useful to distinguish between two different kinds of improper terms.
\begin{definition}
	An improper term is called \emph{almost proper} iff its normalization isn't spurious (w.r.t. a given quasi-preimage) and contains no spurious subterms, otherwise it is said to be \emph{unnormalizably spurious}.
	This means that improper terms are almost proper iff their given quasi-preimage is well-typed.
	Since proper terms have well-typed preimages, they are almost proper (w.r.t. this preimage) as well.
\end{definition}
In order to lift a HOL derivation to DHOL, we first have to choose (quasi)-preimage types for all term occuring in it (at which point we use the notions of spurious and almost proper terms w.r.t. these DHOL types).

%\TODO{Split this section into a part dealing with type indices and a part dealing with making an indexed HOL derivation admissible.}
\subsubsection{Choosing (quasi-)preimage types for a HOL derivation}

\begin{lemma}[Indexing lemma]\label{lem:indexing}
	Assume that $\ctx \dedT \termF:\bool$ holds in DHOL. Given a valid HOL derivation $D$ of the statement $\ctxT\dedPT \PhiAppl{\termF}$, we can choose a DHOL type $T(\tm)$ (called type index) for each occurence of a HOL term $\tm$ in $D$, s.t. the following properties hold:
	\begin{enumerate}
		\item $T(\PhiAppl{\tm})=\Ap$ with $\PhiAppl{\A}=\PhiAppl{\Ap}$ for any DHOL term $\tm$ satisfying $\judg \tm:\A$,
		\item $T(\constname{c})=\A$ if $\constname{c}:\A$ is a constant in $\thy$,
		\item $T(\x)=\A$ if $\x:\A$ is variable declaration in $\ctx$,
		\item $T(\s)=T(\tm)$ for $\s, \tm$ within an equality of the form $\s\termEquals{\A}\tm$ for some HOL type $\A$,
		%\item $T(\x)=\A$ for $\x$ in $\tm:=\lambdaFun{x}{B}\s$ if $T(u)=\piType{x}{\A}\typC$ for some type $C$ and $\tm, \termC{u}$ are the terms on the two sides of a term equality in $D$,
		\item $T(\s)=T(\tm)=\bool$ for $\s, \tm$ within an implication of the form $\s\implC\tm$,
		\item $T(\x)=\A$ for $\x$ in $\left(\lambdaFun{x}{\B}\s\right)\ \tm$ if $T(\tm)=\A$,
		\item $T(\s\termEquals{\A}\tm)=\bool$,
		\item when variables are moved from the context into a $\lambda$-binder or vice versa the index of said variable is preserved
		\item whenever a term $\tm$ occurs both in the assumptions and conclusions of a step $S$ in $D$, the index of $\tm$ is the same all those occurrences of $\tm$ in $S$,
		\item if $\varname{x}, \termC{t}$ in $\lambdaFun{x}{\PhiAppl{\typeC{B}}}\tm$ satisfy $T(\varname{x})=\typeC{A}$ and $T(\termC{t})=\typeC{B}$, then $T(\lambdaFun{x}{\typeC{B}}\termC{t})=\piType{x}{A}\typeC{B}$,
		\item $T(\tm)=\A$ implies $\tm$ has type $\PhiA$ and $\A$ is non-empty.
	\end{enumerate}
\end{lemma}
\begin{proof}[Proof by induction of the shape of $D$]
	This lemma only holds for well-formed derivations of translations of well-typed conjectures over well-formed theories. It will not hold for arbitrary formulae (as can be seen by considering equalities between constants of equal HOL but different DHOL types). The proof of the lemma will therefore use that fact that the final statement in the derivation is the translation of a well-typed DHOL statement (which already determines the "correct" indices for the terms within that statement) and then show that for each step in a well-formed HOL proof concluding a statement that we can correctly index, the assumptions of that step can also be indexed correctly. We will thus proceed by "backwards induction" on the shape of the derivation $D$.
	
	The assumptions of the theory and contexthood rules only contain terms already contained in the conclusions, so the associated cases in the proof are all trivial.
	
	Similarly the assumptions of lookup rules, type well-formedness and type equality rules contain no additional terms, so those cases are also trivial.
	
	The typing rules are about forming larger terms from subterms, so if those larger terms can be consistently (i.e. according to the claim of the lemma) indexed, then the same is necessarily also true for the subterms. Thus the typing rules also have only trivial cases. 
	By the same arguments the cases for the congruence rules \ruleRef{congLam}, \ruleRef{congAppl}, the symmetry and transitivity rules for term equality and the rules \ruleRef{beta}, \ruleRef{eta}, \ruleRef{implType} and \ruleRef{implI} are also all trivial.
	
	It remains to consider the cases for the rules \ruleRef{implE}, \ruleRef{congDed}, \ruleRef{boolExt} and \ruleRef{nonempty}.
	
	\paragraph{Regarding \ruleRef{implE}:}
	Here the assumptions of the rule contain the additional terms $\termF$ and $\termF\implC \termC{G}$. However as both terms are of type $\bool$ all their (quasi)-preimages have type $\bool$ as well and picking indices according to any of the quasi-preimages of $\termF\implC \termC{G}$ will work.
	
	\paragraph{Regarding \ruleRef{congDed}:}
	Here the assumptions of the rule contain the additional terms $\termC{F'}$ and $\termF\termEqB \termC{F'}$. Both terms are of type $\bool$, so by the same argument as in the previous case, we can index them consistently with the claim of this lemma.
	
	\paragraph{Regarding \ruleRef{boolExt}:}
	Here the assumptions of the rule contain the additional terms $\p\ \T$ and $\p\ \F$ and $\T$ and $\F$. All these terms are of type $\bool$, so by the same argument as in the previous two cases, we can index them consistently with the claim of this lemma.
	
	\paragraph{Regarding \ruleRef{nonempty}:}
	Here the second assumption contains an additional variable of type $\A$. As this variable doesn't occur in any other terms, we can index it by the type of any of its (quasi)-preimages. 
	By assumption, the DHOL theory is partially inhabited, so we can chose this type index to be inhabited.
\end{proof}
\begin{rem}
	In the following, we will use the term \emph{type index} to refer to a choice of DHOL types for each HOL term in a derivation satisfying the properties of the previous lemma.
\end{rem}

\subsubsection{Transforming unnormalizably spurious terms into almost proper terms in HOL derivations}
\begin{definition}\label{defn:admissibleDerivation}
	A valid HOL derivation is called \emph{admissible} iff we can choose quasi-preimages for all terms occurring in it s.t. all terms in the derivation are almost proper w.r.t. their chosen quasi-preimages.
\end{definition}
This definition is useful, since admissible derivations are precisely those HOL derivations that allow us to consistently lift the terms occurring in them to well-typed DHOL terms.

In the following, we describe a proof transformation which maps HOL derivations to admissible HOL derivations. 

\begin{definition}\label{defn:statement-transformation}
	A \emph{statement transformation} in a given logic is a map that maps statements in the logic to statements in the logic.
	Similarly an \emph{indexed statement transformation} is a map that maps HOL statements with indexed terms to HOL statements.
\end{definition}
\begin{definition}\label{defn:macro-step}
	A \emph{macro-step} $M$ \emph{for} an (indexed) statement transformation \emph{$T$} \emph{replacing} a step \emph{$S$} in a derivation is a sequence of steps $S_1,\ldots, S_n$ (called \emph{micro-steps} of $M$) s.t. the assumptions of the $S_i$ that are not concluded by $S_j$ with $j<i$ are results of applying $T$ to assumptions of step $S$ and furthermore the conclusion of step $S_n$ is the result of applying $T$ to the conclusion of $S$. The assumptions of those $S_j$ that are not concluded by previous micro-steps of $M$ are called the \emph{assumptions of macro-step} $M$ and the conclusion of the last micro-step $S_n$ of $M$ is called the \emph{conclusion of macro-step} $M$. 
\end{definition}
Thus we we can replace each step in a derivation by a macro step replacing that step, we can transform that derivation to a derivation in which the given indexed statement transformation is applied to all statements. This is useful to simplify and normalize derivations to derivations with certain additional properties. 
In our case, we want to normalize a given HOL derivation into an admissible HOL derivation. Thus, we need to define an indexed statement transformation for which all terms in the image of the transformation are almost proper and the replace all steps in the derivation by macro steps for that statement transformation.

Since the notion of a quasi-proper term only makes sense once we fix a choice of type indices (in the sense of Lemma~\ref{lem:indexing}), the indexed statement transformation will actually depend on the choice of type indices. 

The idea for our transformation is that given a choice of type indices according to Lemma~\ref{lem:indexing}, the only remaining "source" of spuriousness are ill-typed function applications. There are two cases of such function applications, either they are applications that can be beta reduced to non-spurious terms or they can not. As it turns out in the second case the proof cannot meaningfully depend on that term, so we can simply replace it with a proper one. The below definition simply spells this out in detail.

\begin{definition}\label{defn:normalizingPrfTransform}
	A \emph{normalizing statement transformation} $\sRed{\cdot}$ is defined to be an indexed statement transformation that replaces terms in statements as described below. 
	The definition of the transformation of a term depends on its type index\,---\,a \dhole-type $\A$ (called \emph{preimage type})\,---\,for each term $\tm$.
	We will write those types as indices to the HOL terms, so for instance $\tm_{\A}$ indicates a HOL term $\tm$ of type $\PhiAppl{\A}$ and preimage type $\A$. 
	
	These preimage types are used to effectively associate to each term a type of a possible quasi-preimage (hence their name), which is useful as for $\lambda$-functions there are quasi-preimages of potentially many different types.
	We require that for an indexed term $\tm_{\A}$, term $\tm$ has type $\PhiAppl{\A}$ and that for almost proper terms $\tm_{\A}$ with unique quasi-preimage the quasi-preimage has type $\A$. 
	
	Since variables and lambda binders are the sole cause for HOL terms having multiple quasi-preimages, choosing indices for variables in HOL terms induces unique quasi-preimages (respecting those type indices). This uniqueness is a direct consequence of (the proof of) Lemma~\ref{lem:termwise-inj}.
	
	We will consider only those quasi-preimages that respect type indices, for the notions of unnormalizably spurious and almost proper terms.
	
	With respect to these choices, the transformation will do the following two things (in this order) in order to "normalize" unnormalizably spurious terms to almost proper ones:
	\begin{enumerate}
		\item apply beta and eta reductions and in case this doesn't yield almost proper terms
		\item replace unnormalizably spurious function applications of type $\typB$ by the "default terms" $\defaultTerm{\typB}$ of type $\PhiAppl{\typB}$, defined as the translation of a fixed DHOL term of type $\typB$ (whose existence follows from the last property of Lemma~\ref{lem:indexing}).
	\end{enumerate}
	
	As we are assuming a valid HOL derivation indexed according to Lemma~\ref{lem:indexing}, we will only define this transformation on well-typed HOL terms with preimage types consistent with the indexing lemma.
	We can then define the transformation of $\tm_{\A}$ (denoted by $\sRed{\tm_{\A}}$) by first matching line for $\tm_{\A}$ as follows:
	{\begin{align}
			&\sRed{\tm_{\A}}&:=&\tm_{\A}&&\text{if $\tm$ has quasi-preimage of type $\A$}\sredlabel{sredP}\\
			&\sRed{\termf_{\piType{x}{\A}\typB}\ \tm_{\A}}\negSp\!\!&:=&\sRed{\sRed{\termf_{\piType{x}{\A}\typB}}\ \sRed{\tm_{\A}}}\QQQQNegSp&&\nonumber\\& &&
			&&\text{if $\termf_{\piType{x}{\A}\typB}\ \tm_{\A}$ not beta or eta reducible}\sredlabel{sredApplP}
	\end{align}}
	In the following cases, we assume that the term $\tm_{\A}$ in $\sRed{\cdot}$ on the left of $:=$ isn't almost proper with a quasi-preimage of type $\A$: 
	{\small
		\begin{align}
			&\sRed{\tm_{\A}}&:=&\sRed{\betaEtaRed{\tm_{\A}}} \quad\negSp\text{if $\tm$ is beta or eta reducible}\sredlabel{sredBetaEtaRed}\\
			&\sRed{\s_{\A}\termEquals{\PhiAppl{\A}} \tm_{\Ap}}&:=&\sRed{\s_{\A}}\termEquals{\PhiAppl{\A}}\sRed{\tm_{\A}}\sredlabel{sredEq}\\
			&\sRed{\termF_{\bool}\impl \termC{G}_{\bool}}&:=&\sRed{\termF_{\bool}}\impl \sRed{\termC{G}_{\bool}}\QQNegSp\QQNegSp\QQNegSp\sredlabel{sredImpl}\\
			&\sRed{\lambdaFun{x}{\A}\s_{\typB}}&:=&\lambdaFun{x}{\A}\sRed{\s_{\typB}}\sredlabel{sredLam}\\
			%&\sRed{\choiceOp{x}{\PhiAppl{\A}}{\p}}_{\A}&:=&\choiceOp{x}{\PhiAppl{\A}}{\textcolor{black}{\sRed{\p}_{\piType{x}{\A}\bool}}}\sredlabel{sredChoice}\\
			&\sRed{\left(\termf_{\piType{x}{\A}\typB}\ \tm_{\Ap}\right)_{\Bp}}&:=&\defaultTerm{\typB} \qquad\qquad\negSp\text{if $\A\neq \Ap$ or $\typB\neq \Bp$\negSp} \sredlabel{sredAppl}\\
			&\sRed{\subst{\tm_{\A}}{\x_{\typB}}{\s_{\typB}}}&:=&\subst{\sRed{\tm_{\A}}}{\x}{\sRed{\s_{\typB}}} \sredlabel{sredSubst}\\
			&\sRed{\tm_{\A}}&:=&\tm\sredlabel{sredDefault}
	\end{align}}
Since $\sRed{\cdot}$ of a term is always almost proper by (\ref{sredP}), it follows that $\sRed{\sRed{\cdot}}$ is the same as $\sRed{\cdot}$.
\end{definition}

\begin{lemma}\label{lem:normalizingPrfTransform}
	Assume a well-typed \dhole theory $T$ and a conjecture $\ctx\dedT \termC{\phi}$ with $\ctx$ well-formed and $\termC{\phi}$ well-typed. Assume a valid HOL derivation $D$ of $\ctxT\dedPT \PhiAppl{\phi}$. 
	Choose type indices for the terms in $D$ according to the properties of the indexing lemma (Lemma~\ref{lem:indexing}). 
	Then, for any steps $S$ in $D$ we can construct a macro-step for the normalizing statement transformation replacing step $S$ s.t. after replacing all steps by their macro-steps:
	\begin{itemize}
		\item the resulting derivation is valid,
		\item all terms occurring in the derivation are almost proper (w.r.t. the quasi-preimages determined by the type indices).
	\end{itemize}
\end{lemma}
\begin{proof}[Proof of Lemma~\ref{lem:normalizingPrfTransform}]
	%\ednote{CR: The critical parts of this proof are the choice of indices and the cases for the rules \ruleRef{appl} and \ruleRef{lambda} (and maybe also \ruleRef{congAppl} and \ruleRef{congLam}).}
	We will show this by induction on the shape of $D$.
	
	Firstly, we observe that there are no dependent types in HOL and the context and axioms contain no spurious subterms. Hence, well-formedness (of theories, contexts, types) and type-equality judgements are unaffected by the transformation. 
	So there is nothing to prove for the well-formedness and type-equality rules (those steps can be replaced by a macro step containing exactly this single step).
	
	\paragraph{Regarding the type indices:}
	We observe that the properties of the type indices provided by Lemma~\ref{lem:indexing} ensure that the smallest unnormalizably spurious terms (i.e. without unnormalizably spurious proper subterms) are function applications in which function and argument are both almost proper. Furthermore, in such a case the function is not a $\lambda$-function. 
	%The recursive definition of the normalizing statement transformation thus ensures that all terms in its image are almost proper. 
	
	It remains to consider the typing and validity rules and to construct macro steps for the steps in the derivation using them for the normalizing statement transformation. % of the assumptions hold, we then the normalizing statement transformation of the conclusion of the macro step holds:
	
	Since terms indexed by a type $\A$ have type $\PhiAppl{\A}$ it is easy to see from Definition~\ref{defn:normalizingPrfTransform} that the normalizing statement transformation replaces terms of type $\PhiAppl{\A}$ by terms of type $\PhiAppl{\A}$.
	
	\paragraph*{\ruleRef{const}:}
	Since constants are proper terms, there is nothing to prove.
	
	\paragraph*{\ruleRef{var}:}
	Since context variables are proper terms, there is nothing to prove.
	
	\paragraph*{\ruleRef{eqType}:}
	\begin{align}
		\NDLinePTD{\sRed{\s}_{\A}\ofT\PhiAppl{\A}}{\byAss}\label{eqTypeSR1}\\
		\NDLinePTD{\sRed{\tm}_{\A}\ofT\PhiAppl{\A}}{\byAss}\label{eqTypeSR2}\\
		\NDLinePTD{\sRed{\s}_{\A}\termEquals{\PhiAppl{\A}}\sRed{\tm}_{\A}\ofT\bool}{\ruleRef{eqType},(\ref{eqTypeSR1}),(\ref{eqTypeSR2})}\label{eqTypeSR3}\\
		\NDLinePTD{\sRed{\s_{\A}\termEquals{\PhiAppl{\A}}\tm_{\A}}_{\bool}\ofT\bool}{\ref{sredEq},(\ref{eqTypeSR3})}\nonumber
	\end{align}
	
	\paragraph*{\ruleRef{lambda}:}
	\begin{align}
		\NDLinePT{\concatCtx{\contextcolor{\Delta}}{\x_{\A}\ofT\PhiAppl{\A}}}{\sRed{\tm_{\typB}}_{\typB}\ofT\PhiAppl{\typB}}{\byAss}\label{lambdaSR1}\\
		\NDLinePTD{\left(\lambdaFun{\x_{\A}}{\PhiAppl{\A}}\sRed{\tm_{\typB}}_{\typB}\right)\ofT\PhiAppl{\A}\to \PhiAppl{\typB}}{\ruleRef{lambda},(\ref{lambdaSR1})}\label{lambdaSR2}
		\intertext{If $\sRed{t}_{\typB}$ isn't an unnormalizably spurious function application $\sRed{\termf_{\piType{y}{A'}\typB}}\ \x_{\A}$ for which $\x$ doesn't appear in $\termf$:}
		\NDLinePTD{\sRed{\lambdaFun{\x_{\A}}{\PhiAppl{\A}}\sRed{\tm_{\typB}}_{\typB}}\ofT\PhiAppl{\A}\to \PhiAppl{\typB}}{\ref{sredLam},(\ref{lambdaSR2})}\nonumber
		\intertext{Else by (\ref{sredBetaEtaRed}) we have  $\sRed{\lambdaFun{\x_{\A}}{\PhiAppl{\A}}\sRed{\tm_{\typB}}_{\typB}}=\sRed{\termf_{\piType{y}{\A}\typB}}$. By the remark about the type of $\sRed{\cdot}$ it follows that $\sRed{\termf_{\piType{y}{\A}\typB}}$ has type $\PhiAppl{\piType{y}{\A}\typB}=\PhiAppl{\A}\to\PhiAppl{\typB}$. }
		\NDLinePTD{\sRed{\termf_{\piType{y}{A'}\typB}}\ofT\PhiAppl{\A}\to \PhiAppl{\typB}}{see above}\label{lambdaSRfinal}\\
		\NDLinePTD{\sRed{\lambdaFun{\x_{\A}}{\PhiAppl{\A}}\sRed{\tm_{\typB}}_{\typB}}\ofT\PhiAppl{\A}\to \PhiAppl{\typB}}{(\ref{sredBetaEtaRed}),(\ref{lambdaSRfinal})}\nonumber
	\end{align}

	\paragraph*{\ruleRef{appl}:}
	\begin{align}
		\NDLinePTD{\sRed{\termf_{\piType{x}{\A}\typB}}\ofT\PhiAppl{\A}\to \PhiAppl{\typB}}{\byAss}\label{applSR1}\\
		\NDLinePTD{\sRed{\tm_{\Ap}}\ofT\PhiAppl{\A}}{\byAss}\label{applSR2}
		\intertext{If $\A \typeEquals \Ap$:}		
		\NDLinePTD{\sRed{\termf_{\piType{x}{\A}\typB}\ \tm_{\Ap}}\ofT\PhiB}{\ref{sredP},\ruleRef{lambda},(\ref{applSR1}),(\ref{applSR2})}\nonumber
		\intertext{If $\A \not\typeEquals \Ap$ then Lemma~\ref{lem:indexing}(6) implies that $\termf$ is not a lambda function and thus $\sRed{\termf_{\piType{x}{\A}\typB}}\ \sRed{\tm_{\Ap}}$ is not beta reducible. Thus by (\ref{sredAppl}) we have $$\sRed{\termf_{\piType{x}{\A}\typB}\ \tm_{\Ap}}=\sRed{\sRed{\termf_{\piType{x}{\A}\typB}}_{\piType{x}{\A}\typB}\ \sRed{\sRed{\tm_{\Ap}}_{\Ap}}}=\defaultTerm{\PhiAppl{\typB}}.$$ By construction we have $\defaultTerm{\PhiAppl{\typB}}\ofT\PhiAppl{\typB}$:}
		\NDLinePTD{\defaultTerm{\PhiAppl{\typB}}\ofT\PhiAppl{\typB}}{By Construction}\label{applSRC2final}\\
		\NDLinePTD{\sRed{\termf_{\piType{x}{\A}\typB}\ \sRed{\tm_{\Ap}}}\ofT\PhiB}{(\ref{sredApplP}),(\ref{sredAppl}),(\ref{applSRC2final})}\nonumber
	\end{align}
	
	\paragraph*{\ruleRef{implType}:}
	\begin{align}
		\NDLinePTD{\sRed{\termF}_{\bool}\ofT\bool}{\byAss}\label{implTypeSR1}\\
		\NDLinePTD{\sRed{\termG}_{\bool}\ofT\bool}{\byAss}\label{implTypeSR2}\\
		\NDLinePTD{\sRed{\termF}_{\bool}\impl \sRed{\termG}_{\bool}\ofT\bool}{\ruleRef{implType},(\ref{implTypeSR1}),(\ref{implTypeSR2})}\label{implTypeSR3}\\
		\NDLinePTD{\sRed{\termF_{\bool}\impl \termF_{\bool}}\ofT\bool}{\ref{sredImpl},(\ref{implTypeSR3})}\nonumber
	\end{align}
	
	\paragraph*{\ruleRef{axiom}:}
	Since translations of axioms to HOL are always proper terms and the additionally generated axioms are almost proper, there is nothing to prove here.
	
	\paragraph*{\ruleRef{assume}:}
	If the axiom is a typing axiom generated by the translation, it follows that it is almost proper. Similarly, if it is an axiom for a base type. 
	Otherwise:
	\begin{align}
		\namedass{ass}{\sRed{\termF_{\bool}}}&\text{ in $\theorycolor{\Delta}$}&&\text{\byAss}\label{assumeSR1}\\
		\NDLinePTD{\sRed{\termF_{\bool}}}{\ruleRef{assume},(\ref{assumeSR1})}\nonumber
	\end{align}
	By assumption $\sRed{\termF}$ almost proper (with a quasi-preimage of type $\bool$), so the conclusion of the rule is almost proper and there is nothing ot prove here.
	
	\paragraph*{\ruleRef{congLam}:}
	\begin{align}
		\NDLinePTD{\A\typeEquals \Ap}{\byAss}\label{congLamSR1}\\
		\NDLinePT{\concatCtx{\theorycolor{\Delta}}{\x_{\A}\ofT\PhiAppl{\A}}}{\sRed{\tm_{\typB}\termEquals{\PhiAppl{\typB}}\termtp_{\typB}}_{\bool}}{\byAss}\label{congLamSR2}\\
		\NDLinePT{\concatCtx{\theorycolor{\Delta}}{\x_{\A}\ofT\PhiAppl{\A}}}{\sRed{t_{\typB}}_{\typB}\termEquals{\PhiAppl{\typB}}\sRed{\termC{t'_{\typB}}}_{\typB}}{\ref{sredEq},(\ref{congLamSR2})}\label{congLamSR2a}\\
		\NDLinePTD{\lambdaFun{\x_{\A}}{\PhiAppl{\A}}\sRed{t_{\typB}}_{\typB}\termEquals{\PhiAppl{\A}\to \PhiAppl{\typB}}\nonumber\\&\lambdaFun{\x_{\A}}{\PhiAppl{\A}}\sRed{\termtp_{\typB}}_{\typB}}{\ruleRef{congLam},(\ref{congLamSR1}),(\ref{congLamSR2a})}\label{congLamSR3}
	\end{align}
	By assumption $\sRed{t}_{\typB}\termEquals{\PhiAppl{\typB}}\sRed{\termtp}_{\typB}$ almost proper with quasi-preimage consistent with type indices and $\A\typeEquals \Ap$, thus also $\lambdaFun{\x_{\A}}{\PhiAppl{\A}}\sRed{t}_{\typB}\termEquals{\PhiAppl{\A}\to \PhiAppl{\typB}}\lambdaFun{\x_{\A}}{\PhiAppl{\A}}\sRed{\termtp}_{\typB}$ almost proper with quasi-preimage consistent with type indices. 
	\begin{align}
		\NDLinePTD{\sRed{\lambdaFun{\x_{\A}}{\PhiAppl{\A}}\tm_{\typB}\termEquals{\PhiAppl{\A}\to \PhiAppl{\typB}}\lambdaFun{\x_{\A}}{\PhiAppl{\A}}\termtp_{\typB}}\QQNegSp}{\qquad\ref{sredLam},\ref{sredEq},(\ref{congLamSR3})}\nonumber
	\end{align} 
	
	\paragraph*{\ruleRef{congAppl}:}
	\begin{align}
		\NDLinePTD{\sRed{\tm_{\A}\termEquals{\PhiAppl{\A}}\termtp_{\A}}}{\byAss}\label{congApplSR1}\\
		\NDLinePTD{\sRed{\tm_{\A}}_{\A}\termEquals{\PhiAppl{\A}}\sRed{\termtp_{\A}}_{\A}}{\ref{sredEq},(\ref{congApplSR1})}\label{congApplSR1a}\\
		\NDLinePTD{\sRed{\termf_{\piType{x}{\Ap}\typB}\termEquals{\PhiAppl{\A}\to \PhiAppl{\typB}}\termfp_{\piType{x}{\Ap}\typB}}}{\byAss}\label{congApplSR2}\\
		\NDLinePTD{\sRed{\termf}_{\piType{x}{\Ap}\typB}\termEquals{\PhiAppl{\A}\to \PhiAppl{\typB}}\sRed{\termfp}_{\piType{x}{\Ap}\typB}}{\ref{sredEq},(\ref{congApplSR2})}\label{congApplSR2a}
		\end{align}
			Assume $\A\not\typeEquals \Ap$. By property 6. in Lemma~\ref{lem:indexing} $\sRed{\termf}$ and $\sRed{\termfp}$ are not $\lambda$-functions. 
			Consequently, the applications $\sRed{\termf}\ \sRed{\s}$ and $\sRed{\termfp}\ \sRed{\sp}$ are not beta or eta reducible. Thus, \[\sRed{\termf_{\piType{x}{\Ap}\typB}\ \tm_{\A}}=\defaultTerm{\PhiAppl{\typB}}\] and \[\mathsf{sRed}\big(\termfp_{\piType{x}{\Ap}\typB} \tm_{\A}\big)=\defaultTerm{\PhiAppl{\typB}}\] and we yield:
	\begin{align}
		\NDLinePTD{\sRed{\termf_{\piType{x}{\Ap}\typB}\ \tm_{\A}}\termEquals{\PhiB}\sRed{\termfp_{\piType{x}{\Ap}\typB}\ \termtp_{\A}}}{\ruleRef{refl}}\nonumber%\label{congApplSR3a}
	\end{align}
	Otherwise (if $\A\typeEquals \Ap$) the terms $\sRed{\termf}_{\piType{x}{\A}\typB}\ \sRed{\tm_{\A}}$ and $\sRed{\termfp}_{\piType{x}{\A}\typB}\ \sRed{\termtp_{\A}}$ are almost proper with quasi-preimages consistent with type indices. It follows: \[\sRed{\termf_{\piType{x}{\A}\typB}\ \sRed{\tm_{\A}}}=\sRed{\termf_{\piType{x}{\A}\typB}}\ \sRed{\tm_{\A}}\] and \[\sRed{\termfp_{\piType{x}{\A}\typB}\ \termtp_{\A}}=\sRed{\termfp_{\piType{x}{\A}\typB}}\ \sRed{\termtp_{\A}}\] and thus:
	\begin{align}
		\NDLinePTD{\sRed{\termf_{\piType{x}{\A}\typB}\ \tm_{\A}}\termEquals{\PhiAppl{\typB}}\sRed{\termfp}_{\piType{x}{\A}\typB}\ \termtp_{\A}}{\ruleRef{congAppl},(\ref{congApplSR1a}),(\ref{congApplSR2a})}\nonumber%\label{congApplSR3b}
	\end{align}
	In either case, we concluded \[\sRed{\termf_{\piType{x}{\A}\typB}\ \tm_{\A}}\termEquals{\PhiAppl{\typB}}\sRed{\termfp_{\piType{x}{\A}\typB}\ \termtp_{\A}}.\]
	By \ref{sredEq} this is already the desired conclusion of:
	$$\sRed{\termf_{\piType{x}{\A}\typB}\ \tm_{\A}\termEquals{\PhiAppl{\typB}}\termfp_{\piType{x}{\A}\typB}\ \termtp_{\A}}
	.$$
	
	\paragraph*{\ruleRef{refl}:}
	\begin{align}
		\NDLinePTD{\sRed{\tm_{\A}}_{\A}\ofT\PhiAppl{\A}}{\byAss}\label{reflSR1}\\
		\NDLinePTD{\sRed{\tm_{\A}}_{\A}\termEquals{\PhiAppl{\A}}\sRed{\tm_{\A}}_{\A}}{\ruleRef{refl},(\ref{reflSR1})}\label{reflSR2}\\
		\NDLinePTD{\sRed{\sRed{\tm_{\A}}_{\A}\termEquals{\PhiAppl{\A}}\sRed{t_{\A}}_{\A}}}{\ref{sredEq},(\ref{reflSR2})}\nonumber
	\end{align}
	
	\paragraph*{\ruleRef{sym}:}
	\begin{align}
		\NDLinePTD{\sRed{\tm_{\A}\termEquals{\PhiAppl{\A}}\s_{\A}}}{\byAss}\label{symSR1}\\
		\NDLinePTD{\sRed{\tm_{\A}}_{\A}\termEquals{\PhiAppl{\A}}\sRed{\s_{\A}}_{\A}}{\ref{sredEq},(\ref{symSR1})}\label{symSR2}\\
		\NDLinePTD{\sRed{\s_{\A}}_{\A}\termEquals{\PhiAppl{\A}}\sRed{\tm_{\A}}_{\A}}{\ruleRef{sym},(\ref{symSR2})}\label{symSR3}\\
		\NDLinePTD{\sRed{\s_{\A}\termEquals{\PhiAppl{\A}}\tm_{\A}}}{\ref{sredEq},(\ref{symSR3})}\nonumber
	\end{align}
	
	\paragraph*{\ruleRef{beta}:}
	\begin{align}
		\NDLinePTD{\sRed{\left(\lambdaFun{\x_{\A}}{\PhiAppl{\A}}\s_{\typB}\right)\ \tm_{\Ap}}_{\Bp}\ofT\PhiAppl{\typB}}{\byAss}\label{betaSR1}
		\intertext{
			By Lemma~\ref{lem:indexing} (6), it follows that $\A=\Ap$. 
			If $\sRed{\left(\lambdaFun{\x_{\A}}{\PhiAppl{\A}}\s_{\typB}\right)\ \tm_{\A}}_{\Bp}$ is almost proper with quasi-preimage of type $\typB\typeEquals \Bp$, then $\sRed{\left(\lambdaFun{\x_{\A}}{\PhiAppl{\A}}\s_{\typB}\right)\ \tm_{\A}}_{\Bp}=\left(\lambdaFun{\x_{\A}}{\PhiAppl{\A}}\s_{\typB}\right)\ \tm_{\A}$ and thus:}
		\NDLinePTD{\left(\lambdaFun{\x_{\A}}{\PhiAppl{\A}}\s_{\typB}\right)\ \tm_{\A}\ofT\PhiAppl{\typB}}{\ref{sredP},(\ref{betaSR1})}\label{betaSR2}\\
		\NDLinePTD{\left(\lambdaFun{\x_{\A}}{\PhiAppl{\A}}\s_{\typB}\right)\ \tm_{\A}\termEquals{\PhiAppl{\typB}}\subst{\s_{\typB}}{\x_{\A}}{\tm_{\A}}}{\ruleRef{beta},(\ref{betaSR2})}\label{betaSR3}\\
		\NDLinePTD{\sRed{\left(\lambdaFun{\x_{\A}}{\PhiAppl{\A}}\s_{\typB}\right)\ \tm_{\A}\termEquals{\PhiAppl{\typB}} \subst{\s_{\typB}}{\x_{\A}}{\tm_{\A}}}}{\ref{sredEq},\ref{sredP},(\ref{betaSR3})}\nonumber
	\end{align}
	Otherwise by \ref{sredApplP} and \ref{sredSubst}, $$\sRed{\left(\lambdaFun{\x_{\A}}{\PhiAppl{\A}}\s_{\typB}\right)\ \tm_{\Ap}}=\sRed{\subst{\s_{\typB}}{\x_{\A}}{\tm_{\Ap}}}$$ and we yield:
	\begin{align}
		\NDLinePTD{\sRed{\left(\lambdaFun{\x_{\A}}{\PhiAppl{\A}}\s_{\typB}\right)\ \tm_{\Ap}}\termEquals{\PhiAppl{\typB}}\sRed{\subst{\s_{\typB}}{\x_{\A}}{\tm_{\Ap}}}}{\ruleRef{refl},above observation}\label{betaSR5}\\
		\NDLinePTD{\sRed{\left(\lambdaFun{\x_{\A}}{\PhiAppl{\A}}\s_{\typB}\right)\ \tm_{\Ap}\termEquals{\PhiAppl{\typB}}\subst{\s_{\typB}}{\x_{\A}}{\tm_{\Ap}}}}{\ref{sredEq},(\ref{betaSR5})}\nonumber
	\end{align}
	
	\paragraph*{\ruleRef{eta}:}
	\begin{align}
		\NDLinePTD{\sRed{\tm_{\piType{x}{\A}\typB}}\ofT\PhiAppl{\A}\to \PhiAppl{\typB}}{\byAss}\label{etaSR1}\\
		&\text{$\x$ not in $\theorycolor{\Delta}$}&&\text{\byAss}\label{etaSR2}
		\intertext{Since $\sRed{\tm_{\piType{x}{\A}\typB}}$ is almost proper with quasi-preimage of type $\piType{x}{\A}\typB$, it follows that  $\lambdaFun{\x_{\A}}{\PhiAppl{\A}}\sRed{\tm_{\piType{x}{\A}\typB}}\ \x_{\A}$ is also almost proper with quasi-preimage of type $\piType{x}{\A}\typB$. It follows:}
		\NDLinePTD{\sRed{\tm_{\piType{x}{\A}\typB}}\termEquals{\PhiAppl{\A}\to \PhiAppl{\typB}}\lambdaFun{\x_{\A}}{\PhiAppl{\A}}\sRed{\tm_{\piType{x}{\A}\typB}}\ \x_{\A}}{\ruleRef{eta},(\ref{etaSR1}),(\ref{etaSR2})}\label{etaSR3}\\
		\NDLinePTD{\sRed{\tm_{\piType{x}{\A}\typB}\termEquals{\PhiAppl{\A}\to \PhiAppl{\typB}}\lambdaFun{\x_{\A}}{\PhiAppl{\A}}\tm_{\piType{x}{\A}\typB}\ \x_{\A}}}{\ref{sredApplP},\ref{sredLam},\ref{sredEq},(\ref{etaSR3})}\nonumber
	\end{align}

	\paragraph*{\ruleRef{congDed}:}
	\begin{align}
		\NDLinePTD{\sRed{\termF_{\bool}\termEqB \termC{F'}_{\bool}}}{\byAss}\label{congDedSR1}\\
		\NDLinePTD{\sRed{\termC{F'}_{\bool}}}{\byAss}\label{congDedSR3}\\
		\NDLinePTD{\sRed{\termF_{\bool}}\termEqB\sRed{\termC{F'}_{\bool}}}{\ref{sredEq},(\ref{congDedSR1})}\label{congDedSR2}\\
		\NDLinePTD{\sRed{\termF_{\bool}}}{\ruleRef{congDed},(\ref{congDedSR2}),(\ref{congDedSR3})}\nonumber
	\end{align}
	
	\paragraph*{\ruleRef{implI}:}
	\begin{align}
		\NDLinePTD{\sRed{\termF_{\bool}}\ofT\bool}{\byAss}\label{implISR1}\\
		\NDLinePT{\concatCtx{\theorycolor{\Delta}}{\namedass{ass_F}{\sRed{\termF_{\bool}}}}}{\sRed{\termC{G}_{\bool}}}{\byAss}\label{implISR2}\\
		\NDLinePTD{\sRed{\termF_{\bool}}\impl \sRed{\termC{G}_{\bool}}}{\ruleRef{implI},(\ref{implISR1}),(\ref{implISR2})}\label{implISR3}\\
		\NDLinePTD{\sRed{\termF_{\bool}\impl \termC{G}_{\bool}}}{\ref{sredImpl},(\ref{implISR3})}\nonumber
	\end{align}
	
	\paragraph*{\ruleRef{implE}:}
	\begin{align}
		\NDLinePTD{\sRed{\termF_{\bool}\impl \termC{G}_{\bool}}}{\byAss}\label{implESR1}\\
		\NDLinePTD{\sRed{\termF_{\bool}}}{\byAss}\label{implESR3}\\
		\NDLinePTD{\sRed{\termF_{\bool}}\impl\sRed{\termC{G}_{\bool}}}{\ref{sredImpl},(\ref{implESR1})}\label{implESR2}\\
		\NDLinePTD{\sRed{\termC{G}_{\bool}}}{\ruleRef{implE},(\ref{implESR2}),(\ref{implESR3})}\nonumber
	\end{align}
	
	\paragraph*{\ruleRef{boolExt}:}
	\begin{align}
		\NDLinePTD{\sRed{\p_{\bool\to\bool}\ \T_{\bool}}}{\byAss}\label{boolExtSR1}\\
		\NDLinePTD{\sRed{\p_{\bool\to\bool}\ \termF_{\bool}}}{\byAss}\label{boolExtSR2}\\
		\NDLinePTD{\sRed{\p_{\bool\to\bool}}\ \T}{(\ref{sredApplP}),(\ref{sredP}),(\ref{boolExtSR1})}\label{boolExtSR3}\\
		\NDLinePTD{\sRed{\p_{\bool\to\bool}}\ \F}{(\ref{sredApplP}),(\ref{sredP}),(\ref{boolExtSR2})}\label{boolExtSR4}\\
		\NDLinePTD{\univQuant{x}{\bool}\sRed{\p_{\bool\to\bool}}\ \x}{\ruleRef{boolExt},(\ref{boolExtSR3}),(\ref{boolExtSR4})}\label{boolExtSR5}\\
		\NDLinePTD{\sRed{\univQuant{x}{\bool}\p_{\bool\to\bool}\ \x}}{\ref{sredEq},\ref{sredLam},\ref{sredApplP},\ref{sredP},(\ref{boolExtSR5})}\nonumber
	\end{align}
	
	\paragraph*{\ruleRef{nonempty}:}
	\begin{align}
		\NDLinePTD{\sRed{\termF_{\bool}}\ofT\bool}{\byAss}\label{nonemptySR1}\\
		\NDLinePT{\concatCtx{\theorycolor{\Delta}}{\x_{\A}\ofT\PhiA}}{\sRed{\termF_{\bool}}}{\byAss}\label{nonemptySR2}\\
		\NDLinePTD{\sRed{\termF_{\bool}}}{\ruleRef{nonempty},(\ref{nonemptySR1}),(\ref{nonemptySR2})}\nonumber		
	\end{align}
\end{proof}

\subsection{Lifting admissible HOL derivations of validity statements to \dhole}\label{sec:sound}
We finally have all required results to prove the soundness of the translation from \dhole to HOL.
\begin{proof}[Proof of Theorem~\ref{thm:soundPaper}]
	As shown in Lemma~\ref{lem:normalizingPrfTransform}, we may assume that the proof of $\ctxT\dedPT \PhiAppl{\termF}$ is admissible, so it only contains almost-proper terms. 
	Consequently, whenever an equality $\s\termEquals{\PhiAppl{\A}}\tm$ is derivable in HOL and $\sp, \termtp$ are the quasi-preimages of $\s, \tm$ respectively, it follows that it's quasi-preimage $\sp\termEquals{\A}\tm$ is well-typed in \dhole and thus $\sp\ofT\A$ and $\termtp\ofT\A$. 
	Without loss of generality (adding extra assumptions throughout the proof) we may assume that the context of the (final) conclusion is the translation of a \dhole context.
	By Lemma~\ref{lem:termwise-inj} the translation is term-wise injective. 
	
	Therefore, the translated conjecture is a proper validity statement with unique (quasi)-preimage in \dhole. 
	If we can lift a derivation of the translated conjecture to a valid \dhole derivation of its quasi-preimage, the resulting derivation is a valid derivation of the original conjecture. 
	This means, that it suffices to prove that we can lift admissible derivations of a proper validity statement $S$ in HOL to a derivation of a quasi-preimage of $S$.
	
	We prove this claim by induction on the validity rules of HOL as follows:
	
	Given a validity rule $R$ with assumptions $A_1,\ldots,A_n$, validity assumptions (assumptions that are validity statements) $V_1,\ldots,V_m$, non-judgement assumptions (meaning assumptions that something occurs in a context or theory) $N_1,\ldots,N_p$ and conclusion $C$ we will show the following:
	\begin{claim}\label{soundness-inductive-claim}
		Assuming that the $A_i$ and the $N_j$ hold.
		\begin{enumerate}
			\item Assume that the conclusion $C$ is proper with quasi-preimage $C^{-1}$. Then the contexts $C_i$ of the $V_i$ are proper %with unique quasi-preimages\checknote{CR: I don't think we really need this and I'm also not sure we show it holds in all of the following cases. } 
			and the quasi-preimages of the $V_i$ are well-formed. %and uniquely determined from the quasi-preimage $C^{-1}$ of $C$.
			\item Assume that whenever an $V_i$ is proper its quasi-preimage (where we choose the same preimages for identical terms and types with several possible preimages) holds in \dhole and that the conclusion $C$ is proper with quasi-preimage $C^{-1}$. 
			Then, $C^{-1}$ holds in \dhole.
		\end{enumerate}
	\end{claim}
	
	Consider the first part of this claim, namely that if $C$ is proper then the $V_i$ are proper. %and the quasi-preimages of their contexts are uniquely determined from $C^{-1}$.% and have unique quasi-preimages determined from $C^{-1}$.
	Since all formulae appearing in the derivation are almost proper, this implies that the $V_i$ themselves are proper and by construction (choice of quasi-preimage) the contexts of their quasi-preimages fit together with the context of $C^{-1}$.
	
	The translation clearly implies that if an $N_j$ holds in HOL, the corresponding non-judgement assumption ${N_j}^{-1}$ holds in \dhole (e.g. if $\PhiAppl{\termF}$ is an axiom in $\PhiAppl{\theorycolor{T}}$, then $\termF$ must be an axiom in $\theorycolor{T}$). 
	
	Since the validity judgement being derived is proper, it follows from this first part of the claim that the validity assumptions of all validity rules in the derivation are proper. % and the contexts of the quasi-preimages of those assumptions are uniquely determined. 
	
	By induction on the validity rules, if given an arbitrary validity rule $R$ whose assumptions hold and whose validity assumptions all satisfy a property $P$ we can show that $P$ holds on the conclusion of $R$, then all derivable validity judgments have property $P$.
	Since all the validity assumptions and conclusions of validity rules in the derivation are proper, the property of having a derivable quasi-preimage is such a property.
	By this induction principle, it suffices to prove the claim for the validity rules in \hol. 
	
	We will therefore consider the validity rules one by one. For each rule we first prove the first part of the claim. 
	Sometimes we also need that the quasi-preimages of some non-validity (typically typing) assumptions hold, so we will prove that this also follows from the conclusion being proper.
	Then the assumption of the second part, combined with the first part implies that the quasi-preimages of the $V_i$ hold in \dhole and it is easy to prove that also $C^{-1}$ holds in \dhole.
	
	Throughout this proof we will use the notation $\quasiImage{\tm}$ to denote that $\tm$ is some quasi-preimage of $\quasiImage{\tm}$. 
	Since the translation is surjective on type-level we will only need this notation on term-level.
	
	Validity can be shown using the rules \ruleRef{congLam}, \ruleRef{eta}, \ruleRef{congAppl}, \ruleRef{congDed}, \ruleRef{beta}, \ruleRef{refl}, \ruleRef{sym}, \ruleRef{assume}, \ruleRef{axiom}, \ruleRef{implI}, \ruleRef{implE}, \ruleRef{boolExt} and \ruleRef{nonempty}.
	\subparagraph{\ruleRef{congLam}:}

	Since the conclusion is proper, it follows that the preimage%\checknote{This preimage is not quite unique. We can choose several potential preimages for $\PhiAppl{\A}$ in the $\lambda$-functions (although these choices probably all have to be supertypes of $\A$). Think about whether this can cause any problems.} 
	\[
	\ctx\dedT \lambdaFun{x}{\A}\tm\termEquals{\piType{x}{\A}\typB}\lambdaFun{x}{\A}\termtp
	\] of the normalization \[\ctxT\dedPT 
	\univQuant{x}{\PhiAppl{\A}}\univQuant{y}{\PhiAppl{\A}}\termEqT{\A}{\x}{\y}\impl \termEqT{B}{\quasiImage{\tm}\ \x}{\quasiImage{\termtp}\ \y}\] of the conclusion is well-formed. 
	By rule \ruleRef{eqTyping} and rule \ruleRef{sym} we obtain $\ctx\dedT\lambdaFun{x}{\A}\tm\ofT\piType{x}{\A}\typB$ and $\ctx\dedT\lambdaFun{x}{\A}\termtp\ofT\piType{x}{\A}\typB$ in \dhole. 
	\begin{align}
		\NDLineTG{\lambdaFun{x}{\A}\tm\ofT\piType{x}{\A}\typB}{see above}\label{congApplP1.1}\\
		\NDLineTG{\lambdaFun{x}{\A}\termtp\ofT\piType{x}{\A}\typB}{see above}\label{congApplP1.2}\\
		\NDLineT{\concatCtx{\ctx}{\y \ofT\A}}{\left(\lambdaFun{x}{\A}\tm\right)\ \y \ofT\typB}{\ruleRef{appl},\ruleRef{varDed},(\ref{congApplP1.1}),\ruleRef{assume}}\label{congApplP1.3}\\
		\NDLineT{\concatCtx{\ctx}{\y \ofT\A}}{\left(\lambdaFun{x}{\A}\termtp\right)\ \y \ofT\typB}{\ruleRef{appl},\ruleRef{varDed},(\ref{congApplP1.2}),\ruleRef{assume}}\label{congApplP1.4}\\
		\NDLineT{\concatCtx{\ctx}{\y \ofT\A}}{\left(\lambdaFun{x}{\A}\tm\right)\ \y\termEquals{\typB}\subst{\tm}{x}{\y}}{\ruleRef{beta},(\ref{congApplP1.3})}\label{congApplP1.5}\\
		\NDLineT{\concatCtx{\ctx}{\y \ofT\A}}{\left(\lambdaFun{x}{\A}\termtp\right)\ \y\termEquals{\typB}\subst{\termtp}{\x}{\y}}{\ruleRef{beta},(\ref{congApplP1.4})}\label{congApplP1.6}\\		
		\NDLineT{\concatCtx{\ctx}{\x\ofT\A}}{\tm\ofT\typB}{$\alpha$-renaming,\ruleRef{congColon},(\ref{congApplP1.5})}\label{congApplP1.7}\\
		\NDLineT{\concatCtx{\ctx}{\x\ofT\A}}{\termtp\ofT\typB}{$\alpha$-renaming,\ruleRef{congColon},(\ref{congApplP1.6})}\label{congApplP1.8}\\
		\NDLineT{\concatCtx{\ctx}{\x\ofT\A}}{\tm\termEquals{\typB}\tm}{\ruleRef{eqType},(\ref{congApplP1.7}),(\ref{congApplP1.8})}\nonumber
	\end{align}
	Clearly, $\concatCtx{\ctx}{\x\ofT\A}\dedT \tm\termEquals{\typB}\termtp$ is a quasi-preimage of the validity assumption, so this proves the first part of the claim. 
	
	Regarding the second part:
	\begin{align}
		\NDLineT{\concatCtx{\ctx}{\x\ofT\A}}{\tm\termEquals{B}\termtp}{\byAss}\label{congLamC3}\\
		\NDLineTG{\A\typeEquals \A}{\ruleRef{tpEqRefl},\ruleRef{typingTp},(\ref{congLamC3})}\label{congLamC4}\\
		\NDLineTG{\lambdaFun{x}{\A}\tm\termEquals{\piType{x}{\A}\typB}\lambdaFun{x}{\A}\termtp}{\ruleRef{congLam'},(\ref{congLamC4})}\label{congLamC5}
	\end{align}
	
	\subparagraph{\ruleRef{eta}:}
	Since the rule has no validity assumption, the first part of the claim holds. 
	
	For the second part, we still need the quasi-preimage of the assumption to hold, so we will show that it follows from the conclusion being proper.
	
	Since the conclusion is proper, it follows that the preimage 
	\[
	\ctx\dedT \tm\termEquals{\piType{x}{\A}\typB}\lambdaFun{x}{\A}\tm\ \x
	\] of the normalization \[\ctxT\dedPT 
	\univQuant{x}{\PhiAppl{\A}}\univQuant{y}{\PhiAppl{\A}}\termEqT{\A}{\x}{\y}\impl \termEqT{\typB}{\quasiImage{\tm}\ \x}{\left(\lambdaFun{x}{\PhiAppl{\A}}\quasiImage{\tm}\ \x\right)\ \y}\] of the conclusion is well-formed. 
	By rule \ruleRef{eqTyping} and rule \ruleRef{sym} we obtain $\ctx\dedT \tm\ofT\piType{x}{\A}\typB$ and $\ctx\dedT\lambdaFun{x}{\A}\tm\ \x\ofT\piType{x}{\A}\typB$ in \dhole.
	Clearly, $\ctx\dedT \tm\ofT\piType{x}{\A}\typB$ is a quasi-preimage of the validity assumption, so this proves the quasi-preimage of the assumption of the rule. 
	
	Regarding the second part:
	\begin{align}
		\NDLineTG{\tm\ofT\piType{x}{\A}\typB}{see above}\label{etaC1}\\
		\NDLineTG{\tm\termEquals{\piType{x}{\A}\typB}\lambdaFun{x}{\A}\tm\ \x}{\ruleRef{etaPi},(\ref{etaC1})}\nonumber
	\end{align}
	\subparagraph{\ruleRef{congAppl}:}
	Since the conclusion is proper, it follows that the preimage 
	\[
	\ctx\dedT \termf\ \tm\termEquals{B}\termfp\ \termtp
	\] of the normalization \[\termEqT{B}{\quasiImage{\termf}\ \quasiImage{\tm}}{\quasiImage{\termfp}\ \quasiImage{\termtp}}\] of the conclusion is well-formed. 
	By rule \ruleRef{eqTyping} and rule \ruleRef{sym} we obtain $\ctx\dedT \termf\ \tm\ofT\typB$ and $\ctx\dedT \termfp\ \termtp\ofT\typB$ in \dhole.
	Obviously, $\ctx\dedT \tm\termEquals{\A}\termtp$ and $\ctx\dedT \termf\termEquals{\piType{x}{\A}\typB}\termfp$ are quasi-preimages of the validity assumptions. 
	
	Since the validity assumptions use the same context as the conclusion, it follows that they are both proper with uniquely determined context. As observed in the beginning of the proof if a proper assumption of a rule is an equality over a type $\PhiAppl{\A}$, the induction hypothesis implies that the quasi-preimage of that assumption in which the equality is over type $\A$ must be well-formed.
	Hence both $\ctx\dedT \tm\termEquals{\A}\termtp$ and $\ctx\dedT \termf\termEquals{\piType{x}{\A}\typB}\termfp$ are well-formed in \dhole, so we have proven the first part of the claim.
	
	Regarding the second part of the claim:
	\begin{align}
		\NDLineTG{\tm\termEquals{\A}\termtp}{\byAss}
		\label{congApplC7}\\
		\NDLineTG{\termf\termEquals{\piType{x}{\A}\typB}\termfp}{\byAss}
		\label{congApplC8}\\
		\NDLineTG{\termf\ \tm\termEquals{\typB}\termfp\ \termtp}{\ruleRef{congAppl},(\ref{congApplC7}),(\ref{congApplC8})}
	\end{align}
	This is what we had to show.
	
	\subparagraph{\ruleRef{congDed}:}
	Since the conclusion is proper, it follows that the preimage 
	\[
	\ctx\dedT \termF
	\] of the normalization \[\ctxT\dedPT \quasiImage{\termF}\] of the conclusion is well-formed. 
	Thus we have $\ctx\dedT \termF\ofT\bool$. 
	Since the validity assumptions use the same context as the conclusion, it follows that they are both proper with uniquely determined. As observed in the beginning of the proof if a proper assumption of a rule is an equality over a type $\PhiAppl{\A}$ (here $\A=\PhiAppl{\A}=\bool$), the induction hypothesis implies that the quasi-preimage of that assumption in which the equality is over type $\bool$ must be well-formed.
	Clearly, $\ctx\dedT \termFp\termEqB \termF$ and $\ctx\dedT \termF$ are the quasi-preimages of the two validity assumptions. 
	Since the former is a validity statement about the quasi-preimage of an equality, it follows that $\ctx\dedT \termFp\termEqB \termF$ is well-formed.
	We have already seen that $\ctx\dedT \termF$ is well-typed.
	This shows the first part of the claim.

	Regarding the second part:
	\begin{align}
		\NDLineTG{\termFp\termEqB \termF}{\byAss}\label{congDedC1}\\
		\NDLineTG{\termFp}{\byAss}\label{congDedC2}\\
		\NDLineTG{\termF}{\ruleRef{congDed},(\ref{congDedC1}),(\ref{congDedC2})}\nonumber
	\end{align}
	\subparagraph{\ruleRef{beta}:}
	Since the rule has no validity assumptions, the first part of the claim trivially holds.
	
	Since the conclusion is proper, it follows that the preimage 
	\[
	\ctx\dedT \big(\lambdaFun{x}{\A}\s\big)\ \tm\termEquals{\piType{x}{\A}\typB}\subst{\s}{\x}{\tm}
	\] of the normalization \[\ctxT\dedPT 
	\termEqT{B}{\big(\lambdaFun{x}{\PhiAppl{\A}}\quasiImage{\s}\big)\ \quasiImage{\tm}}{\subst{\quasiImage{\s}}{x}{\quasiImage{\tm}}}\] of the conclusion is well-formed. 
	By rule \ruleRef{eqTyping}, we obtain $\ctx\dedT \big(\lambdaFun{x}{\A}\s\big)\ \tm\ofT\typB$ in \dhole.
	Clearly, $\ctx\dedT \big(\lambdaFun{x}{\A}\s\big)\ \tm\ofT\typB$ is a quasi-preimage of the assumption of the rule, so we have proven that the quasi-preimage of the assumption of the rule holds in \dhole.
	
	Regarding the second part:
	\begin{align}
		\NDLineTG{\big(\lambdaFun{x}{\A}\s\big)\ \tm\ofT\typB}{see above}\label{betaC1}\\
		\NDLineTG{\big(\lambdaFun{x}{\A}\s\big)\ \tm\termEquals{\piType{x}{\A}\typB}\subst{\s}{x}{\tm}}{\ruleRef{beta},(\ref{betaC1})}\nonumber
	\end{align}
	
	\subparagraph{\ruleRef{refl}:}
	Once again the rule has no validity assumptions, so the first part of the claim trivially holds.
	
	Since the conclusion is proper, it follows that the preimage 
	\[
	\ctx\dedT \tm\termEquals{\A} \termtp
	\] of the normalization \[\ctxT\dedPT \termEqT{\A}{\quasiImage{\tm}}{\quasiImage{\tm}}\] of the conclusion is well-formed. 
	By Lemma~\ref{lem:termwise-inj} it follows that $\tm$ and $\termtp$ are identical so the quasi-preimage is $\ctx\dedT \termEquals{\A} \tm$. 
	By rule \ruleRef{eqTyping}, we obtain $\ctx\dedT \tm\ofT\A$ in \dhole, the quasi-preimage of the assumption of the rule.
	
	Regarding the second part of the claim:
	\begin{align}
		\NDLineTG{\tm\ofT\A}{see above}\label{reflC1}\\
		\NDLineTG{\tm\termEquals{\A}\tm}{\ruleRef{refl},(\ref{reflC1})}\nonumber
	\end{align}
	
	\subparagraph{\ruleRef{sym}:}
	Since the conclusion is proper, it follows that the preimage 
	\[
	\ctx\dedT \tm\termEquals{\A}\s
	\] of the normalization \[
	\ctxT\dedPT \termEqT{\A}{\quasiImage{\tm}}{\quasiImage{\s}}\] of the conclusion is well-formed. 
	By the rules \ruleRef{eqTyping} and \ruleRef{sym} both $\ctx\dedT \tm\ofT\A$ and $\ctx\dedT \s\ofT\A$ follow. 
	By rule \ruleRef{eqType} it follows that $\ctx\dedT \s\termEquals{\A}\tm$ is well-formed. 
	Clearly, $\ctx\dedT \s\termEquals{\A}\tm$ is the quasi-preimages of the validity assumption, so we have proven the first part of the claim.
	
	Regarding the second part:
	\begin{align}
		\NDLineTG{\tm\termEquals{\A}\s}{\byAss}\label{symC1}\\
		\NDLineTG{\s\termEquals{\A}\tm}{\ruleRef{sym},(\ref{symC1})}
	\end{align}
	
	\subparagraph{\ruleRef{assume}:}
	Once again, there are no validity assumption, so the first part of the claim is trivial.
	
	Since the conclusion is proper, it follows that the preimage 
	\[
	\ctx\dedT \termF
	\] of the normalization \[
	\ctxT\dedPT \quasiImage{\termF}\] of the conclusion is well-formed and thus $\ctx\dedT \termF\ofT\bool$.
	\begin{align}
		\NDLineTG{\termF\ofT\bool}{see above}\label{assumeP1.1}\\
		\NDLineTG{\Type{\bool}}{\ruleRef{typingTp},(\ref{assumeP1.1})}\label{assumeP1.2}\\
		\NDLineT{}{\Ctx{\ctx}}{\ruleRef{tpCtx},(\ref{assumeP1.2})}\label{assumeP1.3}		
	\end{align}
	
	The context assumption may be the translation of a context assumption in \dhole or a typing assumption added by the translation. 
	In the latter case, $\termF$ is of the form $\termF=\PredPhi{\A}{\x}$ for $\ctxIn{\x\ofT\A}{\ctx}$. 
	In that case, the second part of the claim $\ctx\dedT \termF$ can be concluded as follows:
	\begin{align}			
		\NDLineTG{\x\ofT\A}{\ruleRef{var''}\QQQuad}\label{assumeCT1}
	\end{align}
	\vspace{-2.5\baselineskip\baselinestretch}
	\begin{align}
		\NDLineTG{\x\termEquals{\A}\tm}{\ruleRef{refl},(\ref{assumeCT1})}\label{assumeCT2}\\
		\NDLineTG{\termF}{\byAss{} $\termF=\PredPhi{\A}{x}$,(\ref{assumeCT2})}\nonumber
	\end{align}
	
	Otherwise:
	\begin{align}
		&\ctxIn{\namedass{ass}{\termF}}{\ctx}&&\text{\byAss}\label{assumeC1}\\
		\NDLineTG{\termF}{\ruleRef{assume},(\ref{assumeC1}),(\ref{assumeP1.3})}\nonumber
	\end{align}

	\subparagraph{\ruleRef{axiom}:}
	Once again, there are no validity assumption, so the first part of the claim is trivial.
	
	Since the conclusion is proper, it follows that the preimage 
	\[
	\ctx\dedT \termF
	\] of the normalization \[
	\ctxT\dedPT \quasiImage{\termF}\] of the conclusion is well-formed and thus $\ctx\dedT \termF\ofT\bool$.
	\begin{align}
		\NDLineTG{\termF\ofT\bool}{see above}\label{axiomP1.1}\\
		\NDLineTG{\Type{\bool}}{\ruleRef{typingTp},(\ref{axiomP1.1})}\label{axiomP1.2}\\
		\NDLineT{}{\Ctx{\ctx}}{\ruleRef{tpCtx},(\ref{axiomP1.2})}\label{axiomP1.3}		
	\end{align}
	The axiom may be the translation of an axiom in $\thy$, a typing axiom added by the translation or an axiom added for some base type $\A$.
	In the first case, the second part of the claim follows by:
	\begin{align}
		&\thyIn{\namedax{ax}{\termF}}{\thy}&&\text{\byAss}\label{axiomC1}\\
		\NDLineTG{\termF}{\ruleRef{axiom},(\ref{axiomC1}),(\ref{axiomP1.3})}\nonumber
	\end{align}
	If the axiom is a typing axiom then its preimage states that some constant $\c$ of type $\A$ satisfies $\c\termEquals{\A}\tm$ which follows by rule \ruleRef{refl}.
	
	If the axiom is the PER axiom generated for some $\A$ type declared in $\thy$, then it's quasi-preimage states that equality on $\A$ implies itself which is obviously true. 
	
	\subparagraph{\ruleRef{implI}:}
	Since the conclusion is proper, it follows that the preimage 
	\[
	\ctx\dedT \termF\impl \termG
	\] of the normalization \[
	\ctxT\dedPT \quasiImage{\termF}\impl \quasiImage{\termG}\] of the conclusion is well-formed and thus $\ctx\dedT \termF\impl \termG\ofT\bool$.
	\begin{align}
		\NDLineTG{\termF\impl \termG\ofT\bool}{see above}\label{implIP1.1}\\
		\NDLineTG{\termF\ofT\bool}{\ruleRef{implTypingL},(\ref{implIP1.1})}\label{implIP1.2}\\
		\NDLineTG{\termG\ofT\bool}{\ruleRef{implTypingR},(\ref{implIP1.1})}\label{implIP1.3}\\
		\NDLineT{\concatCtx{\ctx}{\namedass{ass}{\termF}}}{\termG\ofT\bool}{\ruleRef{assDed},(\ref{implIP1.3})}\nonumber
	\end{align}
	Obviously $\concatCtx{\ctx}{\namedass{ass}{\termF}}\dedT \termG$ is a quasi-preimage of the validity assumption of the rule, so the first part of the claim is proven.
	
	Regarding the second part:
	\begin{align}
		\NDLineT{\concatCtx{\ctx}{\namedass{ass}{\termF}}}{\termG}{\byAss}\label{implIC1}\\
		\NDLineTG{\termF\impl \termG}{\ruleRef{implI},(\ref{implIP1.2}),(\ref{implIC1})}\nonumber
	\end{align}
	
	\subparagraph{\ruleRef{implE}:}
	Since the conclusion is proper, it follows that the preimage 
	\[
	\ctx\dedT \termG
	\] of the normalization \[
	\ctxT\dedPT \quasiImage{\termG}\] of the conclusion is well-formed and thus $\ctx\dedT \termC{\termG}\ofT\bool$.
	
	Since the validity assumptions use the same context as the conclusion, it follows that they are both proper and uniquely determined.
	
	Since the formula $\quasiImage{\termF}$ (where $\ctxT\dedPT \quasiImage{\termF}$ is the second validity assumption) must be almost proper, it follows that its preimage $\termF$ is well-typed i.e. $\ctx\dedT \termF\ofT\bool$.
	\begin{align}
		\NDLineTG{\termF\ofT\bool}{$\quasiImage{\termF}$ almost proper}\label{implEP1.1}\\
		\NDLineTG{\termG\ofT\bool}{see above}\label{implEP1.2}\\
		\NDLineTG{\termF\impl \termG\ofT\bool}{\ruleRef{implType'},(\ref{implEP1.1}),(\ref{implEP1.2})}\nonumber
	\end{align}
	Clearly, $\ctx\dedT \termF\impl \termG$ and $\ctx\dedT \termF$ are quasi-preimages of the two validity assumptions of the rule, so we have proven the first part of the claim.
	
	Regarding the second part:
	\begin{align}
		\NDLineTG{\termF\impl \termG}{\byAss}\label{implEC1}\\
		\NDLineTG{\termF}{\byAss}\label{implEC2}\\
		\NDLineTG{\termG}{\ruleRef{implE},(\ref{implEC1}),(\ref{implEC2})}\nonumber
	\end{align}
	
	\subparagraph{\ruleRef{boolExt}:}
	Since the conclusion is proper, it follows that the preimage 
	\[
	\ctx\dedT \univQuant{x}{\bool}\p\ \x
	\] of the normalization \[
	\ctxT\dedPT \univQuant{x}{\bool}\univQuant{y}{\bool}\termEqT{\bool}{\x}{\y}\impl \termEqT{\bool}{\left(\lambdaFun{x}{\bool}\T\right)\ \x}{\quasiImage{p}\ \y}\] of the conclusion is well-formed and thus $\ctx\dedT \univQuant{x}{\bool}\p\ \x\ofT\bool$.
	Expanding the definition of $\forall$ yields:
	\begin{align}
		\NDLineTG{\lambdaFun{x}{\bool}\T\termEquals{\piType{x}{\bool}\bool}\lambdaFun{x}{\bool}\p\ \x\ofT\bool\QNegSp}{\Quad see above}\label{boolExtP1.1}\\
		\NDLineTG{\lambdaFun{x}{\bool}\p\ \x\ofT\piType{x}{\bool}\bool}{\ruleRef{eqTyping},\ruleRef{sym},(\ref{boolExtP1.1})}\label{boolExtP1.2}\\
		\NDLineTG{\left(\lambdaFun{x}{\bool}\p\ \x\right)\ \T\ofT\bool}{\ruleRef{appl},(\ref{boolExtP1.2})}\label{boolExtP1.3}\\
		\NDLineTG{\left(\lambdaFun{x}{\bool}\p\ \x\right)\ \termF\ofT\bool}{\ruleRef{appl},(\ref{boolExtP1.2})}\label{boolExtP1.4}\\
		\NDLineTG{\p\ \T\termEqB \left(\lambdaFun{x}{\bool}\p\ \x\right)\ \T}{\ruleRef{sym},\ruleRef{beta},(\ref{boolExtP1.3})}\label{boolExtP1.5}\\
		\NDLineTG{\p\ \F\termEqB \left(\lambdaFun{x}{\bool}\p\ \x\right)\ \F}{\ruleRef{sym},\ruleRef{beta},(\ref{boolExtP1.4})}\label{boolExtP1.6}\\
		\NDLineTG{\p\ \T\ofT\bool}{\ruleRef{eqTyping},(\ref{boolExtP1.5})}\nonumber\\
		\NDLineTG{\p\ \termF\ofT\bool}{\ruleRef{eqTyping},(\ref{boolExtP1.6})}\nonumber
	\end{align}
	Since $\ctx\dedT \p\ \T$ and $\ctx\dedT \p\ \F$ are clearly quasi-preimages of the two validity assumptions of the rule, we have proven the first part of the claim.
	
	Regarding the second part:
	\begin{align}
		\NDLineTG{\p\ \T}{\byAss}\label{boolExtC1}\\
		\NDLineTG{\p\ \F}{\byAss}\label{boolExtC2}\\
		\NDLineTG{\univQuant{x}{\bool}\p\ \x}{\ruleRef{boolExt},(\ref{boolExtC1}),(\ref{boolExtC2})}
	\end{align}

	\subparagraph{\ruleRef{nonempty}:}
	Since the conclusion is well-formed, it follows that the preimage $\ctx\dedT \termF$ of the normalization \[\ctxT\dedPT \quasiImage{\termF}\] of the conclusion is well-formed. This also implies that $\concatCtx{\ctx}{\x:\A}\dedT\termF:\bool$.
	Regarding the second part:
	By choice of the type indices (last property in Lemma~\ref{lem:indexing}) we know that the type \A is non-empty in DHOL (this choice is possible since $\thy$ is partially inhabited), so there is some term \tm with $\ctx\dedT \tm:\A$. 
	\begin{align}
		\NDLineT{\concatCtx{\ctx}{\x:\A}}{\termF}{\byAss}\label{nonempty1}\\
		\NDLineTG{\tm:\A}{\byAss}\label{nonempty2}\\
		\NDLineTG{\univQuant{x}{\A}\termF}{\ruleRef{forallI},(\ref{nonempty1})}\label{nonempty3}\\
		\NDLineTG{\termF}{\ruleRef{forallE},(\ref{nonempty3}),(\ref{nonempty2})}\nonumber
	\end{align}
\end{proof}
\end{appendix}
\end{document}